\documentclass[preprint,12pt]{elsarticle}

\usepackage[margin=1in]{geometry}
\usepackage{subcaption}
\usepackage{amsmath}
\usepackage{amssymb}
\usepackage{amsthm}
\usepackage{graphicx}
\usepackage{xcolor}
\usepackage{soul}
\sethlcolor{yellow}
\renewcommand{\hl}[1]{#1}
\usepackage{booktabs}
\usepackage{array}
\usepackage{float}
\usepackage{placeins}
\usepackage{needspace}
\usepackage[T1]{fontenc}
\usepackage{lmodern}
\usepackage[hidelinks]{hyperref}
\usepackage{xurl}

\allowdisplaybreaks[4]
\newcommand{\NN}{\mathcal{N}}
\newcommand{\LL}{\mathcal{L}}
\newcommand{\card}[1]{\lvert #1\rvert}
\newcommand{\dbar}{\bar{d}}
\newcommand{\Fbar}{\bar{F}}
\newcommand{\Cplan}{\mathcal{C}^{\mathrm{plan}}}
\newcommand{\Cccm}{\mathcal{C}^{\mathrm{CCM}}}
\newtheorem{lemma}{Lemma}
\newtheorem{proposition}{Proposition}
\newtheorem{corollary}{Corollary}
\newtheorem{assumption}{Assumption}
\begin{document}
\sloppy
\hypersetup{pageanchor=false}

\begin{frontmatter}

\title{The Potential Welfare Gains from Curtailment Trading Under Non-Firm Interconnection}

\author{Richard Mahuze}
\author{Charlotte Gressel}
\author{Ali Amadeh}
\author{K. Max Zhang\corref{cor1}}
\address{Sibley School of Mechanical and Aerospace Engineering, Cornell University, Ithaca, NY 14853, USA}
\cortext[cor1]{Corresponding author: K. Max Zhang (kz33@cornell.edu).}

\begin{abstract}
Rapid growth of data centers and other large loads is driving non-firm interconnection agreements that exchange faster grid access for curtailment exposure. Efficient curtailment should reflect the value of lost load (VOLL), which varies by more than 100× across customer classes. Pro-rata allocation ignores this variation. We present the network-constrained Curtailment Credit Market (CCM), which lets loads pay one another to transfer curtailment obligations subject to transmission limits. We prove that non-negative credit flows can implement every planner-feasible final curtailment allocation under the direct-current power-transfer-distribution-factor model. Under truthful bids, CCM clearing therefore attains the same total value of served load as an omniscient planner. Vickrey--Clarke--Groves (VCG) payments make truthful bidding a dominant strategy and leave every participant at least as well off as its no-trade outcome under the stated assumptions. Across three test networks, from a 3-bus system to a reduced New York grid, truthful CCM clearing matches planner welfare to numerical tolerance and produces welfare 1.41--1.83 times the pro-rata level. These fixed cases fall between the 92nd and 100th percentiles of comparable randomized cases, whose median welfare ratios are 1.16--1.26. Across the 24 combinations of three networks, four initial-obligation rules, and two customer-value constructions, trading removes at least 10\% of interruption cost in 86.1\% or more of the randomized cases. Settlement tests show that VCG may require subsidies, uniform pricing is manipulable, and locational credit pricing is revenue adequate in the tested formulation yet remains strategically vulnerable, which warrants further investigation of manipulation safeguards.
\end{abstract}

\begin{graphicalabstract}
\includegraphics[width=\textwidth]{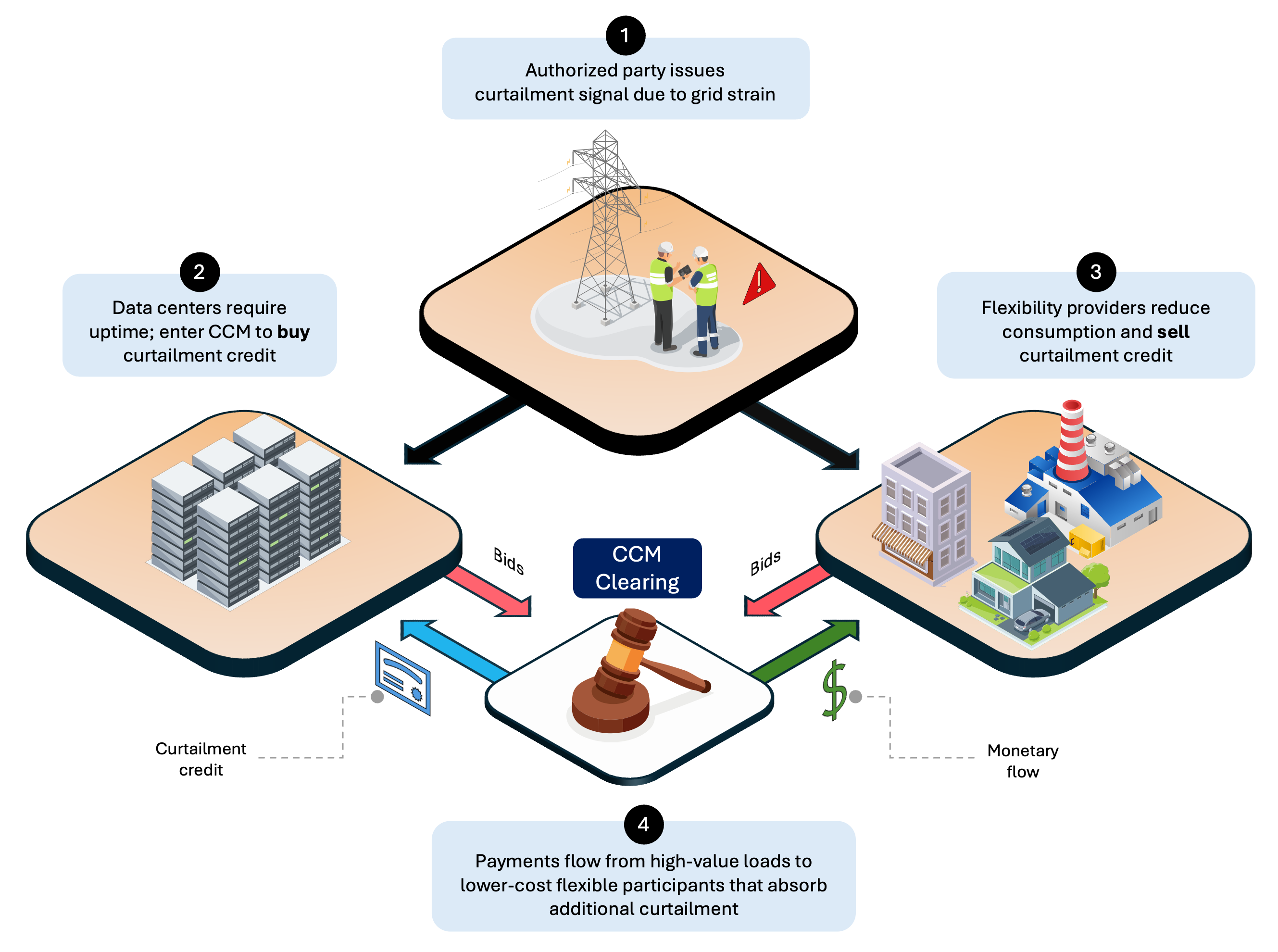}
\end{graphicalabstract}

\begin{highlights}
\item The CCM clears bilateral curtailment trades subject to transmission limits
\item Bilateral credit flows reach every planner-feasible curtailment allocation
\item Truthful CCM clearing matches planner welfare across all three test networks
\item Welfare gains persist across networks, initial rules, and value realizations
\item VCG payments make truthful bidding a dominant strategy
\item[]\vspace{1.5\baselineskip}\small
\textcopyright{} 2026. This manuscript version is made available under the
CC BY-NC-ND 4.0 license:
\url{https://creativecommons.org/licenses/by-nc-nd/4.0/}.
The Version of Record is available at
\url{https://doi.org/10.1016/j.apenergy.2026.128861}.
\end{highlights}


\end{frontmatter}

\hypersetup{pageanchor=true}

\section{Introduction}

Demand from data centers and other large loads is rising faster than utilities can build the new transmission and generation to serve it. A growing number of these loads accept non-firm interconnection agreements (or are required by tariff to take them), which grant quicker access in exchange for mandatory curtailment whenever the grid cannot reliably serve them, whether from a generation shortfall or network congestion~\cite{kahrl_speed_2026,esig_large_loads_2026}. Every load that connects this way enlarges the pool the operator must shed first when capacity runs short. Loads value uninterrupted service very differently, and the value of lost load (VOLL) measures how much each one would lose. Because these values vary enormously across loads, they should guide decisions about which loads give way and by how much. A survey of the ERCOT region puts the system-wide average VOLL near \$35,000/MWh. Yet the values that loads place on uninterrupted service span a wide range, from \$4,000/MWh for residential customers to \$667,000/MWh for small commercial and industrial loads, a $167\times$ spread~\cite{charles_gibbons_value_2024}. Nationally, the same interruption-cost surveys find comparable heterogeneity across U.S.\ service territories~\cite{larsen_ice_2025}. Pro-rata allocation, which cuts every load by the same proportion, ignores this spread~\cite{mehrtash_reserve_2023} and destroys surplus that a price-based mechanism could recover.

This paper introduces the network-constrained Curtailment Credit Market (CCM). A curtailment obligation is a binding requirement for a load to reduce consumption by a specified amount. A curtailment credit transfers one unit of this obligation between participating loads. Buying a credit reduces a load's required curtailment by one unit, whereas selling it increases the seller's obligation by the same amount. A high-VOLL load can therefore pay a lower-VOLL load to accept more curtailment, while the market operator clears the trades subject to transmission limits. The CCM makes an administratively assigned curtailment obligation transferable while preserving network feasibility. The central question is whether such bilateral trades can reproduce the allocations that an ideal planner would choose if it knew each load's private curtailment value.

Markets allocate resources efficiently when participants can reveal their valuations through prices and trade freely~\cite{mas-colell_microeconomic_1995}. This motivates a similar requirement in curtailment allocation, where loads should be able to express their values for served load and exchange curtailment obligations. However, current blunt rules prevent both. When the system reaches firm load shedding, the operator assigns curtailment by fixed rules that give agents' private valuations no role~\cite{mehrtash_reserve_2023,monterde_non-firm_2025}. Even upstream of emergencies, scarcity pricing instruments can compress the signal. For example, the Electric Reliability Council of Texas (ERCOT) Operating Reserve Demand Curve calculates real-time reserve price adders using an administrative value of lost load parameter set equal to the system-wide offer cap, currently \$5,000/MWh, rather than using load-specific interruption values~\cite{ercot_2024_2024}. This single parameter was a reasonable design choice for a market where most loads are price-insensitive, but it does not distinguish among the range of interruption costs that heterogeneous large loads now present. The result is a system that compresses heterogeneity in the pricing instrument and then disregards it entirely in the allocation. A mechanism that lets agents reveal private curtailment costs and trade curtailment obligations subject to network feasibility constraints would recover the surplus that this process leaves unrealized.

If all loads participated in real-time wholesale markets and scarcity pricing fully reflected willingness to pay, locational marginal prices would internalize curtailment values, and those prices alone would sort curtailment efficiently across loads~\cite{schweppe_spot_1988}. Three features of current operations prevent this outcome. First, price caps truncate bids above a regulatory ceiling, so loads whose interruption costs exceed that ceiling cannot express their true willingness to pay. Second, most retail customers purchase electricity at regulated rates that do not vary with wholesale conditions, which insulates them from scarcity signals entirely. Third, when reserves fall below critical thresholds, operators override market prices with non-price emergency actions such as voltage reductions, public conservation appeals, and firm load shedding precisely when the heterogeneity in interruption costs is most consequential~\cite{mehrtash_reserve_2023}. Where loads do participate in wholesale markets, they enter through operator-dispatched demand response programs such as Pennsylvania-New Jersey-Maryland (PJM) Interconnection's Emergency Load Response Program, New York Independent System Operator's Special Case Resources, and ERCOT's Emergency Response Service~\cite{ferc_dr_assessment_2024}. These programs treat curtailment as a supply-side substitute that the operator activates on behalf of the system rather than as an obligation that loads trade among themselves.

Even where loads can bid into energy markets as demand response, the compensation rule itself creates a distortion. Federal Energy Regulatory Commission (FERC) Order 745 requires system operators to pay curtailment at the full locational marginal price (LMP)~\cite{federal_energy_regulatory_commission_demand_2011}. Chao and DePillis~\cite{chao_incentive_2013} show that this rule can induce baseline inflation and excessive curtailment whose value falls below the cost of forgone energy, with the excess cost recovered through socialized charges to other market participants or ratepayers. Because the traded object is an assigned curtailment obligation rather than an estimated counterfactual consumption baseline, the specific inflation mechanism that Chao and DePillis identify does not arise. Under a lateral settlement, the load seeking curtailment relief funds the flexible load that accepts the additional obligation. The payment rule determines whether buyer payments cover seller compensation or require an external subsidy. A complementary curtailment market can therefore improve the residual allocation without assigning the transfer automatically to the operator or system-wide ratepayers.

Which mechanisms already allocate scarcity, and how far do they go toward letting loads price their own curtailment? Several designs approach this question from different directions, distinguished by which decisions are fixed before the mechanism runs and which it determines. Table~\ref{tab:positioning} compares the closest prior art across several axes, including what is fixed before allocation or clearing, what participants supply, what the mechanism determines, and how network conditions enter.

\begin{table}[H]
\centering
\caption{Closest prior art relevant to curtailment allocation. ``Exogenous''
  means fixed before the allocation or clearing step; participant choices are
  listed separately because the mechanism receives them as inputs. The final
  row is the mechanism studied in this paper.}
\label{tab:positioning}
\fontsize{7.5}{8.5}\selectfont
\setlength{\tabcolsep}{1.5pt}
\renewcommand{\arraystretch}{0.94}
\begin{tabular}{@{}>{\raggedright\arraybackslash}p{1.95cm}
                    >{\raggedright\arraybackslash}p{2.55cm}
                    >{\raggedright\arraybackslash}p{2.55cm}
                    >{\raggedright\arraybackslash}p{2.15cm}
                    >{\raggedright\arraybackslash}p{3.15cm}
                    >{\raggedright\arraybackslash}p{3.05cm}@{}}
\toprule
Mechanism & What the mechanism is
& Fixed before allocation or clearing (exogenous)
& Participant-supplied input
& Determined by allocation or clearing (endogenous)
& Network treatment \\
\midrule

Priority service and reliability insurance~\cite{chao_priority_1987,chao_priority_2022,billimoria_insurance_2022}
& Ex-ante reliability differentiation
& Scarcity states, available capacity, and contract terms
& Reliability class or insurance coverage
& Priority and compensation; the insurance variant also sets premiums and
  procures reserves
& No event-time, network-feasible reallocation of assigned obligations \\
\addlinespace

Dispatched demand response~\cite{federal_energy_regulatory_commission_demand_2011,chao_incentive_2013}
& Operator procurement of voluntary curtailment
& System conditions, program rules, and baseline
& Reduction offer or energy-market bid
& Accepted reduction and compensation
& Network constraints may enter dispatch, not obligation reallocation \\
\addlinespace

UKPN Energy Exchange~\cite{ukpn_energy_exchange_2022}
& Trades between the operator and generators or among generators in
  curtailment positions or access rights
& Connection agreements, initial queue, and forecast constraint
& Desired position, curtailment level, or access right
& Matched change in generator curtailment exposure
& Constraint-specific sensitivity factors \\
\addlinespace

BiTraDER~\cite{electricity_north_west_bitrader_2022,afry_bitrader_2023}
& Peer-to-peer changes to a local curtailment order among demand and generation
  resources
& Local constraint, merit order, and eligibility
& Bids, offers, and baseline information
& Matched trades and revised order under the initial same-direction restriction
& One identified constraint; eligible participants initially equally effective \\
\addlinespace

\textbf{Curtailment credit market}
& \textbf{Bilateral reallocation of assigned load-curtailment obligations}
& \textbf{Total reduction, initial obligations, capability bounds, transfer
  factors, and line limits}
& \textbf{VOLL bids}
& \textbf{Final curtailment, bilateral transfers, and payments}
& \textbf{Every modeled transmission constraint enforced in clearing} \\

\bottomrule
\end{tabular}
\end{table}

Table~\ref{tab:positioning} shows that priority service pricing lets a supplier offer several reliability grades at different prices, allowing each customer to choose a preferred grade before scarcity occurs. Chao and Wilson~\cite{chao_priority_1987} showed that letting loads self-select into curtailment tiers that reflect their willingness to pay Pareto-dominates random rationing, meaning no consumer is worse off and some are better off, and they proposed a market for tradeable ``priority points'' as one implementation path. Their framework assumes a single aggregate supply node with no transmission network. Chao, Oren, and Wilson~\cite{chao_priority_2022} revisited priority pricing for systems with high renewable penetration under uncertainty, still without explicit network-feasibility enforcement. Deng and Oren~\cite{deng_priority_2001} extended the idea to zonal network access through a priority-differentiated transmission tariff, in which traders self-select a strike price that sets both their scheduling priority and their compensation if curtailed. Mou et al.~\cite{mou_bilevel_2020} formulate priority-service pricing as a bilevel program that they solve as a single-level mixed-integer linear program, and the design still sets a producer's menu of price-reliability options. Billimoria et al.~\cite{billimoria_insurance_2022} develop a reliability-insurance implementation of priority service in which consumers choose insurance coverage, the VOLL specified in each contract determines curtailment order, and the insurer procures strategic reserves to reduce unserved energy. Related interruptible-service tariffs differentiate reliability more coarsely at administratively specified prices, without eliciting granular willingness-to-pay information.

Dispatched demand response reaches the same boundary from the operational side because the operator calls the reduction and sets the compensation, so participants never bid their interruption costs against one another. Such programs begin from a system condition, program rules, and a measurement baseline, and they end when the operator accepts a reduction offer and pays for it. Broader surveys of demand-side flexibility~\cite{abedrabboh_applications_2023,cano-martinez_scoping_2025} find the same pattern in aggregator-dispatched virtual power plant programs. These programs overwhelmingly assume that the operator already knows each participant's interruption cost rather than eliciting it through bids. In other words, the operator procures a service while each participant retains its existing obligation.

Two distribution-system pilots let participants change their own curtailment exposure. The UK Power Networks Energy Exchange explored trades between the network operator and flexible distributed generators and among the generators themselves. These trades concerned queue positions, curtailment levels, or capacity rights, and the design accounts for location through sensitivity factors that record how far a generator at a given point relieves the constrained part of the network~\cite{ukpn_energy_exchange_2022}. BiTraDER specifically pilots peer-to-peer trading of curtailment obligations among distributed energy resources, covering both demand and generation~\cite{electricity_north_west_bitrader_2022}. Its rules swap positions in a local curtailment queue and exclude a load that already owes a reduction from transferring that obligation to another obligated load. They also assign every eligible participant an effectiveness factor of one, so trades within the eligible group are treated as electrically equivalent~\cite{afry_bitrader_2023}. The CCM instead chooses which loads bear the final obligations and enforces transmission limits on that allocation.

On the regulatory side, recent policies impose curtailment obligations on large loads and then allocate them administratively through fixed quantities or eligibility rules that take no account of loads' heterogeneous VOLL. Texas Senate Bill 6 directs ERCOT to impose curtailment on non-critical large loads once market services are exhausted~\cite{mcguirewoods2025sb6}, and Southwest Power Pool's Conditional High Impact Large Load Service offers curtailable transmission service for loads that lack firm capacity~\cite{spp_chills_2026}. These regulations establish curtailment obligations without providing a market that can redistribute them. A market that reallocates these obligations must also ensure that its trades respect network limits.

A separate literature develops network-feasibility rules for transmission usage rather than for demand-side obligations. Hogan~\cite{hogan_contract_1992} proposed financial transmission rights, which pay their holder the congestion charge between two points on the network and therefore hedge the cost of moving power across a congested grid. He showed that if the awarded rights satisfy a simultaneous feasibility test, congestion rents under nodal pricing are sufficient to fund all payments. Chao and Peck~\cite{chao_market_1996} designed a market mechanism for transmission access built on tradeable rights to use individual transmission lines, with a trading rule that specifies how each energy transaction loads each line. In practice, financial transmission-rights markets have drawn substantial criticism on grounds of revenue adequacy, strategic behavior, and investment incentives. Both designs presuppose the wholesale energy market and organize access to it, and the object their participants hold is a financial claim or a right to use a line. A curtailment obligation is a different object, and the clearing that reallocates it must elicit private curtailment costs that the system operator does not observe.

Eliciting curtailment costs requires estimating interruption values and inducing agents to report them truthfully. Survey evidence shows that interruption costs depend on outage duration, customer characteristics, backup capability, and advance warning and that the estimates carry statistical uncertainty~\cite{charles_gibbons_value_2024,larsen_ice_2025}. The present model treats each agent's marginal VOLL as fixed within a clearing event, while event-specific or duration-dependent values remain an extension. By contrast, the generation side has already confronted the second problem. Xu and Low~\cite{xu_efficient_2017} showed that Vickrey--Clarke--Groves (VCG) payments achieve dominant-strategy incentive compatibility in wholesale energy markets, Sessa et al.~\cite{sessa_exploring_2017} extended VCG to reserve procurement and characterized conditions under which collusion is unprofitable, and Exizidis et al.~\cite{exizidis_incentive_2019} applied VCG to a two-stage stochastic market with high wind penetration. That dominant-strategy property is a classical mechanism-design result~\cite{nisan_algorithmic_2007}. These electricity-market applications use VCG as an incentive benchmark. We use VCG in the same role for the CCM.

To the best of our knowledge, no published formal mechanism jointly (1) elicits private curtailment values from parties that already hold assigned obligations, (2) allows one obligated party to transfer part of that obligation to another obligated party, and (3) enforces system-wide transmission feasibility on the final allocation. Existing mechanisms supply individual elements of this design. Network-constrained rights auctions enforce simultaneous feasibility~\cite{hogan_contract_1992,oneill_joint_2002}, distribution pilots trade curtailment positions or obligations~\cite{electricity_north_west_bitrader_2022,ukpn_energy_exchange_2022}, and local flexibility markets incorporate network constraints while procuring services~\cite{prat_design_2021, lee_regulated_2023}. A separate line of work asks how curtailment should be allocated across an interconnected system by rule rather than by market, and it complements the mechanism studied here~\cite{fraunholz_demand_2025}.

The CCM takes curtailment obligations fixed by upstream energy and reliability processes and clears agents' VOLL bids into bilateral credit flows subject to direct-current power-flow constraints. The central structural result concerns market reachability. A centralized planner assigns final curtailment levels directly, while the CCM must realize the same allocation through non-negative trades between specific pairs of agents. We prove that every planner-feasible allocation can be decomposed into bilateral credit flows that satisfy network limits and non-negativity because transmission constraints depend on each agent's final curtailment level rather than the trade path used to reach it. The bilateral credit-flow representation can therefore implement every allocation available to the planner, $\Cccm=\Cplan$. Under truthful bidding, the CCM selects the same curtailment allocation and attains the same total value of served load, $W=\sum_i v_i(L_i-c_i)$, as an omniscient planner. VCG payments make truthful bidding a dominant strategy and provide an incentive-compatibility benchmark for practical payment rules. The welfare equivalence concerns the allocation, while financing depends on the payment rule. VCG settlement may require an external subsidy when the mechanism's net revenue is negative.

The core formulation applies the same clearing model to each event, while the numerical study examines multi-hour curtailment durations and consecutive daily profiles. The network representation uses the direct-current (DC) power-flow approximation and is evaluated on three test networks. These comprise a 3-bus network for transparent inspection of credit flows, an IEEE 24-bus network for topological realism, and a reduced New York grid with five zonal-proxy agents for inter-area aggregation. In these experiments, the CCM produces social welfare $1.41$ to $1.83$ times the pro-rata level. A Monte Carlo study over randomized instances, reported in the robustness section, evaluates how representative the three test-network results are within the tested design. Across the 24 combinations formed by three networks, four rules for assigning initial curtailment, and two VOLL patterns, trading reduces interruption costs by at least 10 percent in 86.1 to 100 percent of the randomized cases, depending on the combination. Under pro-rata initial curtailment, median welfare ratios across the randomized cases range from $1.16\times$ to $1.26\times$ across the three networks and two VOLL patterns. Welfare ratios of $1.41\times$ to $1.83\times$ from the three test networks fall between the 92nd and 100th percentiles among randomized cases with the same customer-value pattern, so they represent high-gain realizations within the tested design. The specific contributions are as follows.

\begin{enumerate}
\item \textbf{Feasible-set equivalence.} Trading introduces no loss of allocative capability. The set of curtailment allocations reachable through bilateral credit flows coincides exactly with the central planner's feasible set under the direct-current power-transfer-distribution-factor network model, so a market can replace administrative rationing without giving up any outcome the planner could reach.

\item \textbf{Planner-welfare equivalence and payment-rule properties.} Under truthful bidding, the CCM attains the planner's benchmark welfare. Among settlement rules, VCG makes truthful bidding a dominant strategy and, under the stated feasibility and leave-one-out assumptions, guarantees that every agent ends at least as well off as it would without trading, on any network and under any initial-obligation rule, though it may require an external subsidy. Uniform-price settlement can violate that participation guarantee under congestion. Locational credit pricing is revenue adequate in the studied formulation and fits nodal settlement practice. We identify it as a conditional candidate for further evaluation.

\item \textbf{\hl{Exact MILP reformulation.}} \hl{For any submitted bids, operational CCM clearing is a linear program. Under valid Big-M bounds, we show that a mixed-integer linear program (MILP) reproduces the same allocation, using binary variables to represent the operator's optimality conditions. The offline bid-deviation analysis supplies one candidate bid at a time, and the MILP returns the market clearing for that bid. At 200 participants, operational clearing takes 41~ms of solver time, whereas one MILP evaluation takes 18.6~s.}
\end{enumerate}

\hl{Section}~\ref{sec:prelim} \hl{formalizes the model, establishes its key properties, and derives the fixed-bid MILP reformulation. Section}~\ref{sec:numerical} \hl{introduces the test networks and benchmark regimes. Section}~\ref{sec:results} \hl{presents and interprets the numerical evidence. Section}~\ref{sec:conclusion} \hl{concludes.}

\section{Market Setup and Properties}\label{sec:prelim}\label{sec:properties}

\subsection{Market Boundary and Terminology}

The CCM operates after energy and reliability processes have fixed the
curtailment target~$D$~(MW) and assigned an initial
obligation~$\dbar_i$ to each agent~$i$. These processes may commit generators,
dispatch them subject to network limits, procure reserves, calculate energy
locational marginal prices (LMPs), or invoke tariff and emergency rules. The
CCM is a layer on top of these processes. It treats $D$, $\dbar_i$, and an
operator-approved pre-curtailment load reference~$L_i$ as fixed inputs and
reallocates the obligations among participating loads. The upstream process
selects the measurement rule used to establish $L_i$ before CCM clearing
begins.

Throughout this paper, a \emph{curtailment obligation}~$\dbar_i$ is the
amount that agent~$i$ is initially required to shed. A \emph{curtailment
credit}~$x_{ij}$ is the cleared amount of obligation transferred from
buyer~$j$ to seller~$i$ within the same event. The associated
\emph{relief} is the reduction in the buyer's curtailment obligation.
After clearing, a buyer of relief curtails less than its initial obligation
and pays, whereas a seller of relief accepts additional curtailment and is
paid. Let $c_i$ denote agent~$i$'s final cleared curtailment. The net position
$q_i=\dbar_i-c_i$ records the transaction direction. A positive $q_i$ denotes
a buyer, and a negative $q_i$ denotes a seller. A settlement rule converts $q_i$ into a monetary
transfer~$p_i$, where $p_i>0$ means that agent~$i$ pays the mechanism. We
evaluate Vickrey--Clarke--Groves (VCG), uniform-price, and locational
settlement rules. These rules apply different transfers to the same cleared
curtailment allocation.

\phantomsection\label{sec:lifecycle}
Operationally, the initial obligations provide the starting allocation.
Participants state how costly interruption would be. Subject to physical and
network limits, the market moves part of the obligation from buyers, who value
continued service more, to sellers willing to curtail at lower cost. After the
event, the operator verifies compliance with each final service limit and
records a debit for each buyer and a credit for each seller. A curtailment
credit applies only to that event.

Consider a one-hour event requiring 10~MW of curtailment from two loads, each
initially assigned 5~MW. The first load values continued service more and buys
3~MW of relief from the second. The first load then curtails 2~MW and the second
curtails 8~MW. If the event uses a single cleared price of \$1,000/MWh, the
buyer pays \$3,000 and the seller receives \$3,000. Trading changes the
distribution of curtailment while preserving the 10~MW system requirement.
The following formulation extends this transaction to a network.

\subsection{Formal Market Setup and Properties}

At CCM opening, the bus-line topology, power transfer distribution factor
(PTDF) matrix~$\Phi$, remaining directional line-flow margins~$\Fbar_\ell^+$ and
$\Fbar_\ell^-$, agent locations~$n(i)$, operator-approved pre-curtailment load
references~$L_i$, and obligations~$\dbar_i$ are fixed inputs. These margins
measure the incremental-flow headroom at the operating point set by the
upstream process rather than the lines' absolute thermal ratings. For an
admissible CCM event, the initial obligations satisfy
$\sum_{i \in \NN} \dbar_i = D$ and $0 \le \dbar_i \le L_i$ for every agent~$i$.
The upstream operating point also satisfies the directional line limits, so
$\Fbar_\ell^{+}\ge 0$ and $\Fbar_\ell^{-}\ge 0$ for every line~$\ell$. Each entry $\Phi_{\ell,n}$ gives the
fraction of a one-MW injection at bus~$n$ that flows on line~$\ell$. We write
$\Phi_{\ell,n(i)}$ for the coefficient at agent~$i$'s bus. The matrix
translates bus-level curtailment changes into line-flow changes. The clearing
problem optimizes the reallocation of curtailment through credit trades. A
large load participates because it can pay a lower-cost agent to absorb part
of its assigned obligation, reducing its own interruption losses.

\subsubsection{CCM Clearing Formulation}\label{sec:ccm-formulation}

The following formulation makes these transaction directions precise on an
arbitrary network.
Let $\mathcal{N} = \{1, \dots, N\}$ denote the set of agents and $\mathcal{L}$ the set of
transmission branches. A bilateral credit flow $x_{ij} \ge 0$ moves curtailment
obligation from agent~$j$ to agent~$i$, where agent~$j$ purchases
$x_{ij}$~MW of curtailment relief, and agent~$i$ accepts
$x_{ij}$~MW of additional curtailment.  By convention, $x_{ii} = 0$ for all~$i$. All summations written $\sum_{j \ne i}$ range over $\NN \setminus \{i\}$, and all summations written $\sum_j$ range over the full set $\NN$ unless otherwise specified.

Credit flows determine each agent's realized curtailment and served
load.  Agent~$i$'s realized curtailment is
\begin{equation}\label{eq:curtailment}
  c_i \;=\; \dbar_i
    - \sum_{j \ne i} x_{ji}
    + \sum_{j \ne i} x_{ij},
\end{equation}
and its served load is $s_i = L_i - c_i$.  The net injection change
at bus~$n$ is
$\Delta P_n = \sum_{i:\,n(i)=n}(c_i - \dbar_i)$.

Given declared bids $b = (b_1, \dots, b_N)$, the CCM operator
solves the following linear program.
\begin{subequations}\label{eq:ccm}
\begin{alignat}{3}
  &\max_{x,\,c,\,s}
    &&\quad \sum_{i \in \NN} b_i \, s_i
    \label{eq:ccm:obj}\\[3pt]
  &\text{s.t.}
    &&\quad c_i \;=\; \dbar_i - \!\!\sum_{j \ne i} x_{ji}
        + \!\!\sum_{j \ne i} x_{ij}
    &&\quad \forall\, i \in \NN,
    \label{eq:ccm:curtbal}\\
  & &&\quad s_i \;=\; L_i - c_i
    &&\quad \forall\, i \in \NN,
    \label{eq:ccm:served}\\
  & &&\quad \sum_{i \in \NN} c_i \;=\; D,
    \label{eq:ccm:totbal}\\
  & &&\quad 0 \;\le\; c_i \;\le\; L_i
    &&\quad \forall\, i \in \NN,
    \label{eq:ccm:agbnd}\\
  & &&\quad -\Fbar_\ell^{-} \;\le\;
        \sum_{n} \Phi_{\ell n}\,\Delta P_n\bigl(c(x)\bigr)
        \;\le\; \Fbar_\ell^{+}
    &&\quad \forall\, \ell \in \LL,
    \label{eq:ccm:ptdf}\\
  & &&\quad x_{ij} \;\ge\; 0
    &&\quad \forall\, i \ne j.
    \label{eq:ccm:nonneg}
\end{alignat}
\end{subequations}
The operator maximizes bid-weighted served load subject to network
feasibility within those remaining margins, using credit flows as the
allocation instrument. The upstream process must compute the margins from the
post-obligation operating state before each CCM event; the CCM does not
reconstruct base dispatch or total line loading. The upper bound $c_i = L_i$ permits full disconnection. Equivalently, each agent
can absorb up to $L_i - \bar{d}_i$~MW of additional curtailment beyond its initial
obligation. This assumes that every agent is fully curtailable with a marginal
interruption cost~$v_i$ that does not vary with curtailment depth,
and that sellers of relief have the operational authority to deliver the
additional curtailment obligation they accept from buyers.

\subsubsection{Planner Benchmark and Feasible Sets}\label{sec:feasible-sets}

To evaluate the CCM, we compare it with an omniscient planner who
knows each agent's true value of lost load,
$v_i > 0$~(\$/MWh), and assigns curtailment directly without
credit flows. The planner solves
\begin{subequations}\label{eq:planner}
\begin{align}
W_{\mathrm{plan}} \;=&\; \max_{c}\; \sum_{i \in \NN} v_i\,(L_i - c_i)
\label{eq:planner:obj}\\
\text{s.t.}\; &\; \sum_{i} c_i = D,
\label{eq:planner:bal}\\
&\; 0 \le c_i \le L_i, \quad \forall\, i \in \NN,
\label{eq:planner:bnd}\\
&\; -\Fbar_\ell^{-} \;\le\;
\sum_{n} \Phi_{\ell n}\,\Delta P_n
\;\le\; \Fbar_\ell^{+}
\quad \forall\,\ell \in \LL.
\label{eq:planner:flow}
\end{align}
\end{subequations}

The two optimization problems operate over different but related
constraint sets. The \emph{planner feasible set} collects all
curtailment vectors that satisfy system balance, agent bounds, and
PTDF flow limits.
\begin{equation}\label{eq:Cplan}
\begin{aligned}
\Cplan = \Bigl\{c \in \mathbb{R}^N :\;&
\sum_{i} c_i = D,\;
0 \le c_i \le L_i\;\forall\, i,\\
&-\Fbar_\ell^{-} \le \textstyle\sum_n \Phi_{\ell n}\,\Delta P_n(c)
\le \Fbar_\ell^{+}\;\forall\, \ell
\Bigr\}.
\end{aligned}
\end{equation}
The \emph{CCM feasible set} is the projection of the full CCM
constraint set onto the curtailment vector~$c$.
\begin{equation}\label{eq:Cccm}
\begin{aligned}
\Cccm = \Bigl\{c \in \mathbb{R}^N :\\
\exists\, x \ge 0\;\text{s.t.}\;\\
c_i = \dbar_i - \sum_{j} x_{ji} + \sum_{j} x_{ij},\; \forall\, i,\\
\sum_i c_i = D,\; 0 \le c_i \le L_i\;\forall\, i,\\
-\Fbar_\ell^{-} \le \textstyle\sum_n \Phi_{\ell n}\,\Delta P_n(c)
\le \Fbar_\ell^{+}\;\forall\, \ell
\Bigr\}.
\end{aligned}
\end{equation}
Both sets impose the same constraints on~$c$, namely system balance,
agent bounds, and PTDF flow limits. However, $\Cccm$ additionally
requires the existence of non-negative credit flows that
produce~$c$ via~\eqref{eq:ccm:curtbal}; whether this requirement
restricts the achievable allocations is taken up in
Section~\ref{sec:equivalence}.

The following analysis establishes three results. First, the CCM clearing
LP is well-posed and admits a compact reduced form
(Section~\ref{sec:wellposed}).  Second, the CCM and planner
feasible sets coincide, so the bilateral credit flow representation introduces no loss of allocative capability (Section~\ref{sec:equivalence}).
\hl{Third, for any fixed bid vector, the CCM clearing problem can be written
as a mixed-integer linear program (MILP) without changing its optimal
allocation, provided the Big-M bounds are valid. This form lets the offline
incentive analysis test how the clearing outcome changes when an agent changes
its bid} (Section~\ref{sec:milp}).

\subsubsection{Well-Posedness and Reduced Representation}%
\label{sec:wellposed}

\begin{lemma}[Well-posedness of CCM clearing]\label{lem:wellposed}
For any CCM instance satisfying
$\sum_{i \in \NN} \dbar_i = D$,
$0 \le \dbar_i \le L_i$ for all~$i \in \NN$,
and $\Fbar_\ell^{+},\Fbar_\ell^{-} \ge 0$ for all~$\ell \in \LL$,
and for any bid
vector $b \ge 0$, the CCM clearing
problem~\eqref{eq:ccm} is feasible and has finite optimal value. Consequently, it admits an
optimal solution.
\end{lemma}

\begin{proof}[Proof sketch]
The no-trade allocation $(x,c,s)=(0,\bar d,L-\bar d)$ is
feasible, and $\sum_i b_i s_i \le \sum_i b_i L_i$ bounds the objective from
above. The feasible set can be unbounded in the accounting variables $x$ because
zero-net credit cycles leave $c$, $s$, and all PTDF flows unchanged. This
unboundedness does not affect existence of an optimum because the feasible set is a
nonempty closed polyhedron and the LP has finite optimal value. The full
verification is in \ref{app:proofs}.
\end{proof}

\begin{lemma}[Reduced LP]\label{lem:reduced}
Define $c_i(x) = \dbar_i - \sum_{j \ne i} x_{ji}
+ \sum_{j \ne i} x_{ij}$ and
$s_i(x) = L_i - c_i(x)$ for each $i \in \NN$.
Substituting into~\eqref{eq:ccm}, the CCM clearing problem is
equivalent to the reduced LP in the single variable
family $x = (x_{ij})_{i \ne j}$.
\begin{subequations}\label{eq:reduced}
\begin{alignat}{3}
  &\max_{x}\;
    &&\sum_{i \in \NN} b_i \bigl(L_i - c_i(x)\bigr)
    \label{eq:reduced:obj}\\
  &\text{s.t.}
    &&\quad c_i(x) \;\ge\; 0
    &&\quad \forall\, i \in \NN,
    \label{eq:reduced:lb}
    \tag{R1}\\
  & &&\quad c_i(x) \;\le\; L_i
    &&\quad \forall\, i \in \NN,
    \label{eq:reduced:ub}
    \tag{R2}\\
  & &&\quad -\Fbar_\ell^{-} \;\le\;
        \sum_{n} \Phi_{\ell n}\,\Delta P_n
        \;\le\; \Fbar_\ell^{+}
    &&\quad \forall\,\ell \in \LL,
    \label{eq:reduced:ptdf}
    \tag{R3}\\
  & &&\quad x_{ij} \;\ge\; 0
    &&\quad \forall\, i \ne j.
    \label{eq:reduced:nonneg}
    \tag{R4}
\end{alignat}
\end{subequations}
The balance constraint~\eqref{eq:ccm:totbal} is automatically
satisfied because each credit flow $x_{ij}$ adds to $c_i$ and subtracts
from $c_j$, so
$\sum_i c_i(x) = \sum_i \dbar_i = D$ identically.
\end{lemma}

\begin{proof}[Proof sketch]
Substituting the affine definitions of $c$ and $s$ into
the objective and inequality constraints
yields~\eqref{eq:reduced}.  System balance holds identically
because each credit flow $x_{ij}$ adds to $c_i$ and subtracts
from $c_j$, so
$\sum_i c_i(x) = \sum_i \dbar_i = D$.
See \ref{app:proofs} for the full derivation.
\end{proof}

With the reduced LP established, we turn to the relationship
between the CCM and planner feasible sets.

\subsubsection{Feasible-Set Equivalence and Welfare Recovery}%
\label{sec:equivalence}

The central structural question is whether the credit-flow
requirement restricts the set of achievable curtailment
allocations.  Both $\Cplan$ and $\Cccm$ impose the same
constraints on~$c$ (system balance, agent bounds, and PTDF flow
limits), but $\Cccm$ additionally requires the existence of
non-negative credit flows that produce~$c$
via~\eqref{eq:ccm:curtbal}.  This existential requirement could,
a priori, make $\Cccm$ strictly smaller than~$\Cplan$.
Proposition~\ref{prop:feasible} shows that it does not.

\begin{proposition}[Feasible-set equivalence under DC power flow]%
\label{prop:feasible}
Under the DC power flow approximation, $\Cccm = \Cplan$.  A curtailment vector~$c$ is achievable
through credit flows satisfying
\eqref{eq:ccm:curtbal}--\eqref{eq:ccm:nonneg}
if and only if it lies in the planner feasible set.
\end{proposition}

\begin{proof}
$(\Cccm \subseteq \Cplan)$.
If $c \in \Cccm$, then $c$ satisfies $\sum_i c_i = D$,
$0 \le c_i \le L_i$, and the PTDF bounds by definition, so
$c \in \Cplan$.

$(\Cplan \subseteq \Cccm)$.
Let $c^* \in \Cplan$.  Define
$\delta_i = \dbar_i - c_i^*$ for each agent~$i$.  Since
$\sum_i \dbar_i = \sum_i c_i^* = D$, the residuals satisfy
$\sum_i \delta_i = 0$.  Partition the agents into sellers
$S = \{i : \delta_i < 0\}$, who sell relief and therefore accept additional
curtailment above~$\dbar_i$, and buyers $B = \{j : \delta_j > 0\}$, who purchase
relief and reduce curtailment below~$\dbar_j$.  The zero-sum condition guarantees
$\sum_{i \in S} |\delta_i| = \sum_{j \in B} \delta_j$.

By the balanced transportation theorem~\cite{dantzig_linear_1998},
there exist flows $x_{ij} \ge 0$ for
$(i,j) \in S \times B$ satisfying
$\sum_{j \in B} x_{ij} = |\delta_i|$ for all $i \in S$ and
$\sum_{i \in S} x_{ij} = \delta_j$ for all $j \in B$.  Set
all other flows to zero.  These flows produce the target
allocation. For every agent~$i$,
$\sum_{j \ne i} x_{ji} - \sum_{j \ne i} x_{ij} = \delta_i$ by construction, so
$c_i(x) = \dbar_i - \delta_i = c_i^*$.

It remains to verify the PTDF constraints.  The net injection
change $\Delta P_n(c) = \sum_{k:\,n(k)=n}(c_k - \dbar_k)$
depends on the curtailment vector~$c$ alone, not on which
bilateral flows~$x_{ij}$ produced it.  Because
$c^* \in \Cplan$ satisfies the PTDF bounds, the constructed
point $(x, c^*)$ satisfies them as well.  Hence
$c^* \in \Cccm$.
\end{proof}

The reverse inclusion carries the substantive content, which is that any
planner-feasible allocation can be decomposed into bilateral
credit flows that respect both non-negativity and network limits.
The consequence is that the CCM, despite operating through credit trades rather than centralized dispatch,
introduces no loss of allocative capability.

\textit{Remark (Role of the DC power flow model).}
The DC/PTDF model is used here to
obtain a linear, convex clearing formulation and a tractable Karush--Kuhn--Tucker (KKT)/MILP
reformulation. The credit-flow decomposition itself relies on a different
property, where network feasibility is evaluated as a function of the final physical
curtailment allocation rather than of the contractual path by which
credits are exchanged. Under alternating-current (AC) power flow, physical line flows depend nonlinearly on
nodal injections, voltage magnitudes, voltage angles, and operating constraints,
so the planner feasible set would need to be defined over both curtailment and
network-state variables. Extending the computational formulation to AC optimal power flow (OPF) or
convex AC relaxations is therefore beyond the present DC/PTDF model, but the
bilateral accounting decomposition remains valid whenever both the planner and
the CCM impose the same feasibility projection over final curtailment
allocations.

\begin{corollary}[Planner-welfare equivalence under truthful
bidding]\label{cor:welfare}
If $b_i = v_i$ for all $i \in \NN$, then
$W_{\mathrm{CCM}} = W_{\mathrm{plan}}$.
\end{corollary}

\begin{proof}
Under truthful bidding, the CCM
objective~\eqref{eq:ccm:obj} becomes
$\sum_i v_i\,s_i = \sum_i v_i\,(L_i - c_i)$, which equals
the planner objective~\eqref{eq:planner}.
Proposition~\ref{prop:feasible} establishes
$\Cccm = \Cplan$, so the two problems maximize the same
objective over the same feasible set.  Their optimal values
therefore coincide.
\end{proof}

Corollary~\ref{cor:welfare} is a conditional allocative-efficiency statement, where \textit{if}
agents bid truthfully, \textit{then} the CCM's selected curtailment vector achieves the same total value of served load as the planner's. It does not imply that the payment rule is budget balanced.  The next result shows that VCG payments make truthful bidding dominant, so the allocative equivalence is achievable under strategic play.

\begin{proposition}[VCG dominant-strategy truthfulness]%
\label{prop:vcg}
Under the fixed, report-independent tie-breaking rule specified in
\ref{app:tiebreak}, VCG Clarke-pivot payments make truthful bidding
a dominant strategy for every agent on any network.
\end{proposition}

\begin{proof}[Proof sketch]
The feasible set $\Cccm$ depends only on exogenous network
parameters, not on bids, and each agent's valuation depends
only on its own allocation.  These two properties satisfy the
hypotheses of Nisan's
Theorem~9.17~\cite{nisan_algorithmic_2007}, from which
dominant-strategy truthfulness follows under a fixed,
report-independent tie-breaking rule.
The full argument and tie-breaking formalism are in \ref{app:proofs} and \ref{app:tiebreak}.
\end{proof}

Together, Proposition~\ref{prop:feasible},
Corollary~\ref{cor:welfare}, and
Proposition~\ref{prop:vcg} establish a chain. The CCM can reach
every planner-feasible allocation, it achieves planner-welfare equivalence
under truthful bidding, and VCG payments incentivize truthful
bidding as a dominant strategy. VCG is not proposed as the operational payment rule. It may run a budget deficit and requires solving $N+1$ optimization sub-problems, so an external subsidy may be needed to finance the transfer imbalance. This does not affect the welfare-equivalence result, because $W_{\mathrm{CCM}}=W_{\mathrm{plan}}$ concerns only the selected curtailment vector and is independent of whether the mechanism is self-financing. Accordingly, VCG serves here as an incentive-compatibility benchmark that bounds the design space for practical payment rules.

\textit{Remark (Valuation uncertainty).}
VCG makes honest reporting of an agent's own VOLL estimate optimal while
leaving estimation error unchanged. Interruption cost depends on event duration,
timing, production state, and backup availability. Estimation errors can change
the cleared allocation when they alter the value ordering or another marginal trade-off,
and bid-dependent payments can change even when the allocation does not.
Event-specific bids, duration-dependent bid curves, and sensitivity analysis
over plausible valuation errors are direct extensions of the scalar-VOLL model.

Concretely, VCG assigns each agent~$i$ the Clarke pivot payment
\begin{equation}\label{eq:vcg-payment}
  p_i(b) \;=\; \sum_{j \ne i} b_j\, s_j^{*(-i)}(b_{-i})
          \;-\; \sum_{j \ne i} b_j\, s_j^{*}(b),
\end{equation}
where $s^{*(-i)}$ is the optimal served-load vector when agent~$i$
is excluded from the market.  The payment is computed from reported
bids and equals the declared welfare loss that agent~$i$'s
participation imposes on all other agents. Under truthful bidding,
$b_j=v_j$, and this declared externality equals the true externality.
Agent~$i$'s total surplus is then
$S_i = v_i\,s_i^* - p_i$, which is the decomposition reported in
Tables~\ref{tab:peragent} and~\ref{tab:case24-peragent}.
Computing VCG payments requires $N+1$ solves of the clearing LP
per event. One solve with all $N$ agents obtains $s^*$, and
$N$ additional leave-one-out solves (one for each agent $i$) to
obtain $s^{*(-i)}$.

The leave-one-out problem uses a convention in which agent~$i$ is removed
entirely, the total requirement falls to $D-\dbar_i$, and the agent's obligation
is not redistributed. This convention matters because the reduced market's
no-trade point remains feasible and re-inserting agent~$i$ at its own obligation
changes no line flow.

\begin{proposition}[Individual rationality of VCG transfers]
\label{prop:ir}
Fix any network $(\Phi, \Fbar)$ with $\Fbar \ge 0$, loads $L_i \ge 0$,
initial obligations $0 \le \dbar_i \le L_i$ with $\sum_i \dbar_i = D$, and
true values $v_i \ge 0$. Under truthful bidding and the Clarke pivot
payments of Eq.~\eqref{eq:vcg-payment}, the clearing problem and every
leave-one-out problem are feasible with attained optima. Each agent's
surplus satisfies $S_i = W_{\mathrm{CCM}} - W^{(-i)}$ independently of the
optimizer the solver selects. Here, $W^{(-i)}$ is the leave-one-out
optimal value. Every agent does no worse than its no-trade baseline.
\begin{equation}\label{eq:ir}
  S_i \;=\; v_i\,s_i^{*} - p_i \;\ge\; v_i\,(L_i - \dbar_i)
  \qquad \forall\, i \in \NN .
\end{equation}
\end{proposition}

\begin{proof}[Proof sketch]
Take an optimal solution of the leave-one-out problem and re-insert agent
$i$ at $c_i = \dbar_i$ with no trades. Totals rebalance because the
obligation returns with the agent, and $i$'s injection deviation is zero,
so every flow expression is unchanged. The composite allocation is
therefore feasible for the full market, and its value,
$v_i(L_i - \dbar_i) + W^{(-i)}$, bounds $W_{\mathrm{CCM}}$ from below.
Combining with the identity $S_i = W_{\mathrm{CCM}} - W^{(-i)}$ yields
Eq.~\eqref{eq:ir}. The full proof, including attainment and
tie-breaking independence, is in \ref{app:proofs}.
\end{proof}

Within the stated assumptions, Proposition~\ref{prop:ir} permits equality and
applies to any initial-obligation rule because its proof requires only a
feasible reduced-market no-trade point and zero flow change when the removed
agent is re-inserted at its own obligation.

We next evaluate two alternatives to VCG, uniform-price and locational
settlement.

\medskip
\noindent\textit{Uniform-price settlement rule.}
Whether an agent benefits from participating depends on the settlement rule
applied to the cleared trades. Under uniform-price settlement at a price $\lambda$, an agent's
gain over its baseline is exactly
$(v_i - \lambda)(\dbar_i - c_i^{*})$. An agent is harmed only when it
trades against the price. In the uncongested case, optimal clearing prevents an
agent from trading against the uniform price, so uniform settlement is
individually rational. A binding network constraint can break this price
alignment and leave an optimally cleared agent below its no-trade baseline.
This possibility makes individual rationality, alongside budget balance and
incentive compatibility, a necessary criterion for evaluating practical
payment rules.

\medskip
\noindent\textit{Locational settlement rule.}
We evaluate a third settlement rule that derives bus-specific credit prices
from linear-program duals. The credit-flow program prices trades only up to
a common constant because Lemma~\ref{lem:reduced} shows that the
total-curtailment balance~\eqref{eq:ccm:totbal} is implied by the
participant accounting identities. Shifting the balance multiplier by any
constant, with offsetting shifts in the accounting multipliers, leaves the
allocation unchanged while moving every absolute bus price. We therefore
compute prices from the equivalent direct-curtailment clearing program,
which replaces $v_i$ with the declared bid $b_i$ in~\eqref{eq:planner} and
optimizes over the net curtailments~$c$ directly. In this program the
total-requirement constraint~\eqref{eq:planner:bal} is independent.
\hl{The independent constraint pins the price level that the credit-flow
program left free. A degenerate optimum can still admit several valid dual
price vectors. In those cases the numerical study uses the vector returned by
the solver, and a deployed tariff would state the selection rule.} The
numerical experiments compute all locational prices from this
program.
Let~$\mu_\ell^+,\mu_\ell^-\ge 0$ denote the multipliers
on the two directional flow limits~\eqref{eq:planner:flow}. The credit price at bus~$n$ is
\begin{equation}\label{eq:loc-price}
  \pi_n = \lambda
  - \sum_{\ell\in\LL}(\mu_\ell^+-\mu_\ell^-)\Phi_{\ell n},
\end{equation}
and agent~$i$ settles its net position at its own bus,
$p_i^{\mathrm{LOC}}=\pi_{n(i)}q_i$. Positive payments are buyer debits and
negative payments are seller credits. When no modeled flow limit binds, all
flow multipliers vanish and the rule reduces to the uniform price~$\lambda$.
With congestion, the locational settlement residual is
\begin{equation}\label{eq:loc-residual}
  \sum_i p_i^{\mathrm{LOC}}
  = \sum_{\ell\in\LL}
  \left(\mu_\ell^+\Fbar_\ell^+ +
        \mu_\ell^-\Fbar_\ell^-\right) \ge 0.
\end{equation}
To obtain this identity, use $\sum_iq_i=0$ and
$\Delta P_n=-\sum_{i:n(i)=n}q_i$:
\[
\sum_i\pi_{n(i)}q_i
=\sum_{\ell\in\LL}(\mu_\ell^+-\mu_\ell^-)
  \sum_n\Phi_{\ell n}\Delta P_n.
\]
Complementary slackness substitutes
$\mu_\ell^+\sum_n\Phi_{\ell n}\Delta P_n
=\mu_\ell^+\Fbar_\ell^+$ and
$-\mu_\ell^-\sum_n\Phi_{\ell n}\Delta P_n
=\mu_\ell^-\Fbar_\ell^-$, yielding
Eq.~\eqref{eq:loc-residual}. Every term on its right-hand side is
nonnegative because both the directional margins and their multipliers are
nonnegative. The settlement residual therefore cannot be negative, so the rule
requires no external subsidy in the stated LP. A positive residual makes the
rule revenue adequate but not budget balanced. Equation~\eqref{eq:loc-residual}
establishes the absence of a settlement deficit, whereas truthful bidding and
participant gains depend on separate incentive and participation conditions.
Section~\ref{sec:payment-rules} compares
locational settlement with VCG and uniform pricing using budget outcomes,
manipulation exposure, participation, and computational cost. Under truthful
bids, stationarity for $c_i$ in the direct-curtailment program gives
$v_i-\pi_{n(i)}=\alpha_i-\beta_i$, where $\alpha_i$ and $\beta_i$ are the
multipliers on the bounds $c_i\ge 0$ and $c_i\le L_i$. If $q_i>0$, then
$c_i<\dbar_i\le L_i$, so complementary slackness gives $\beta_i=0$ and
$v_i-\pi_{n(i)}\ge0$. If $q_i<0$, then $c_i>\dbar_i\ge0$, so
$\alpha_i=0$ and $v_i-\pi_{n(i)}\le0$; the case $q_i=0$ is immediate.
Thus $(v_i-\pi_{n(i)})q_i\ge 0$, which gives every agent at least its
no-trade surplus in the single-event model. This
guarantee leaves strategic bidding unresolved and does not account for failed
delivery or an agent's alternatives outside the CCM. Operational implementation
therefore requires governing market rules that specify bid monitoring, select
one settlement price when the LP has multiple valid dual prices, and determine
who receives any congestion residual.

\subsubsection{Exact MILP Reformulation}\label{sec:milp}

\hl{Ordinary CCM operation uses the linear program above. The offline
incentive analysis asks a different question: how does the optimal allocation
change when an agent changes its bid? To answer it, we express the clearing
problem's optimality conditions as constraints and use binary variables to
keep those constraints linear. The resulting MILP reproduces the optimal
allocation for each fixed bid and can be re-solved across candidate bids.}

\begin{assumption}[Valid Big-M constants]\label{ass:bigm}
The bounded representative KKT formulation contains at least one KKT
certificate for every lower-level optimal allocation admitted by the stated
bilevel selection rule. For each complementarity pair in that formulation,
finite constants $M_a,M_b>0$ bound the corresponding primal slack and dual
multiplier.
\end{assumption}

Assumption~\ref{ass:bigm} holds whenever the dual optimal set of
the reduced LP is bounded. Sufficient conditions and heuristic
bound formulas (including the treatment of co-located agents, for
whom PTDF differences vanish) are given in the Appendix.

\begin{proposition}[Exact MILP reformulation]\label{prop:milp}
Under Assumption~\ref{ass:bigm}, the CCM clearing LP is
allocation-equivalent to a single-level MILP with
\begin{equation}\label{eq:bincount}
  \underbrace{\card{\NN}}_{\text{lower bounds on } c_i} \;+\; \underbrace{\card{\NN}}_{\text{upper bounds on } c_i} \;+\; \underbrace{2\card{\LL}}_{\text{PTDF limits}} \;+\; \underbrace{\card{\NN}(\card{\NN}-1)}_{\text{credit nonnegativity}}
\end{equation}
binary variables, one per complementarity pair in the reduced KKT
system.
\end{proposition}

\begin{proof}[Proof sketch]
The reduced LP (Lemma~\ref{lem:reduced}) has only inequality
constraints, so KKT conditions are necessary and
sufficient~\cite{bertsimas_introduction_1997}.
Linearizing each complementarity pair with the Fortuny-Amat
and McCarl method~\cite{fortuny-amat_representation_1981}
introduces one binary per pair, yielding the
count~\eqref{eq:bincount}.  The full KKT system and
linearization are in \ref{app:kkt}.
\end{proof}

\hl{The distinction between the clearing LP and its MILP reformulation
separates operational clearing from offline incentive
analysis. The operator can clear the CCM by solving the original LP once. The
offline analysis instead supplies each candidate bid externally and re-solves
the MILP at each point. Under the stated Big-M assumption, every solve
preserves the operator's optimal response to the submitted bid. The same
fixed-bid representation supplies the lower-level optimality block required
to build a single-program strategic-bidding model. When the LP has multiple
allocation-relevant optima, the MILP selects one according to the tie-breaking
convention discussed in}~\ref{app:tiebreak}.

The preceding formulation states these assumptions explicitly or implicitly.
Table~\ref{tab:assumptions} collects them in one place to clarify the scope of
the theoretical and numerical results.

\begin{table}[htp]
\centering
\scriptsize
\setlength{\tabcolsep}{5pt}
\renewcommand{\arraystretch}{1.05}
\caption{Key modeling assumptions and their implications.}
\label{tab:assumptions}
\begin{tabular}{@{}>{\raggedright\arraybackslash}p{5.1cm}>{\raggedright\arraybackslash}p{9.0cm}@{}}
\toprule
Assumption & Implication \\
\midrule
\multicolumn{2}{@{}l}{\textit{Market boundary}} \\
\addlinespace[2pt]
The requirement $D$, references $L$, and obligations $\dbar$ are fixed before
CCM clearing.
& Results are conditional on the upstream reference, obligation, and network-headroom rules. \\
\addlinespace
One static clearing period
& Results exclude storage, load recovery, intertemporal limits, and repeated-game incentives. \\
\addlinespace[3pt]

\multicolumn{2}{@{}l}{\textit{Physical representation}} \\
\addlinespace[2pt]
Direct-current power flow with fixed PTDFs and upstream line-flow headroom.
& Network results cover incremental CCM flows, not total alternating-current system feasibility. \\
\addlinespace
Each agent is one metered load at one bus.
& Aggregated agents require bus-level disaggregation or explicit internal constraints. \\
\addlinespace
Each agent can and does deliver its cleared curtailment.
& Implementation requires certified capability limits, derating, and non-performance rules. \\
\addlinespace[3pt]

\multicolumn{2}{@{}l}{\textit{Economic and mechanism-design structure}} \\
\addlinespace[2pt]
Each agent has a constant marginal VOLL $v_i$.
& Interruption costs and bids are linear; bid curves can represent non-linear costs. \\
\addlinespace
Every buyer can trade with every seller.
& Eligibility restrictions can make some planner-feasible allocations unreachable. \\
\addlinespace
VOLL bids are the only private information.
& The incentive analysis excludes private capability, delivery probability, and VOLL-estimation errors. \\
\addlinespace
VCG uses pivot transfers, report-independent tie-breaking, and a leave-one-out
requirement of $D-\dbar_i$.
& Truthful bidding is dominant and truthful surplus is at least the no-trade baseline. \\
\bottomrule
\end{tabular}
\end{table}

\FloatBarrier
\subsection{Real-World Operation and Policy Relevance}%
\label{sec:policy-relevance}

In operation, the CCM would sit downstream of the existing reliability
process. The applicable tariff or emergency procedure would determine the
system curtailment requirement, allocate the initial obligations, and identify
eligible loads. A short bidding window would then precede network-constrained
clearing under Eq.~\eqref{eq:ccm}. The operator would communicate each load's
final service limit, verify compliance under the tariff's measurement
rule, and post the resulting buyer debit or seller credit to the event
settlement. Existing demand-response
systems already provide useful infrastructure for registration, dispatch,
metering, performance verification, and capacity or energy payments
\cite{ferc_dr_assessment_2024}. The CCM adds a clearing and settlement layer
for transferring pre-existing curtailment obligations among eligible loads.

Settlement also requires a clear performance reference. The model treats
$L_i$ as an operator-approved input fixed before CCM bidding, while leaving its
measurement rule to the upstream process. Possible rules include recent
interval-metered demand, a historical reference verified by the operator, and
a contractual maximum service entitlement. Recent meter data are readily
observable but can vary with short-term operating conditions and create an
incentive to raise consumption before the reference interval. Historical
references fit existing demand-response systems but require periodic updates
when a facility's normal demand changes. A contractual entitlement provides a
clear limit for large non-firm loads where the interconnection agreement
already defines enforceable service rights, although it can exceed actual
event demand.

A practical implementation would use an operator-verified historical or
registered reference to determine eligible curtailment capability and an
absolute post-clearing ceiling to verify event compliance. PJM follows this
structure by calculating eligible demand-response capacity as the difference
between a Peak Load Contribution and a predetermined Firm Service Level,
adjusted for system losses. NYISO similarly uses the difference between an
Average Coincident Load and a committed maximum demand, adjusted for its
applicable transmission loss factor
\cite{pjm_capacity_manual_2025,nyiso_icap_manual_2026}. Where a non-firm
interconnection agreement already establishes an enforceable service limit,
an operator-approved contractual maximum could serve as $L_i$ instead of a
historical reference. Under any of these choices, the final obligation~$c_i$
sets the highest consumption
permitted during the event, $L_i-c_i$, and meter data test actual consumption
against that limit. For example, $L_i=10$~MW and $c_i=3$~MW produce a 7~MW
service ceiling. A meter reading of 6.8~MW complies, whereas 7.2~MW violates the
assigned limit.

BiTraDER uses a seller-declared forecast as its performance reference. Before
the event, each seller reports its expected generation or demand for each
half-hour settlement period. This declared value defines the trading volume
and provides the reference from which delivered action is measured. For
example, a generator that declares expected output of 2~MW and sells 1~MW of
curtailment receives an instruction to reduce output to 1~MW. If the meter
records 1.4~MW, BiTraDER counts 0.6~MW of delivered curtailment relative to the
declared 2~MW baseline. If the generator would have produced only 1.5~MW without
the instruction, the declaration overstates event-caused curtailment by
0.5~MW. The proposed rules address this risk through routine ex post comparisons
of declared baselines and meter data, with penalties for systematic
misstatement~\cite{afry_bitrader_2023}.

\Needspace{9\baselineskip}
The operational distinction concerns who establishes the reference and when.
BiTraDER measures delivery from a seller-declared forecast for the event's
half-hour settlement period. A CCM implementation can fix $L_i$ through an
operator-approved upstream rule before CCM bids are submitted and settle
compliance against the cleared absolute ceiling. This approach avoids reliance
on an event-specific seller forecast, although an approved reference can still
become stale or misrepresent available capability. The governing market rules
would therefore prequalify controllable load, cap each seller's position at
certified capability, update load references when operating conditions change,
derate capability after repeated under-delivery, secure collateral, and impose
non-performance charges when metered consumption exceeds the cleared ceiling.

BiTraDER and NYISO's Special Case Resources illustrate complementary
implementation components. BiTraDER places trading after the distribution
network operator specifies the constraint and master curtailment order, whereas
NYISO provides established processes for demand-response registration,
dispatch, metering, and capacity or energy payments
\cite{afry_bitrader_2023,electricity_north_west_bitrader_2022,nyiso_dr_report_2024}.
A CCM implementation would additionally require rules for eligibility, bidding,
clearing, performance measurement, non-performance, collateral, payment
coordination, and settlement imbalances. The theoretical results are
conditional on these institutional inputs and rules.

\FloatBarrier

\section{Test Networks and Benchmark Design}\label{sec:numerical}

We test the CCM on three networks, each chosen to isolate a different structural feature of
the mechanism. The 3-bus, 6-agent instance places agents on opposite sides of a single
binding constraint, so every credit flow and dual variable can be inspected directly. The
modified IEEE 24-bus, 8-agent instance distributes agents across a meshed 38-line topology~\cite{ieee_rts_1979}
where congestion can appear on multiple lines simultaneously. The reduced 57-bus New York
grid instance assigns five zonal-proxy agents to a sequential corridor in which the most
valuable trades must cross a chain of nested interface constraints to reach the highest-VOLL
zone. All models are implemented in Pyomo~\cite{hart_pyomo_2011} and solved with Gurobi Optimizer 13.0.1~\cite{gurobi_1301_2026}
under a Web License Service license (absolute and relative MIPGap
$=10^{-6}$). The remainder of this section defines the fixed benchmarks,
randomized scenario design, and multi-hour case construction.
Section~\ref{sec:results} reports their results.

\subsection{Test Networks}\label{sec:testbeds}
 
\subsubsection{3-Bus Network}

The first network consists of three buses and three transmission lines with symmetric
capacities (i.e., $\bar{F}^+_\ell = \bar{F}^-_\ell \equiv \bar{F}_\ell$):
$\bar{F}_{1\to2} = 80$~MW, $\bar{F}_{2\to3} = 100$~MW, and $\bar{F}_{3\to1} = 40$~MW
(a deliberate bottleneck). All line
reactances are equal, and Bus~3 serves as the reference bus.  Six
agents participate in a single-period scarcity event requiring
$D = 230$~MW of total curtailment.  Table~\ref{tab:agents-3bus}
lists each agent's bus location, baseline load, exogenous
curtailment obligation, and true VOLL.  The VOLL values span a 100:1 ratio, from \$50{,}000/MWh (DC-1,
a data center running real-time inference) to \$500/MWh (VPP-2,
a battery and HVAC aggregator).  This range is conservative
relative to the $167\times$ spread reported in recent customer
surveys~\cite{charles_gibbons_value_2024}.
 
\begin{table}[ht]
\centering
\caption{Agent parameters for the 3-bus network. Here, $L_i$ denotes the agent's load level before the scarcity-event curtailment reallocation. It is not a counterfactual demand-response baseline used to measure avoided consumption. $\dbar_i$ is the
  exogenous curtailment obligation, and $v_i$ is
  value of lost load (VOLL), the cost of unserved
  energy.}
\label{tab:agents-3bus}
\begin{tabular}{lcccc>{\raggedright\arraybackslash}p{2.6cm}}
\toprule
Agent & Bus & $L_i$ (MW) & $\dbar_i$ (MW)
  & $v_i$ (\$/MWh) & Type \\
\midrule
DC-1   & 1 & 200 & 50 & 50{,}000 & Data center \\
DC-2   & 1 & 150 & 50 & 20{,}000 & Data center \\
C\&I   & 2 & 100 & 30 &  5{,}000 & C\&I \\
Res    & 2 &  80 & 20 &  1{,}500 & Residential agg. \\
VPP-1  & 3 &  60 & 40 &  1{,}000 & VPP \\
VPP-2  & 3 &  50 & 40 &    500   & Battery + HVAC \\
\midrule
Total  &   & 640 & 230 &         & \\
\bottomrule
\end{tabular}
\end{table}

Figure~\ref{fig:network-credit-flows} shows the 3-bus topology
with representative CCM credit flows; the Bus~3 to Bus~1
bottleneck limits relief trades from low-VOLL to high-VOLL agents.

\begin{figure}[ht]
\centering
\includegraphics[width=0.5\textwidth]{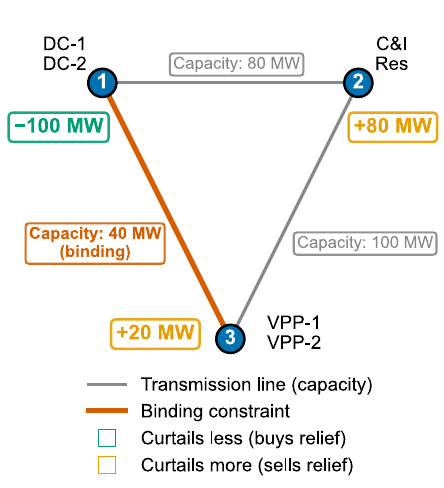}
\caption{3-bus network schematic with representative CCM credit
outcomes. Signed labels show the net change in curtailment at each bus after trading. A negative value (green) indicates the bus buys relief and curtails less, while a positive value (yellow) indicates the bus sells relief and curtails more. The Bus~3 to Bus~1 corridor is the deliberate network
bottleneck.}
\label{fig:network-credit-flows}
\end{figure}

Under the DC approximation with equal reactances and Bus~3 as
reference, the PTDF matrix is
\[
  \Phi =
  \begin{bmatrix}
    +\tfrac{1}{3} & -\tfrac{1}{3} \\[2pt]
    +\tfrac{1}{3} & +\tfrac{2}{3} \\[2pt]
    -\tfrac{2}{3} & -\tfrac{1}{3}
  \end{bmatrix},
\]
with rows corresponding to lines $(1{\to}2,\; 2{\to}3,\;
3{\to}1)$ and columns to non-reference buses $(1, 2)$.

\subsubsection{IEEE 24-Bus Network}
 
The second network places eight strategic agents on the IEEE
Reliability Test System, preserving the full 24-bus,
38-branch transmission topology.  Table~\ref{tab:agents-case24}
lists each agent's bus, baseline load, exogenous curtailment
obligation, and VOLL\@.  The agents span the same four facility
types as the 3-bus instance (data centers, commercial and
industrial loads, residential aggregators, and virtual power
plants), and their VOLL values preserve the 100:1 ratio
(\$50{,}000/MWh to \$500/MWh).  The total curtailment obligation
is 492~MW, distributed across the eight agents.

\begin{table}[ht]
\centering
\caption{IEEE 24-bus agent locations, pre-curtailment loads $L_i$, assigned
curtailment obligations $\dbar_i$, values of lost load $v_i$, and facility
types.}
\label{tab:agents-case24}
\begin{tabular}{lcccc>{\raggedright\arraybackslash}p{2.6cm}}
\toprule
Agent  & Bus & $L_i$ (MW) & $\dbar_i$ (MW)
  & $v_i$ (\$/MWh) & Type \\
\midrule
DC-1   & 22 & 200 & 105 & 50{,}000 & Data center \\
DC-2   & 21 & 150 &  88 & 20{,}000 & Data center \\
C\&I-1 & 3  & 180 &  88 &  5{,}000 & C\&I \\
C\&I-2 & 8  & 170 &  70 &  3{,}000 & C\&I \\
Res-1  & 9  & 100 &  53 &  1{,}500 & Residential agg. \\
Res-2  & 10 &  80 &  35 &  1{,}000 & Residential agg. \\
VPP-1  & 16 &  60 &  35 &    800   & VPP \\
VPP-2  & 23 &  50 &  18 &    500   & Battery + HVAC \\
\midrule
Total  &    & 990 & 492 & & \\
\bottomrule
\end{tabular}
\end{table}

Congestion is induced by derating the transformers that feed
Bus~9, which creates a binding inter-area bottleneck analogous
to the 40~MW corridor in the 3-bus network.  This setup embeds
the CCM in a standard transmission topology to test whether its
welfare and participation properties survive network
heterogeneity, while still abstracting from generator dispatch,
reserve activation, and operator-specific curtailment protocols.

\subsubsection{NYGrid Network}\label{sec:edit-nygrid-testbed}

The third network uses five representative buses associated with NYISO zones A, E, F,
I, and K as CCM agents. They span the main upstate-to-downstate congestion corridor,
from Niagara~E
(Zone~A) through Moses~W and Rotterdam
(Zones~E and F) to the downstate load centers at CE~UG (Zone~I)
and Northport on Long Island (Zone~K). In the reduced PTDF model, IF\_8 acts as the gateway from the upstate proxy buses into the
downstate corridor, IF\_9 captures the broader Southeast New York (SENY) transfer
cut, and IF\_11 separates Northport from the other four proxy buses. CE~UG is geographically downstate but lies on the upstream side of IF\_11 in this reduced model. Although the reduced network retains all 11 NYISO interface constraints, the operationally important corridor for the CCM is the sequential path IF\_8--IF\_9--IF\_11, which governs transfers toward Long Island.  The full $11 \times 5$ interface PTDF matrix is given in Table~\ref{tab:ptdf-matrix} (\ref{app:additional-results}). Each agent is a zonal proxy anchored at one representative bus.  This tests whether the CCM remains effective beyond facility-level
examples while preserving network bottlenecks from a
realistic New York grid model.

A full-year production cost simulation of the New York
grid (built on NLR's Sienna platform as part of the ACORN project~\cite{sienna_nlr_2026,open_source_nys_grid_2022,acorn_nygrid_2025}) supplies the scarcity scenario.  The simulation applies
climate-adjusted 2021 load profiles under a Representative Concentration Pathway (RCP)~4.5
warming trajectory, scales all loads by $1.8\times$ to
induce deep scarcity, and solves 365 daily unit
commitment problems with PTDF-based DC power flow.  We
extract the peak curtailment hour (July~9, 2021 at
6:00~PM), which requires $D = 6{,}815$~MW of total
load shedding across the five agent buses.

Table~\ref{tab:agents-nygrid} lists the agent
parameters.  The VOLL values increase from
\$500/MWh in western upstate (Zone~A) to
\$50{,}000/MWh on Long Island (Zone~K), preserving
the 100:1 spread used in the other networks.
Curtailment obligations are assigned pro-rata
($\dbar_i = L_i \cdot D / \sum_j L_j$) on
this network.

\begin{table}[ht]
\centering
\caption{Agent parameters for the NYGrid network.
  $L_i$ is load at the representative bus, $\dbar_i$ is the assigned
  curtailment obligation, and $v_i$ is value of
  lost load (VOLL), the cost of unserved energy.}
\label{tab:agents-nygrid}
\begin{tabular}{lccccc}
\toprule
Agent & Bus & Zone & $L_i$ (MW)
  & $\dbar_i$ (MW) & $v_i$ (\$/MWh) \\
\midrule
Niagara E   & 55 & A & 1{,}990
  &   869 &     500 \\
Moses W     & 47 & E & 1{,}406
  &   614 &   1{,}250 \\
Rotterdam   & 41 & F & 3{,}419
  & 1{,}492 &   3{,}000 \\
CE UG       & 78 & I & 2{,}804
  & 1{,}224 &  20{,}000 \\
Northport   & 80 & K & 5{,}993
  & 2{,}616 &  50{,}000 \\
\midrule
Total       &    &   & 15{,}612
  & 6{,}815 & \\
\bottomrule
\end{tabular}
\end{table}

The reduced network enforces 11 interface constraints
(IF\_1 through IF\_11), derived from the NYISO zonal transfer limits
in the reduced Northeast Power Coordinating Council (NPCC) model \cite{open_source_nys_grid_2022}, with
asymmetric flow limits.  The binding constraint is
IF\_11 (the Long Island interface), which permits
1{,}290~MW southbound but only 515~MW northbound in the simulation parameterization~\cite{acorn_nygrid_2025}.
All interface limits are inherited from the ACORN NYGrid parameterization,
and the experiments apply them directly as the available directional margins
for incremental flows. The congestion results therefore measure market
behavior under the inherited interface limits of the test system. Although
the upstream ACORN--Sienna simulations record branch and interface flows, the
CCM input export used here does not include a validated base-flow state matched
to the pro-rata starting allocation. We therefore apply the inherited
interface limits as assumed directional margins for incremental flows rather
than as event-specific remaining headroom.
This directional bottleneck limits how much
curtailment relief the downstate proxy agents can
import from the upstate proxy agents.

\subsection{Benchmark Regimes}\label{sec:benchmark-regimes}

We evaluate five curtailment regimes against a common metric.
Each regime produces a curtailment vector~$c$; we then
compute \emph{social welfare} for the reduced
curtailment-allocation problem as
\begin{equation}\label{eq:welfare-metric}
  W \;=\; \sum_{i \in \NN} v_i \,(L_i - c_i),
\end{equation}
using each agent's VOLL~$v_i$ (a private valuation that can
only be estimated in practice and is unobserved by the market
operator), regardless of how the regime assigns curtailment.
Social welfare equals the total value of served load, weighting
each agent's curtailment loss by its own VOLL.  Because
generation dispatch and the aggregate curtailment requirement
are fixed ex ante, supply-side costs do not vary across regimes,
and social welfare reduces to the demand-side total value of
served load.  Throughout this paper, ``welfare'' and ``social
welfare'' refer to this same quantity~$W$.

The first two regimes are non-market benchmarks that assign
curtailment by fixed rules and involve no credit trading.

\textit{Administrative allocation.}
Each agent bears its exogenous obligation:
$c_i^{\mathrm{admin}} = \dbar_i$.  This regime reproduces the
status quo in which the ISO assigns curtailment without
considering private valuations.  It tests whether the CCM
improves on an allocation that ignores VOLL heterogeneity.

\textit{Pro-rata allocation.}
Each agent is curtailed at the system-wide ratio:
$c_i^{\mathrm{pr}} = L_i \cdot (D / \sum_j L_j)$. Pro-rata sharing distributes curtailment in proportion to load. U.S. system operators do not uniformly apply pro-rata at emergency load shed; current practice is heterogeneous and includes rotating feeder-based disconnection, priority-tier rules, and supervisory control and data acquisition (SCADA) under-frequency load shedding. Pro-rata is used here as a welfare-comparable numerical baseline, where it tests whether proportional allocation outperforms or underperforms administrative assignment when VOLL varies across agents.

Because administrative and pro-rata allocations are fixed-rule benchmarks rather than optimization outcomes, we separately audit their feasibility under the same PTDF constraints imposed on the CCM and planner. For each fixed-rule allocation $c^r$, $r \in \{\mathrm{admin},\mathrm{pr}\}$, we compute the induced injection changes $\Delta P_n(c^r)$, line/interface flows $\sum_n \Phi_{\ell n}\Delta P_n(c^r)$, directional slacks, and the minimum PTDF slack. A fixed-rule welfare value is physically comparable with CCM and planner welfare only when this minimum slack is nonnegative up to numerical tolerance.

The remaining three regimes solve optimization problems.

\textit{CCM (market).}
The CCM clearing LP~\eqref{eq:ccm} is solved under truthful
bidding ($b_i = v_i$), with VCG payments.  This regime tests
whether the market mechanism recovers the planner's social welfare
through decentralized bids, as
Corollary~\ref{cor:welfare} predicts.

\textit{Omniscient planner.}
The planner LP~\eqref{eq:planner} is solved with full knowledge
of all VOLLs.  The planner takes the same curtailment
target~$D$, network topology, and line limits as the CCM and the
fixed-rule benchmarks; it does not control generation dispatch or set the
supply--demand balance.  Its only advantage over the CCM is
direct access to private valuations, which removes the need for
bids.  The planner therefore provides the welfare upper bound
$W_{\mathrm{plan}}$ against which all other regimes are measured, where
any gap between $W_{\mathrm{CCM}}$ and $W_{\mathrm{plan}}$
reflects information loss from decentralized bidding rather than a
difference in available supply or network capacity.

\textit{Copperplate (unconstrained).}
The copperplate benchmark solves the planner LP~\eqref{eq:planner}
with all transmission constraints removed.  This isolates the welfare cost
imposed by network congestion, since any gap between copperplate and
constrained planner welfare is attributable to binding line limits.

To quantify the CCM's approach to planner-welfare equivalence, we
report the efficiency ratio
\begin{equation}\label{eq:eta-eff}
  \eta_{\mathrm{eff}}
    \;=\; \frac{W_{\mathrm{CCM}}}{W_{\mathrm{plan}}}.
\end{equation}
A ratio of one confirms that the CCM achieves the planner's
welfare on that instance.

For the randomized and multi-hour studies, we also report the share of
the initial interruption cost removed by trading. Let
$C_0=\sum_i v_i\dbar_i$ denote interruption cost before reallocation and
$C_{\mathrm{CCM}}=\sum_i v_i c_i^{\mathrm{CCM}}$ denote interruption cost
after CCM clearing. The avoided-cost fraction is
\begin{equation}\label{eq:avoided-cost}
  A_{\mathrm{CCM}}
  = \frac{C_0-C_{\mathrm{CCM}}}{C_0}.
\end{equation}
This measure isolates the cost component that reallocation can change.

\subsection{Randomized Scenario Design}\label{sec:randomized-design}

The randomized study tests whether results from the three fixed operating
points persist when customer composition, interruption values, curtailment
depth, and network capacity vary jointly. One scenario realization fixes all
of these inputs before any obligation rule or market outcome is computed. We
then evaluate the four initial-obligation rules on that same realization, so
within-scenario differences can be attributed to the obligation rule. All
pseudo-random inputs originate from root seed 20260802.

\paragraph{Initial-obligation rules}
Each rule is an upstream method for constructing~$\dbar$. The CCM objective,
network constraints, and payment rule remain fixed across rules. Every rule
satisfies
\begin{equation}\label{eq:obligation-rule-contract}
  \sum_i \dbar_i=D,
  \qquad 0\leq \dbar_i\leq L_i.
\end{equation}
Here, $L_i$ is customer~$i$'s scenario load and $D$ is the scenario-wide
curtailment requirement. Pro-rata allocation sets
$\dbar_i=L_iD/\sum_jL_j$ and redistributes any amount above a customer's
load bound among the remaining customers. Connection order fills obligations
along an exogenous queue until $D$ is assigned. This stylized rule follows the
UK Power Networks' Energy Exchange example of trading positions in a
curtailment queue~\cite{ukpn_energy_exchange_2022}. The implementation omits
the project's curtailment index, the score used to order participants. The
experiment fixes one random customer order per scenario realization.
Contractual priority places
each customer in a non-firm or firm tier with equal probability, exhausts the
non-firm tier first, and allocates pro-rata within a tier
~\cite{ferc_pro_forma_oatt_2020,nerc_naesb_tlr_manual_2008}. Locational
allocation weights customer~$i$ by
$|\Phi_{\ell^\star n(i)}|L_i$ for a stressed element~$\ell^\star$ declared
before clearing, then assigns any remainder pro-rata among customers with zero
exposure. The declared element is line 3--1 in the 3-bus case. It is the two
derated transformers feeding Bus~9 in the IEEE 24-bus network and IF\_11 in NYGrid. These queue,
tier, and network inputs remain fixed across the paired value families in each
scenario realization.

\paragraph{System conditions}
The system-condition sample tests whether the CCM's performance depends on
system severity. We vary transmission capacity with the
multiplier $\kappa_F\in[0.5,2.0]$. A value of one retains every line and
interface limit defined for the three test networks in
Section~\ref{sec:testbeds}. A value of 0.5 halves those limits, and 2.0 doubles
them. We
vary curtailment depth with
$\rho=D/\sum_iL_i\in[0.05,0.50]$, so the required reduction ranges from 5\% to
50\% of total customer load. This range includes the approximately 36\%, 50\%,
and 44\% shares in the fixed 3-bus, IEEE 24-bus, and NYGrid examples. A scrambled
Latin hypercube, which spreads draws across each range, samples $\kappa_F$ and
$\rho$ jointly. Varying these two inputs tests whether the reported gains
persist as congestion and the amount of required curtailment change together.

\paragraph{Customer populations}
This population sampling tests whether the results depend on who participates
and where those customers connect, while holding each test network's total
load fixed. The design crosses three customer mixes, three participant counts,
and two data-center VOLL values. Mix~1 (M1) assigns 40\% of customers to
residential loads, 40\% to medium or large commercial and industrial loads,
and 20\% to data centers. Mix~2 (M2) uses shares of 20\%, 30\%, and 50\% in the
same order to represent a data-center-heavy population. Mix~3 (M3) uses 60\%
residential and 40\% commercial and industrial customers, with no data
centers. The participant count is either the original count for the test
network (six, eight, or five), or 20 or 50. The data-center VOLL is
\$20,000/MWh or \$50,000/MWh, the two values used by the data-center agents in
the fixed examples. For each draw, the sampler apportions the customer count
across the existing load buses in proportion to bus load, using
largest-remainder rounding with at least one customer per bus. Class
membership follows an exact largest-remainder quota, so a declared share
yields the same customer count in every draw and a declared class always
receives at least one customer. Within each bus, a customer's load share
equals a \hl{scenario-design} class size weight multiplied by an independent
Gamma(2,1) random draw and is normalized so that the bus total load is
conserved exactly. The size weights are 0.3 for residential, 1.0 for medium
or large commercial and industrial, 2.0 for data-center, and 0.05 for small
commercial and industrial customers, and they encode the typical relative
sizes of these customer classes. This construction changes customer
composition, VOLL, and bus placement without changing total demand.

\paragraph{Customer-value families}
We use five customer-and-VOLL constructions, labeled F0 through F4, where F
denotes family. Each family answers a different sensitivity question. F0 keeps
the original agents, loads, bus locations, and VOLLs, providing an exact link
to the fixed test-network examples. F1 is the baseline class-value
construction. It assigns survey estimates of \$3,964/MWh to residential
customers and \$22,721/MWh to medium or large commercial and industrial
customers~\cite{charles_gibbons_value_2024}. Data centers use either
\$20,000/MWh or \$50,000/MWh, the two original test-network values. F2 keeps
the F1 loads and bus locations but assigns every customer \$22,721/MWh, the F1
value for medium or large commercial and industrial customers. Using a value
already present in F1 removes differences among customer VOLLs without
introducing another value level. Reallocation cannot change welfare when all
customers have the same VOLL, making F2 a no-heterogeneity control. F3 assigns
10\% of customers to small commercial and
industrial loads at the survey estimate of \$666,907/MWh and rescales the
other class shares by 0.9. This case tests the influence of the survey's
largest class estimate. F4 uses the same customers as F3 but \hl{caps VOLL at
the pre-specified sensitivity value of} \$150,000/MWh. Comparing F4 with F3 shows how much the
result depends on that largest estimate.

\paragraph{Balanced sample}
The main sample gives equal weight to every population design and evaluates
all four obligation rules on the same scenario realization. The three mixes,
three participant counts, and two data-center VOLLs form
$3\times3\times2=18$ population designs. We generate 150 realizations of each,
giving 2,700 realizations per testbed, and evaluate every realization under F1,
F3, and F4. We add 100 F2 controls to verify that equal VOLLs produce no
welfare gain. We also add 54 exact operating-point cases, covering the 18
population designs and three families at $\kappa_F=1$ and the original
test-network curtailment depth, plus the single F0 case. The resulting 8,255
scenarios per testbed yield 99,060 obligation-rule allocations across the
three testbeds. Equal representation prevents a customer mix, participant
count, or data-center value from dominating the summaries.

\paragraph{Additional sensitivity checks}
Two separate checks extend the main sample. The participant-count check draws
fresh populations of 100 and 200 customers using the same class VOLLs,
customer-size rules, and bus-assignment method. Holding total load fixed makes
these customers smaller on average, so the check tests whether market outcomes
depend on participant granularity or the resulting VOLL and bus placements. A
network check varies the NYGrid IF\_11 limit over 200, 515, 1,000, 1,290,
1,500, 2,000, and 3,000~MW while holding the other interface limits fixed. This
check isolates how transfer capacity into Long Island limits curtailment
reallocation.

\paragraph{VOLL-estimation error}
The CCM must clear using estimated customer VOLLs, so this study asks how
estimation errors change who curtails and the welfare evaluated at the
underlying values. The exact-value case provides the reference allocation.
Common multipliers of 0.5, 0.8, 1.25, and 2.0 make every VOLL estimate too low
or too high by the same proportion. These cases test sensitivity to the
overall calibration of the VOLL estimates while preserving the customer
ranking. Independent mean-preserving lognormal errors with log standard
deviations 0.1, 0.3, and 0.5 introduce increasing customer-specific variation
without shifting the mean estimate. Class-correlated lognormal errors use a
class standard deviation of 0.3 and an individual standard deviation of 0.15,
representing an error shared by a customer class plus customer-specific noise.
We also swap the estimated VOLLs of two customers on opposite sides of the
exact allocation cutoff. This simple reversal tests the error most likely to
change which customer curtails. The same random draws are reused at each error
magnitude, so changes across cases come from the error size. These cases study
errors in the values available to the market. The following test instead
allows one customer to misreport its value deliberately.

\paragraph{Strategic bid-deviation tests}
We test unilateral deviations at the fixed operating point of each network.
For each customer in turn, we set its bid to $b_i=m v_i$ while holding every
other bid at the corresponding true VOLL. The 3-bus test uses 100 multipliers
from 0.1 to 10, the IEEE 24-bus test uses 30 multipliers from 0.1 to 4.977,
and the NYGrid test uses 30 log-spaced multipliers from 0.1 to 10 together
with the exact truthful point, for 31 points in total. We also evaluate NYGrid
on a 100-point grid over the same range as a finer sensitivity test. We clear
the CCM at each bid and calculate the focal customer's payoff using its true
VOLL under VCG, uniform-price, and locational settlement. Under the
uniform-price rule, every net
credit purchase or sale settles at the system balance price. For settlement
rule $r$, we report
\begin{equation}\label{eq:bid-gain-ratio}
  g_i^r
  = \max\!\left\{0,
      \frac{\max_{m\in\mathcal{M}} U_i^r(m)-U_i^r(1)}
           {|U_i^r(1)|}
    \right\}.
\end{equation}
This ratio is the nonnegative difference between the best payoff on the tested
bid grid and the truthful-bid payoff, normalized by the magnitude of the
truthful-bid payoff.
Section~\ref{sec:incentive-results} reports the resulting manipulation gains.

\subsection{Multi-Hour Case Construction}\label{sec:multihour-design}

\paragraph{Duration and day selection}
From our analysis of public demand-response and market-monitoring records for
2005 through 2025~\cite{potomac_nyiso_som_2006,nyiso_dr_report_2010,nyiso_dr_report_2025},
we identify 79 distinct activations.
Seventy-eight report both a start and an end time. Among these 78 events, the
median duration is six hours, the 90th percentile is eight hours, and 24 last
exactly six hours. We therefore study episodes of one through six consecutive
hours, covering durations through the historical median. January 21 and 22,
2025 are the most recent consecutive six-hour mandatory calls that cover all
five mapped zones and provide complete load and response tables. We use January
21 as the primary day and January 22 as the repeat day. Both calls run from
16:00 to 22:00~\cite{nyiso_dr_report_2025}, which fixes the start time in the
core comparison.

We use demand-response records because they report when customers were asked
to reduce load, how long each call lasted, and how many megawatts the operator
requested. These data set the episode lengths and curtailment shares in the
repeated-hour cases. Reserve-shortfall records track deficits in contingency
reserves, and their duration shows how long the system was short of reserves.
Customer curtailment calls and reserve shortages can therefore last for
different periods. Firm load shedding would directly measure forced customer
interruption. The assembled NYISO record contains essentially no firm
load-shedding events caused by inadequate system supply with both duration and
location available.

\paragraph{Hourly load and curtailment depth}
For each of the five proxy buses, we aggregate its associated NYISO zone's
five-minute P-58B metered-load readings to hourly values. We divide each zone's
24 hourly values by its daily mean and scale that normalized shape to the
corresponding proxy-bus baseline in the NYGrid customer table,
Table~\ref{tab:agents-nygrid} in Section~\ref{sec:edit-nygrid-testbed}
~\cite{nyiso_p58b_2025}. This imports the observed within-day load pattern
without replacing the test network's load scale. The historical calls also set
the central curtailment depth. For the 22 six-hour activations with complete
quantities, requested megawatts divided by measured footprint load has a
median of 3.795\%. The 25th and 75th percentiles are 1.941\% and 4.513\%. We
use these three values as the central, low, and high curtailment shares. Thus,
the public record determines the core day, start time, duration range, hourly
load shape, and curtailment depth before the market cases are constructed.

\paragraph{Customer capability}
Hourly obligations are proportional,
\begin{equation}\label{eq:multihour-obligation}
  \dbar_{i,t}=L_{i,t}\frac{D_t}{\sum_jL_{j,t}},
\end{equation}
and customer~$i$ may accept at most
\begin{equation}\label{eq:multihour-capability}
  c^{\max}_{i,t}
  =\min\{L_{i,t},\,\kappa_C\dbar_{i,t}\}.
\end{equation}
We evaluate $\kappa_C=1$, 2, 4, and the original full-load bound
$c^{\max}_{i,t}=L_{i,t}$. These four cases test how the ability to accept extra
curtailment affects trading. At $\kappa_C=1$, a customer can accept only its
initial obligation, leaving no room for trade. Values of 2 and 4 allow up to
two or four times that obligation, subject to the customer's load. The
full-load case allows curtailment up to the customer's entire current load and
reproduces the assumption in the fixed test-network examples.

\paragraph{Interruption cost over consecutive hours}
Repeated curtailment can change the harm associated with the next curtailed
megawatt-hour. To test this effect, we set customer~$i$'s marginal interruption
cost in its $k$th consecutive curtailed hour to $v_i m_k$ dollars per MWh,
where $v_i$ is the customer's first-hour VOLL and $m_k$ is a duration weight.
We evaluate constant, decreasing, and increasing duration weights. The flat
shape preserves the original assumption with $m_k=1$. The decreasing shape uses
$(1.000,0.522,0.440,0.395,0.366,0.345)$, obtained from the duration terms and
residential interruption-cost values in the Interruption Cost Estimate (ICE) Calculator 2.2
report~\cite{larsen_ice_phase12_2026}. The increasing shape mirrors the
decreasing shape around the flat case and uses
$(1.000,1.478,1.560,1.605,1.634,1.655)$. The decreasing curve provides an
evidence-based direction from residential full-outage costs, while the
increasing curve is a symmetric sensitivity case for rising marginal harm.

\paragraph{Sensitivity to duration, capability, and network stress}
The primary design combines the two historical days, six episode durations,
three cost shapes, and four capability limits. It uses the 16:00 start time and
3.795\% curtailment share derived above. This grouping makes duration, customer
capability, and duration-dependent cost separate comparison axes.
The secondary design then tests one departure from the primary design at a
time. It places each duration at every start hour that remains within the primary day,
checks the two historical severity quartiles, applies stress depths of 5, 10,
15, 20, 25, 30, and 43.653\%, and varies remaining IF\_11 headroom over the
full original limit and 600, 400, 300, 250, 200, 150, 100, and 50~MW. The
43.653\% depth reproduces the fixed NYGrid example. These checks show how the
results change with start time, curtailment severity, or corridor congestion
while the other primary settings remain fixed. The network model omits hourly
generation dispatch, so the headroom values bound only the incremental flow
created by reallocation.

\paragraph{Hourly clearing sequence}
The CCM reallocates the curtailment requirement for one operating interval at
a time. The multi-hour case therefore clears each episode hour by hour as loads
and duration weights change. No episode crosses midnight. The model omits
storage state, load recovery, horizon-wide optimization, and intertemporal
bidding. The resulting conclusions concern a sequence of hourly CCM clears
under a duration-cost policy.

\FloatBarrier

\section{Results and Discussion}\label{sec:results}

\subsection{Social Welfare Comparison}\label{sec:welfare-results}

The fixed cases use the participant populations, curtailment requirements,
initial obligations, and network limits specified for the 3-bus, IEEE 24-bus,
and NYGrid test networks in Section~\ref{sec:testbeds}. Figure~\ref{fig:welfare-all}
provides the primary fixed-case comparison. It tests whether truthful CCM
clearing recovers the efficient allocation and compares that result with two
no-trade allocations. Across all three networks, measured CCM welfare matches
planner welfare to numerical tolerance ($\eta_{\mathrm{eff}}\geq0.9999$), as
predicted by Corollary~\ref{cor:welfare}. Relative to pro-rata curtailment, the
fixed-case welfare ratios range from $1.41\times$ to $1.83\times$. These
operating points lie near the upper tail
of the tested range. Section~\ref{sec:robustness} places them between the
92nd and 100th percentiles of comparable randomized cases, whose median ratios
range from $1.16\times$ to $1.26\times$. The fixed cases therefore characterize
high-gain realizations within the experimental design. The randomized medians
provide the summary of typical outcomes within that design.

The first no-trade allocation in Figure~\ref{fig:welfare-all} keeps each
participant's test-case obligation. The second assigns the total requirement
in proportion to load. Section~\ref{sec:robustness} expands this comparison to
four systematic rules for assigning initial obligations across randomized
realizations.

The 3-bus case shows why proportional assignment need not improve welfare.
The fixed initial obligations produce \$9.96M of welfare, while pro-rata
curtailment produces \$8.78M. The difference arises because the largest loads
also have the highest value of lost load (VOLL), which measures interruption
cost. Pro-rata assigns every participant the same 35.9\% curtailment rate
because $D/\sum_jL_j=230/640$. DC-1 therefore receives 71.9~MW of curtailment,
even though its fixed initial obligation is 50~MW. The additional 21.9~MW is
assigned to the participant with the highest VOLL and contributes \$1.09M to
the \$1.19M total welfare reduction relative to the fixed initial-obligation
allocation.

CCM corrects this mismatch by transferring curtailment from participants with
higher VOLL to those with lower VOLL. DC-1 and DC-2 buy full relief, so their
final curtailment is $c_i^{\mathrm{CCM}}=0$. C\&I, VPP-1, VPP-2, and the residential
aggregator accept the transferred obligations. This reallocation raises
welfare to \$13.26M, which is $1.33\times$ the fixed initial-obligation result
and equals the planner outcome.

Participant surplus requires a separate interpretation from total welfare.
Under truthful bids, VCG transfers, feasible initial obligations, and
the exit convention in Section~\ref{sec:properties},
Proposition~\ref{prop:ir} guarantees that participation leaves each agent at
least as well off as keeping its initial obligation. The exit convention is
part of that guarantee because an exiting agent's obligation is removed with
the agent instead of being reassigned to the remaining participants. The
proposition guarantees a nonnegative participation gain and permits equality. In
the three fixed cases, every participant records a strictly positive gain. The
strict gains are a result of the numerical cases studied in this paper. The
proposition guarantees only that participants are no worse off.

The 3-bus outcomes illustrate the two ways that surplus can increase. DC-1
avoids \$2.50M of interruption cost and gains \$2.25M net of its \$0.25M VCG
payment. The residential aggregator accepts 60 additional MW and receives a
VCG transfer that raises
its surplus from \$0.09M to \$0.36M. These gains do not imply that the
settlement balances because VCG can require an external subsidy. Detailed
3-bus and IEEE 24-bus participant outcomes appear in
\ref{app:fixed-outcomes}.

\begin{figure}[htbp!]
  \begin{subfigure}[t]{0.48\textwidth}
    \includegraphics[width=\linewidth]{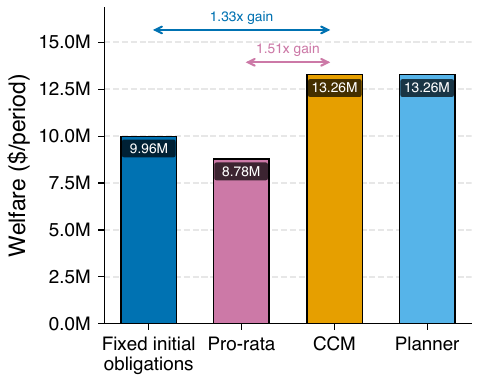}
    \subcaption{3-bus network}
    \label{fig:welfare-3bus}
  \end{subfigure}
  \hfill
  \begin{subfigure}[t]{0.48\textwidth}
    \includegraphics[width=\linewidth]{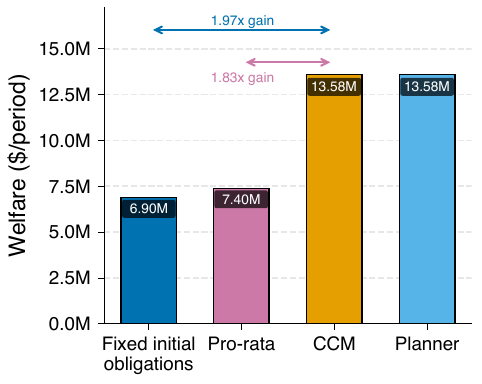}
    \subcaption{IEEE 24-bus network}
    \label{fig:welfare-case24}
  \end{subfigure}

  \bigskip
  \centering
  \begin{subfigure}[t]{0.5\textwidth}
    \includegraphics[width=\linewidth]{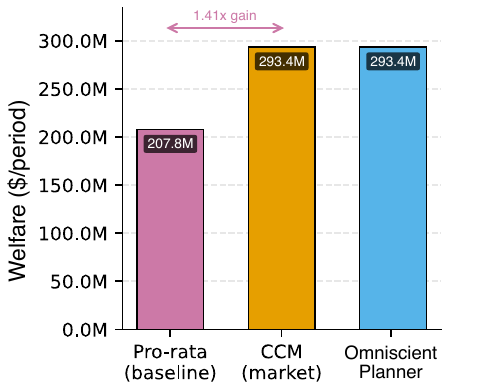}
    \subcaption{NYGrid network}
    \label{fig:welfare-nygrid_v2}
  \end{subfigure}

  \caption{Social welfare comparison across three networks.
  Panels~(a) and~(b) compare the fixed initial-obligation allocation,
  pro-rata allocation (curtailment proportional to load), CCM, and the planner.
  Panel~(c) reports only three regimes because NYGrid
  uses pro-rata initial obligations. Welfare is reported per scarcity event.
  Annotations report multiplicative gains relative to the stated comparison
  allocation. In all three networks, CCM matches planner welfare and exceeds
  both plotted no-trade allocations where both are present.}
  \label{fig:welfare-all}
\end{figure}

The IEEE 24-bus case extends the comparison from one deliberate bottleneck to
a meshed 38-line network. CCM raises welfare from \$6.90M under the fixed
initial obligations to \$13.58M and gives all eight participants a surplus
gain under VCG transfers (Figure~\ref{fig:welfare-case24}). Here, pro-rata
assignment also improves on the fixed obligations, raising welfare to
\$7.40M. This ordering reverses the 3-bus pattern. The system-wide
pro-rata curtailment rate is 49.7\%, while the fixed obligations curtail 52.5\%
of DC-1's load and 58.7\% of DC-2's load. Pro-rata therefore reduces the
burden on both high-VOLL data centers and raises welfare by \$0.50M relative
to the fixed initial obligations.

Taken together, the two facility-level cases show that neither fixed initial
obligations nor pro-rata assignment has a general welfare advantage. Their
ranking depends on how much curtailment each rule places on high-VOLL
participants. CCM removes this dependence at the final allocation stage by
allowing DC-1 and DC-2 to buy full relief in both cases while lower-VOLL
participants accept the transferred obligation. The randomized results in
Section~\ref{sec:robustness} test whether this conclusion persists across four
initial-obligation rules.

NYGrid tests the same allocation result at a coarser spatial scale. Its five
participants are zonal proxies attached to representative buses in a reduced
New York network. CCM raises welfare from \$207.79M under pro-rata initial
obligations to \$293.41M, a $1.41\times$ increase, and again matches the
planner ($\eta_{\mathrm{eff}}=1.0000$). Removing all interface limits raises
the ratio to $1.71\times$. The difference shows that network capacity, and
especially the Long Island interface IF\_11, prevents part of the beneficial
reallocation. This observation motivates the congestion analysis that
follows.

\subsection{Network Congestion Effects}\label{sec:congestion}

The NYGrid result shows that network limits can block a large share of the
otherwise efficient reallocation. We measure this effect by comparing each
network-constrained outcome with a copperplate version of the same case, in
which all transmission limits are removed. The resulting congestion cost is
\$40{,}000 per period (0.3\%) on the 3-bus network, \$161{,}295 (1.17\%) on
the IEEE 24-bus network, and \$62.34M (17.5\%) on NYGrid. Thus, congestion has
a minor aggregate effect in the first two cases and removes more than one
sixth of unconstrained welfare in NYGrid. These congestion costs occur in the
fixed cases, which Section~\ref{sec:robustness} characterizes as high-gain
realizations within the experimental design.

Congestion also creates different credit prices at different buses. The nodal
credit-price range is the difference between the highest and lowest bus prices.
It reaches \$8{,}000/MWh on the 3-bus network, \$7{,}441/MWh across participant
buses on the IEEE 24-bus network, and \$47{,}000/MWh between Long Island and
every upstate NYGrid bus. A line shadow price instead measures the increase in
welfare produced by one additional MW of capacity on a binding line or
interface. On NYGrid, IF\_11 is the only binding interface, so its
directional limit creates the entire \$62.34M congestion cost and the observed
upstate-to-Long-Island price separation.

The size of the welfare loss depends first on how much beneficial trade must
cross the bottleneck. The binding line in the 3-bus case blocks only some of
the 30 ordered participant pairs. IF\_11, by contrast, lies on the transfer
path for most beneficial upstate-to-downstate trades in NYGrid and therefore
blocks a much larger volume of welfare-improving trades. The value
difference across the bottleneck then determines the cost of each blocked MW.
The 3-bus line separates participants valued from \$500/MWh to
\$50{,}000/MWh but blocks little volume. IF\_11 separates upstate proxy values
of \$500--\$20{,}000/MWh from Northport's \$50{,}000/MWh value while the
system assigns 6{,}815~MW of total curtailment. CE~UG is geographically
downstate but lies upstream of IF\_11 in the reduced model. This combination
makes the NYGrid welfare effect much larger.

Price separation depends additionally on how strongly a transfer between two
buses loads the binding element. The stationarity condition in
Equation~\eqref{eq:app:stationarity} links the credit-price difference for an
active trade to the line shadow price and to the difference between the two
buses' power transfer distribution factor (PTDF) coefficients. The PTDF
difference is the fraction of a one-MW transfer that uses the constrained
element. A high shadow price and a large PTDF difference can therefore produce
a wide nodal price range even when little trade is blocked. This is why the IEEE
24-bus case loses only 1.2\% of welfare while its participant-bus credit prices
differ by \$7{,}441/MWh. The large nodal price range therefore does not imply a large
aggregate welfare loss.

CCM clearing must represent participant locations when the network is
congested because a single binding interface can change both the feasible
allocation and the price of curtailment relief. The next test varies that
interface capacity directly to identify when these effects appear.

Figure~\ref{fig:linesweep-all} traces this transition on the 3-bus network as
CCM welfare rises when the line~$3\to1$ limit is relaxed and reaches its maximum
at the 50~MW grid point, where the line stops binding. The corresponding shadow price
$\mu_{31}$ falls from about \$30{,}000/MWh to zero. At the tightest capacity
levels, the fixed pro-rata allocation violates the same PTDF limit imposed on
CCM and the planner because its assigned bus-level curtailment changes create
too much flow on line~$3\to1$. CCM remains feasible because its clearing
program admits only reallocations that satisfy that limit.

The NYGrid sweep gives the same welfare pattern on a larger dollar scale. CCM
and planner welfare overlap at every IF\_11 capacity under truthful bidding,
consistent with Corollary~\ref{cor:welfare}. Their welfare approaches the
approximately \$355M copperplate value as southbound capacity increases. The
pro-rata comparison stays near \$208M because the assigned pro-rata
curtailment does not change with interface capacity and remains PTDF-feasible
at every plotted point. The sweep therefore isolates the welfare unlocked by
additional transfer capability.

Additional capacity has the greatest value when the interface is tight. The
first increments admit trades with the largest VOLL differences, while later
increments admit progressively lower-value trades. Welfare stops increasing
once all beneficial trades fit within the network. In the static model, the
shadow price on a binding directional limit equals the local welfare gain from
adding one MW in that direction. For example, $\mu^+_{\mathrm{IF\_11}}$ values
one additional MW of southbound IF\_11 capacity. This one-period shadow price
measures local operating value. Evaluating IF\_11 as an investment would
require many operating periods and occurrence probabilities, which the present
study does not provide.

\begin{figure}[htbp!]
  \begin{subfigure}[t]{0.48\textwidth}
    \includegraphics[width=\linewidth]{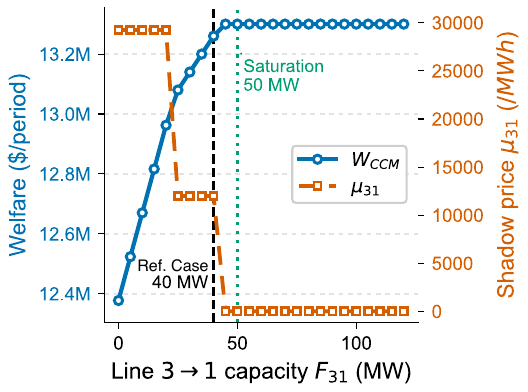}
    \subcaption{3-bus line~$3{\to}1$ capacity sweep}
    \label{fig:linesweep}
  \end{subfigure}
  \hfill
  \begin{subfigure}[t]{0.48\textwidth}
    \includegraphics[width=\linewidth]{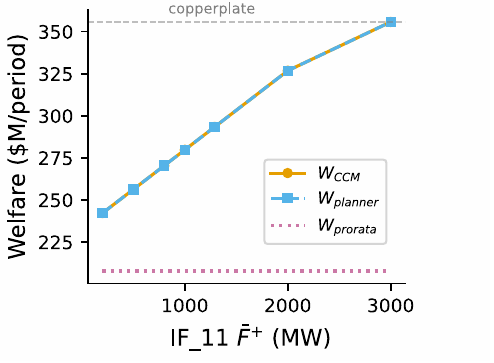}
    \subcaption{NYGrid IF\_11 southbound ($\bar{F}^{+}$) capacity sweep}
    \label{fig:linesweep-nygrid}
  \end{subfigure}
  \caption{Social welfare sensitivity to binding interface capacity.
  The left panel shows CCM welfare and the line-capacity shadow price
  $\mu_{31}$ on the 3-bus network. Welfare rises
  with line capacity and saturates as $\mu_{31}$ falls to zero.
  The right panel shows welfare under CCM,
  planner, pro-rata, and copperplate benchmarks. CCM and planner
  overlap under truthful bidding, as predicted by
  Corollary~\ref{cor:welfare}. The fixed pro-rata comparison is
  PTDF-feasible at every plotted IF\_11 capacity point because
  NYGrid obligations are assigned pro-rata by construction, so
  the no-trade comparison induces $\Delta P \equiv 0$.}
  \label{fig:linesweep-all}
\end{figure}

The next tests move beyond fixed interface-capacity sweeps to examine whether
the welfare gains persist across broader combinations of customers, networks,
and initial obligations.

\subsection{Robustness of the Welfare Gains}\label{sec:robustness}

The randomized study tests whether the fixed-case welfare gains persist when
customer composition, network conditions, and initial curtailment obligations
vary together. Figure~\ref{fig:wp1-main} organizes this test around 24
combinations of three networks, four initial-obligation rules, and the F1 and
F3 customer-value constructions defined in
Section~\ref{sec:randomized-design}. The outcome measure is avoided
interruption cost, which equals the fraction of no-trade interruption cost
removed by CCM. Dividing by the no-trade cost permits comparisons across
networks with different dollar scales.

\begin{figure}[htbp!]
\centering
\includegraphics[width=\textwidth]{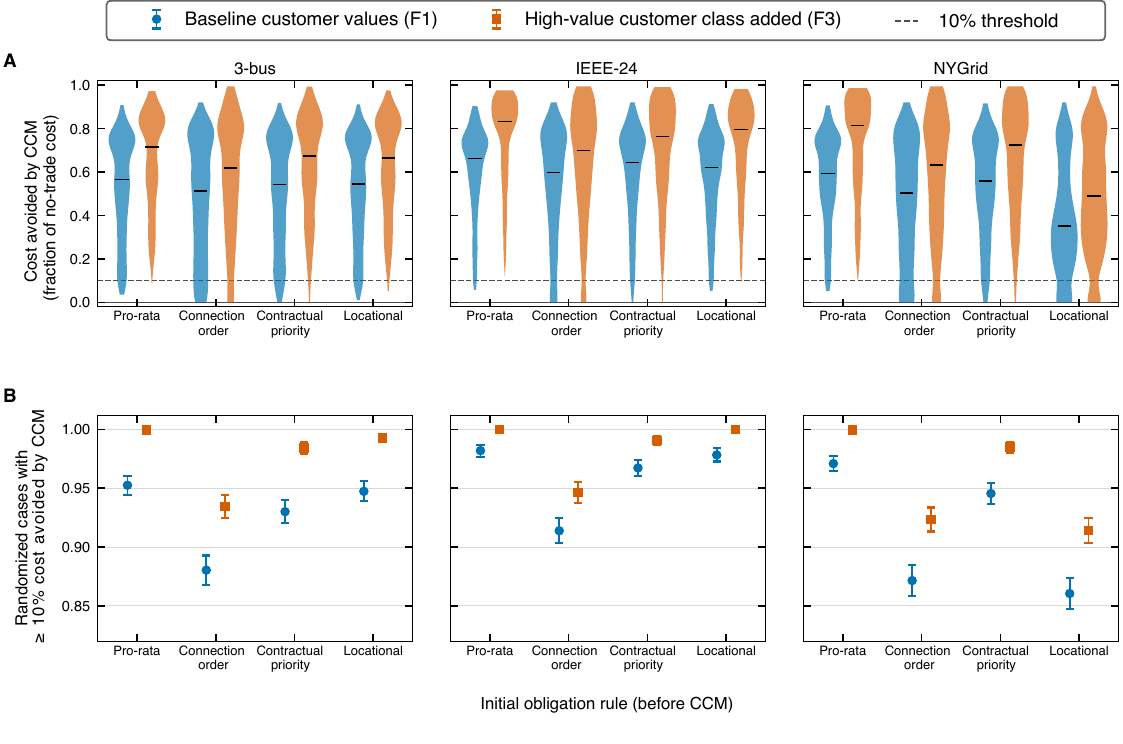}
\caption{Randomized CCM performance relative to no trading under each
  initial-obligation rule. Avoided interruption cost is the fraction of
  no-trade interruption cost removed by CCM. Blue denotes the baseline
  customer-value construction F1, and orange denotes F3, which adds a
  high-value small commercial and industrial class. Panel~(A) shows violin
  plots of avoided interruption cost. Wider sections of a violin indicate a
  higher concentration of scenario realizations. Black horizontal bars mark
  medians, and the dashed line marks the 10\% threshold. Panel~(B) reports the
  fraction of randomized cases meeting that threshold. Vertical bars extend
  two standard errors on either side of each estimate, and each point uses
  $n=2{,}718$ cases. Section~\ref{sec:randomized-design} defines the four
  upstream rules of pro-rata, connection order, contractual priority, and
  locational allocation.}
\label{fig:wp1-main}
\end{figure}

Every plotted combination in Panel~(A) has a positive median avoided
interruption cost. The medians range from 0.350 to 0.831. Panel~(B) shows that
trading removes at least 10\% of interruption cost in 86.1\% to 100\% of
scenario realizations. Together, the panels show that the welfare gains persist
across the tested randomized design and that the initial-obligation rule
affects their size. Welfare ratios are reported later to connect these
distributions with the fixed cases. The equal-VOLL controls are omitted from
the figure and yield zero welfare gain to numerical tolerance because
reallocating curtailment among equally valued customers cannot change total
interruption cost.

\subsubsection{Effect of the Initial-Obligation Rule}

The four upstream rules are pro-rata, connection order, contractual priority,
and locational allocation. They change how much inefficient curtailment remains
for CCM to reallocate. Panel~(A) of Figure~\ref{fig:wp1-main} shows that
pro-rata initial obligations produce the largest median gain from trading on
every network under both plotted customer-value constructions. This pattern
shows that CCM corrects more initial inefficiency under pro-rata assignment.

Panel~(A) also shows that the rule leaving the smallest median opportunity
depends on the network. On the 3-bus and IEEE 24-bus networks, connection order
produces the lowest median gain. On NYGrid, the locational rule produces the
lowest median, at 0.350 for F1 and 0.489 for F3. Panel~(B) shows the same
network-specific ordering for the share of randomized cases in which CCM
removes at least 10\% of interruption cost. NYGrid provides a physical
explanation in that its locational rule initially assigns more curtailment to
customers whose load reductions have greater modeled influence on the single
stressed interface. CCM therefore has less remaining beneficial reallocation
to perform across that interface.

Initial obligations affect the network constraint because CCM models flow
changes from the starting allocation. In
Equation~\eqref{eq:planner:flow}, the bus-level change is
$\Delta P_n=\sum_{i:\,n(i)=n}(c_i-\dbar_i)$. Thus, the same final curtailment
can create a different incremental flow when it begins from a different
$\dbar$. This comparison assumes that each starting point is feasible and
does not reconstruct the upstream power flow. It therefore explains how the
four rules affect trading from their assigned starting points, not which rule
produces the best absolute network state before trading.

\subsubsection{Customer VOLL and Estimation Error}

Panel~(A) of Figure~\ref{fig:wp1-main} compares F1 with F3 under the assumption
that declared VOLL equals the value assigned to each scenario realization.
Adding the high-value small commercial and industrial class in F3 increases
median avoided interruption cost in all 12 network and obligation-rule
combinations. Panel~(B) shows that F3 also increases the share of cases meeting
the 10\% threshold in all 12 combinations. The median increase over F1 ranges
from 0.033 to 0.129, and 82.9\% of paired scenario realizations increase. F1 and
F3 use the same random seeds, so the comparison removes independent sampling
variation. F3 also changes the customer shares, customer sizes, loads, and
values associated with the added class. The result therefore measures the
combined effect of including that class, whose VOLL is the survey estimate of
\$666{,}907/MWh.

\hl{F4 keeps the F3 customers and locations and applies a pre-specified VOLL
cap.} The \$150{,}000/MWh cap lowers avoided interruption cost in all 12
combinations, with median reductions from 0.008 to 0.081. F4
still produces a higher median gain than F1 in every combination. The gains
therefore increase when the high-value class is included even after its survey
estimate is capped. The paired distributions appear in
\ref{app:randomized}.

The next test allows declared values to differ from the values assigned to the
scenario realization. Figure~\ref{fig:wp1-error} reports the resulting welfare
loss as a fraction of the gain available under correct values. This
denominator measures how much of the benefit from trading is lost through
estimation error.

A common scaling error leaves the allocation unchanged because multiplying all
declared VOLL values by the same positive factor preserves their ordering and
the objective's trade-offs. Multipliers of 0.5, 0.8, 1.25, and 2.0 produce a
maximum welfare loss of $3.5\times10^{-9}$, the numerical floor. Relative
errors between customers are more consequential because they can change the
VOLL ranking that determines who curtails. Independent lognormal errors with $\sigma$ equal to 0.1,
0.3, and 0.5 lose 5.3\%, 17.2\%, and 35.0\% of the achievable gain at the 95th
percentile. Class-correlated error loses 22.2\% at the same percentile. These
losses show that errors in relative values can remove a substantial fraction
of the gains from trade.

A constructed rank-swap test targets the allocation margin by exchanging the
declared values of the nearest differently valued customers on opposite sides
of that margin. This test isolates the effect of a local ranking error where
the market separates curtailed from served customers. Its median loss is
5.0\%, while the 95th-percentile loss is
86.9\%. The loss exceeds 100\% in 3.4\% of draws, meaning that the resulting
allocation has lower welfare than no trade. For example, a loss of 120\% gives
up the entire gain from truthful trading and a further 20\% of that gain
relative to no trade. The swapped customers are
immediately adjacent in the value ranking in 20.3\% of solved draws, and the
median rank gap is two. The practical requirement is a reliable relative VOLL
ranking near the curtailment margin because a common multiplicative bias leaves
the allocation unchanged.

\begin{figure}[htbp!]
\centering
\includegraphics[width=\textwidth]{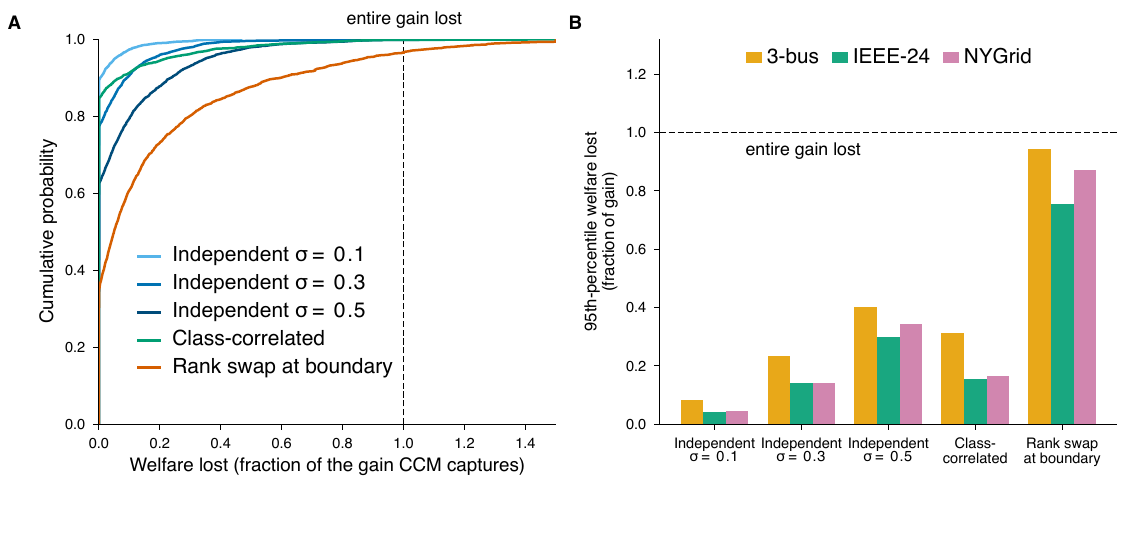}
\caption{Estimation errors that change relative VOLL rankings can erode or
  reverse the gains from trading. Welfare lost when the market clears on estimated values
  and is evaluated at the values assigned to each scenario realization,
  expressed as a fraction of the gain
  the market captures under truthful clearing. (A) Empirical cumulative
  distributions. The dashed line marks a loss equal to the entire gain.
  (B) 95th-percentile loss by testbed. Controls are omitted because
  truthful clearing and any common bias in declared values give zero loss
  to machine precision (maximum $3.5\times10^{-9}$). $n = 3{,}000$ draws
  per case, except the rank-swap case at 2,997.}
\label{fig:wp1-error}
\end{figure}

\subsubsection{Participant Count, Network Limits, and Fixed-Case Placement}

The gains persist when participant count increases beyond the five to eight
participants in the fixed cases. Figure~\ref{fig:wp1-sensitivity} summarizes
the participant-count and network-limit tests and then locates the fixed cases
within the randomized distributions. Pooled across all three networks, median
avoided interruption cost is 0.566 and 0.557, and the 10\% threshold is met in 98.1\%
and 97.4\% of draws. The two counts use separately generated populations, so
the 0.009 difference combines participant-count and sampling variation. The
near-equal medians support persistence at both counts without establishing a
monotonic trend. When the calculation is restricted to NYGrid, the
corresponding medians are 0.543 and
0.540, with threshold rates of 98.3\% and 98.0\%. The pooled medians exceed the
NYGrid-only medians by 0.023 and 0.017. Thus, the persistence result also holds
on NYGrid, although its median gains are slightly lower than the cross-network
aggregate.

The NYGrid reference case uses a 1,290~MW IF\_11 southbound limit. We sweep this
limit from 200 to 3,000~MW to test how the binding bottleneck affects the
result. With every other directional interface limit fixed, the pooled median
avoided cost rises from 0.466 to 0.620 and equals 0.562 at the reference-case
limit. Comparing the two endpoints across the same 500 scenario realizations,
the 200~MW limit lowers cleared welfare in 138 pro-rata, 106 connection-order,
121 contractual-priority, and 407 locational cases. The locational rule is most
sensitive because it initially assigns curtailment at Northport, the only proxy
bus with nonzero IF\_11 sensitivity; relieving this high-VOLL participant then
uses southbound transfer capacity.

Panel~(C) of Figure~\ref{fig:wp1-sensitivity} places the selected fixed cases
within 50 draws of their queue and contract inputs. Median avoided costs span
0.227 to 0.851. Depending on the network and initial-obligation rule, 94\% to
100\% of draws exceed the 10\% threshold. The fixed cases therefore fall within
the tested avoided-cost distributions. Under pro-rata initial obligations,
the randomized median ratios range from
$1.16\times$ to $1.26\times$. The fixed ratios of $1.41\times$--$1.83\times$ lie between the 92nd and
100th percentiles of the F1 distribution, which characterizes them as
attainable high-gain realizations within the design.

\begin{figure}[htbp!]
\centering
\includegraphics[width=\textwidth]{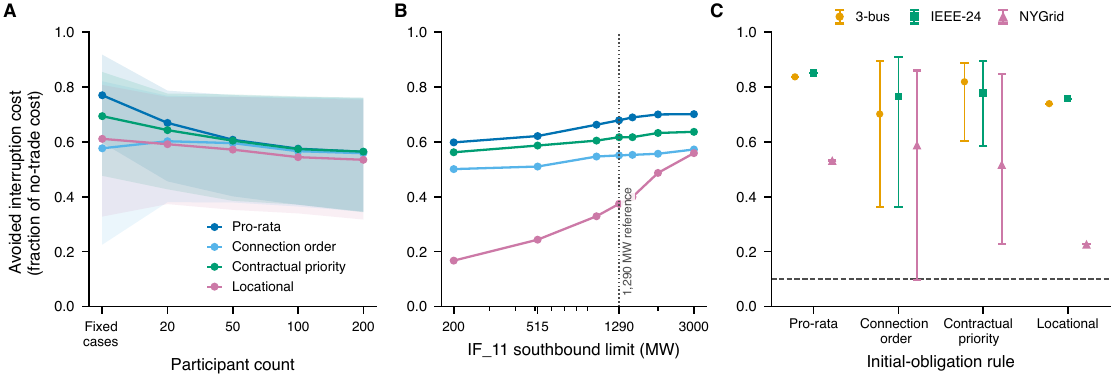}
\caption{The gains persist at 100 and 200 participants and across the tested
  IF\_11 southbound-limit range. (A) Median avoided interruption cost against
  participant count, pooled over testbeds. The fixed-case point combines the
  original testbed populations of five to eight participants. Shaded bands show
  interquartile ranges. (B) Median avoided cost against the NYGrid
  IF\_11 southbound limit, with every other directional interface limit held
  at its reference-case value. Panel~(B) varies the same limit as
  Figure~\ref{fig:linesweep-nygrid}, but reports medians across 500 randomized
  scenario realizations per rule and limit rather than fixed-case welfare. The
  dotted line marks the 1,290~MW reference value. (C) Avoided cost
  at each testbed's fixed case over
  50 draws of the queue and contract inputs. Markers show medians and
  bars show the 5th to 95th percentile range. The randomized inputs are the
  connection queue and the firm or non-firm contractual tier. Pro-rata and the
  locational rule use neither input, so their range is zero by construction.}
\label{fig:wp1-sensitivity}
\end{figure}

Each synthetic draw receives equal weight because operating data do not show
how often its combination of inputs occurs. The reported frequencies therefore
characterize the randomized design. The NYGrid capacity
multiplier and IF\_11 sweep are synthetic stress tests, and additional
customers do not add sub-zonal network detail. The evidence thus establishes
persistence across the tested static cases and assigns no real-world
probabilities to them. The next section examines a different boundary by
repeating the clearing across several hours.

\subsection{Duration, Capability, and Repeated Clearing}
\label{sec:multiperiod}

Figure~\ref{fig:wp3-duration} compares repeated hourly clearing on January 21
and 22, 2025 for episodes lasting one to six hours. Each hour uses its observed
load and a duration weight, with decisions cleared independently across hours.
Seller capability is the controlling factor in these cases. Changing the
capability bound changes avoided interruption cost by tens of percentage
points, whereas extending the episode changes it by less than one percentage
point.

As seen in Figure~\ref{fig:wp3-duration}, the capability multiplier $\kappa_C$
caps a customer's final curtailment at a multiple of its initial obligation.
We include this limit to prevent clearing
from concentrating all available curtailment on the lowest-VOLL customers when
their loads allow it. At $\kappa_C=1$, no customer can accept curtailment above
its initial obligation, so system balance leaves no room for trade. Values of
two and four allow progressively more reallocation, and
the original full-load bound corresponds to about $\kappa_C=26.35$ at the
central requirement. For the six-hour January 21 episode, CCM reduces
interruption cost by 0\% at $\kappa_C=1$, 81.31\% at $\kappa_C=2$, 95.23\% at
$\kappa_C=4$, and 97.91\% under the full-load bound. Holding
$\kappa_C=2$, extending the episode from one to six hours changes the reduction
only from 80.83\% to 81.31\%.

\begin{figure}[htbp!]
\centering
\includegraphics[width=\textwidth]{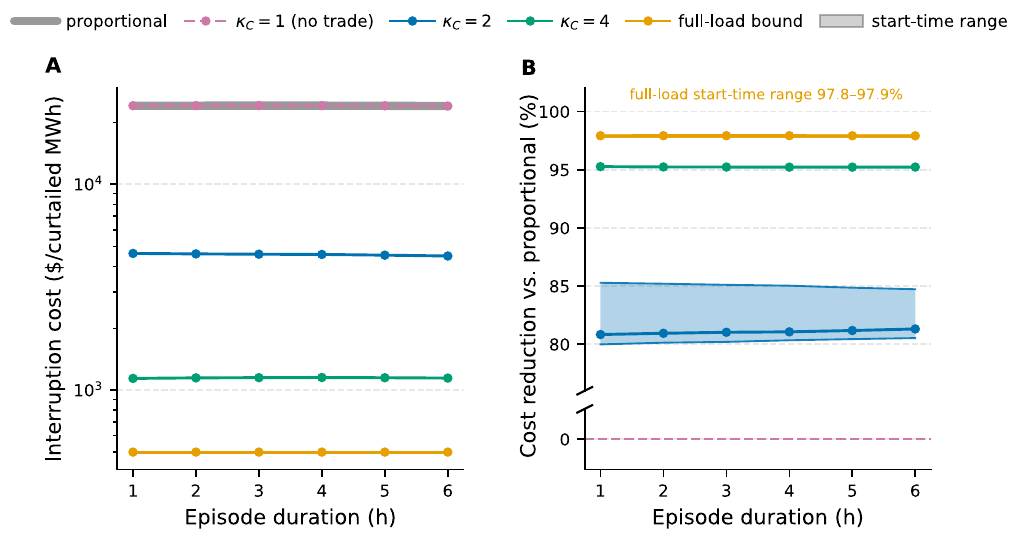}
\caption{Curtailment capability determines the benefit of repeated hourly
  clearing. Episodes last one to six consecutive hours on January 21, 2025 and
  are placed by their start hour. (A)~Interruption cost per curtailed MWh on a
  logarithmic axis. (B)~Reduction in interruption cost relative to
  proportional assignment, with a vertical axis break separating the
  no-trade line at zero from the remaining results. The capability multiplier
  $\kappa_C$ caps each customer's curtailment at $\kappa_C$ times its initial
  obligation, the full-load bound caps it at the customer's full load, and
  proportional assignment divides the total requirement in proportion to
  load. Lines start each episode at 16:00, the start of the historical
  demand-response call. Shaded bands span the minimum and maximum over all
  feasible start hours ($n=129$ combinations of episode length and start
  hour).}
\label{fig:wp3-duration}
\end{figure}

The capability limit also determines how concentrated the final curtailment
becomes.
Over the six-hour January 21 episode, the most exposed customer bears 39.1\%
of curtailed energy under proportional assignment and 43.4\% at
$\kappa_C=2$. Under the full-load bound, one customer absorbs 100\% of the
requirement in every hour. The 97.91\% cost reduction therefore depends on an
extreme concentration of delivery. The detailed customer allocations and
hourly load profiles appear in \ref{app:multihour}.

Changing the tested duration-cost shape, which applies a sequence of hourly
weights to each customer's first-hour VOLL, does not change the allocation. Across
480 comparable customer cases, flat, decreasing, and increasing shapes select
the same customers. At $\kappa_C=2$ over six hours, their cost reductions are
81.31\%, 81.19\%, and 81.35\%. The closest weighted-value ratio in these cases
is 1.51, while reversing the customer order requires 2.5. These value gaps are
wide enough that none of the tested duration weights reverse the customer
ranking. Duration-dependent costs will affect the allocation only when the
weights change which customers have the lowest weighted VOLL.

Network capacity matters only when the curtailment requirement becomes much
larger than the observed calls or the interface becomes much tighter. IF\_11
first binds when the requirement reaches 18.5\% of the system load in that
hour. Before that threshold, the full-load-bound reduction falls slightly from
97.91\% at a 10\% requirement to 97.36\% at 15\% because the lowest-VOLL
customer's load becomes fully curtailed and the remaining requirement moves to
higher-VOLL customers. The historical NYISO demand-response call distribution
has a median of 3.795\% and a maximum of 5.627\%.
At $\kappa_C=2$, the interruption-cost reduction falls from 81.31\% before
IF\_11 binds to 48.58\% at the fixed-case depth of 43.653\%. Holding the
requirement fixed produces the same pattern when remaining IF\_11 capacity
falls below the analytically derived 264~MW binding threshold. These flows measure only the change from proportional
initial obligations and exclude the interface loading created by hourly
generation dispatch. The binding thresholds therefore describe incremental
headroom available to CCM and do not establish whether IF\_11 is uncongested
in the full system dispatch.

Repeat-call data show that certified capability can change from one day to the
next. Across 25 pairs of New York calls on consecutive days, the delivered
share of requested response is a median of 1.3 percentage points lower on the
second day. The change ranges from $-28.48$ to $+8.80$ percentage points, and
five pairs decline by at least ten points. This sample is too small to estimate
a duration-cost curve. Its dispersion nevertheless supports updating each
participant's certified capability after delivery in a repeated market.

The multi-hour evidence supports repeated application of the single-period
clearing under the tested loads, duration weights, and capability bounds. A
horizon-coupled mechanism would additionally link decisions across hours and
represent uncertain future requirements, delivery, and strategic bidding over
time. Those dynamic allocation and incentive properties remain open. The
single-period incentive tests below therefore return to the truthful-bidding
assumption used in the welfare analysis.

\subsection{Incentive Analysis}\label{sec:incentive-results}

The welfare results assume truthful bids, so we next measure whether an agent
can gain by changing its bid alone. Section~\ref{sec:randomized-design}
defines this finite-grid deviation test and the gain ratio in
Equation~\eqref{eq:bid-gain-ratio}. VCG is the incentive benchmark because
Proposition~\ref{prop:vcg} makes truthful bidding a dominant strategy under a
fixed, report-independent tie-breaking rule. Uniform pricing and locational
credit pricing lack that guarantee. Because the search uses a finite bid grid,
any profitable deviation found is evidence of manipulation, while the largest
measured gain remains a lower bound on possible exposure.

Figure~\ref{fig:bid-deviation} shows how VPP-1's payoff changes as its bid moves
away from its true VOLL in the 3-bus case. Under VCG, the payoff remains at its
truthful level throughout the tested range, consistent with
Proposition~\ref{prop:vcg}. Under uniform pricing, VPP-1 can raise the price
applied to its cleared credits while its allocation remains fixed. It raises
its bid to $4.94$ times its true VOLL, remains a 10~MW net seller, and increases
its payoff by 197.1 percent. At larger bids, its allocation switches from
seller to buyer and the deviation becomes unprofitable. This allocation change
marks the end of the profitable deviation region in the figure.

\begin{figure}[htbp!]
\centering
\includegraphics[width=0.6\textwidth]{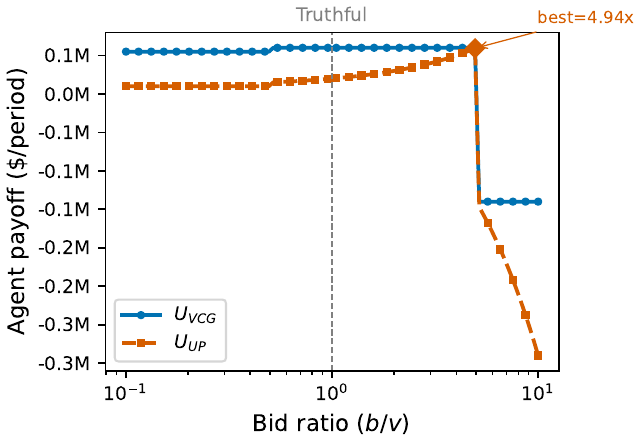}
\caption{VPP-1's payoff under unilateral bid deviations on the 3-bus
  network. The horizontal axis is the bid divided by true VOLL. While VPP-1
  remains a 10~MW net seller, its VCG payoff stays at \$60,000 and its
  uniform-price payoff rises with its bid. The largest measured
  uniform-price gain is 197.1 percent at $b/v\approx4.9$. At larger bids,
  VPP-1 becomes a 40~MW buyer and both payoffs fall.}
\label{fig:bid-deviation}
\end{figure}

The IEEE 24-bus results preserve VCG's zero measured gain and show smaller
exposure under the other rules. The largest gain is 14.6\% under uniform
pricing and 4.4\% under locational pricing. Locational pricing leaves the
3-bus maximum unchanged at 197.1\% because VPP-1 is at the reference bus,
where the uniform and locational prices coincide.

The NYGrid extension measures all three payment rules for the five zonal-proxy
participants. VCG yields no measured gain. On the 31-point primary grid, the
Rotterdam proxy increases its payoff by 162.4\% under both uniform and
locational pricing at a bid 6.21 times its true VOLL. Its cleared curtailment
remains 2{,}093.18~MW. Its local price equals the system balance price along
this deviation, which rises from \$3{,}000/MWh to \$18{,}630.51/MWh. The
settlement receipt consequently rises from \$1.80M to \$11.19M under either
rule. The 100-point sensitivity places the best tested bid at 6.55 times true
VOLL and raises both finite-grid lower bounds to 173.0\%. Locational pricing
also gives the Northport proxy a 34.4\% gain at one-tenth of its true VOLL,
whereas its uniform-price gain remains zero. The other three participants have
no profitable uniform or locational deviation on either grid. Thus,
location-specific settlement changes which participants can benefit from a
misreport but does not uniformly reduce manipulation exposure. The participant
results and the network-level randomized settlement summary appear in
\ref{app:settlement}. These incentive findings next need to be
considered together with the funding and participation properties of each
payment rule.

\subsection{Payment-Rule Comparison}\label{sec:payment-rules}

No tested rule combines dominant-strategy truthful bidding, exact budget
balance, and all measured participation and implementation properties.
Table~\ref{tab:payment-rule-comparison} separates these criteria so that the
tradeoffs are made clear.

VCG provides the truthful-bidding benchmark but may require external funding.
It runs deficits of \$115,000, \$248,868, and \$10.208 million at the three
fixed cases. These deficits equal 0.87, 1.83, and 3.48 percent of
cleared welfare, and 3.49, 3.73, and 11.92 percent of gains from trade. Across
18,000 randomized rule-level cases, VCG has a material deficit in 37.6 percent
of cases. Among deficit cases, median deficits range from 2.6 to 4.3 percent
of cleared welfare across networks. Congestion is neither required nor
sufficient for a deficit.

Uniform pricing removes the funding shortfall because buyer payments exactly
equal seller receipts. This budget property comes with measured deviation
gains of 197.1\% on the 3-bus network, 14.6\% on the IEEE 24-bus network, and
173.0\% on the finer NYGrid grid. The corresponding 31-point NYGrid value is
162.4\%.
Under congestion, the rule can also leave a participant below its no-trade
surplus. Budget balance therefore does not provide either truthful bidding or
participant protection.

Locational pricing addresses the funding and participation problems within the
stated single-event linear program. Section~\ref{sec:properties} shows that its
settlement residual is a non-negative congestion rent and that truthful
participants receive at least their no-trade surplus. The revenue identity
holds in all 18,000 randomized cases. Fixed-point residuals are \$480,000,
\$425,465, and \$60.630 million. The deviation tests still find a 4.4\%
maximum gain on the IEEE 24-bus case and the unchanged 197.1\% exposure at the
3-bus reference bus. On NYGrid, the maximum is 162.4\% on the primary grid and
173.0\% on the finer grid, equal to the uniform-price result for the Rotterdam
proxy. Locational pricing is therefore revenue adequate in the tested
formulation but remains strategically vulnerable.

\begin{table}[ht]
\centering
\scriptsize
\caption{Payment-rule evidence. Fixed-point net revenue is listed for the
  3-bus, IEEE 24-bus, and NYGrid cases. Deviation gains are the largest
  values found across the declared primary and finer finite bid grids and are
  listed in the same network order. The matched 31-point NYGrid uniform and
  locational gains are both 162.4\%; the 100-point sensitivity raises both to
  173.0\%. These values are lower bounds on manipulation exposure. Revenue
  adequacy means that a rule does not require an external settlement subsidy.}
\label{tab:payment-rule-comparison}
\begin{tabular}{@{}>{\raggedright\arraybackslash}p{3.3cm}>{\raggedright\arraybackslash}p{3.7cm}>{\raggedright\arraybackslash}p{3.5cm}>{\raggedright\arraybackslash}p{4.0cm}@{}}
\toprule
Property & VCG & Uniform price & Locational credit price \\
\midrule
Truthful bidding
  & dominant strategy
  & not guaranteed
  & not guaranteed \\
Budget property
  & may require subsidy
  & exactly balanced
  & revenue adequate with congestion residual \\
Fixed-point net revenue
  & $-\$0.115$M / $-\$0.249$M / $-\$10.208$M
  & \$0 / \$0 / \$0
  & $+\$0.480$M / $+\$0.425$M / $+\$60.630$M \\
Randomized settlement incidence
  & deficit in 37.6\% of 18,000 cases
  & balanced in 100\%
  & revenue adequate in 100\% \\
Largest measured gain
  & 0\% / 0\% / 0\%
  & 197.1\% / 14.6\% / 173.0\%
  & 197.1\% / 4.4\% / 173.0\% \\
No-trade surplus protection under truthful bids
  & guaranteed by Proposition~\ref{prop:ir}
  & guaranteed only without congestion
  & guaranteed within the stated single-event LP \\
Operational programs per event
  & $N+1$
  & 1
  & 1 \\
Unresolved implementation condition
  & deficit funding
  & manipulation monitoring
  & dual selection, residual allocation, and manipulation monitoring \\
\bottomrule
\end{tabular}
\end{table}

Among the three rules, locational credit pricing is the conditional candidate
for further evaluation when a design requires network-specific prices and no
external settlement subsidy. Its measured manipulation exposure would still
require monitoring or another strategic safeguard. The NYGrid result reinforces
this condition because locational pricing leaves Rotterdam's measured exposure
unchanged and creates a profitable tested deviation for Northport. VCG remains the
truthful-bidding benchmark when external funding and $N+1$ solves are
acceptable. Uniform pricing remains the exact-budget benchmark for
uncongested cases.

Locational settlement also leaves two tariff choices unresolved. The tariff
must select one price when the clearing linear program has multiple valid dual
solutions, and it must specify who receives the congestion residual. The
revenue identity determines the residual's sign but not its recipient. The
network-level randomized settlement and participant-level deviation evidence
appears in \ref{app:settlement}. These unresolved choices
connect the numerical comparison to the implementation requirements.

\subsection{Computational Performance}\label{sec:computational}

Figure~\ref{fig:wp1-scaling} compares three computational tasks. Operational
clearing solves one linear program to allocate curtailment. VCG settlement
repeats that program with each participant removed. \hl{The offline
bid-deviation analysis supplies one candidate bid at a time and solves a
fixed-bid MILP to recover the operator's clearing decision. The MILP evaluates
the supplied bid, and the outer bid grid performs the search.} Only operational
clearing is required during an event.

Because operational clearing determines the event allocation, its runtime is
the relevant event-time measure. All timing runs use one solver thread on an
AMD Ryzen Threadripper PRO 3995WX processor with Gurobi 13.0.1 and
Python 3.13 under Linux. At 200
customers, the solver takes 41~ms and model construction takes 80~ms. VCG
settlement requires $N+1$ independent linear-program solves and takes 8.41~s
of total solver time at 200 customers. These results cover participant counts
up to 200 on fixed test-network topologies. They do not measure scaling with
the number of buses, lines, or monitored interfaces.

\hl{The fixed-bid MILP used in that offline analysis grows more quickly.} At
200 customers, it contains 40,276 binary variables and 120,867 constraints.
\hl{One fixed-bid evaluation takes 10.2~s to construct and 18.6~s to solve. A
complete grid study repeats this evaluation for every candidate bid and
participant, so 18.6~s is not the runtime of the complete bid search.} The
quadratic term comes from the
$\card{\NN}(\card{\NN}-1)$ bilateral credit variables in
Equation~\eqref{eq:bincount}. This mixed-integer program is a research
instrument and is not part of operational clearing or settlement.

\begin{figure}[htbp!]
\centering
\includegraphics[width=\textwidth]{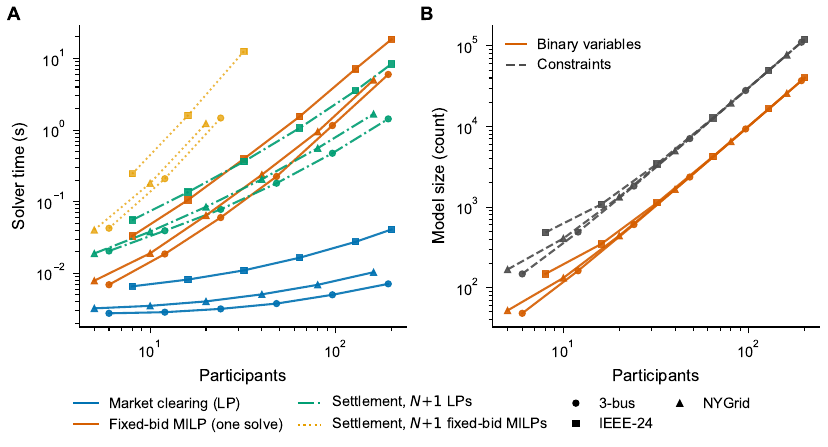}
\caption{On fixed network topologies at 200 participants, operational clearing
  takes 41~ms, compared with 8.41~s for VCG settlement and \hl{18.6~s for one
  fixed-bid MILP evaluation}. Panel~(A) separates solver time for one operational clearing, VCG
  settlement using $N+1$ operational linear programs, and \hl{one
  MILP evaluation}. The additional mixed-integer settlement curve shows the cost of
  using the \hl{MILP} for the same leave-one-out allocations.
  Panel~(B) reports model size. Clearing and \hl{MILP} values are
  medians of 20 one-thread repeats. Linear-program settlement values are
  medians of three complete $N+1$ passes. All measured solves finish within
  300~s.}
\label{fig:wp1-scaling}
\end{figure}

The figure supports participant-count scaling claims only. The operational
program adds flow constraints with monitored network elements. The strategic
program grows quadratically with participants and linearly with monitored
elements. These expressions describe how formulation size changes.
Measured network-size runtimes remain untested. We report the fixed-case Big-M
sensitivity check in \ref{app:bigm}. We derive the quadratic participant-count
term in \ref{app:bincount}.

\subsection{Limitations and Future Work}\label{sec:limitations}

The results establish efficient truthful allocation and quantify its value
within the declared experiments. Their interpretation remains bounded by the
scenario design, VOLL model, repeated-hour formulation, settlement rules, and
delivery assumptions.

The randomized scenarios are equally weighted because the study has no
observed frequencies for their joint input combinations. The reported
percentages therefore describe only the tested design and cannot be read as
estimated real-world probabilities. The fixed welfare ratios are high-gain
realizations within those ranges. Increasing the participant count improves
customer heterogeneity. The network remains at the same zonal or bus
resolution. The DC power-flow model also
omits voltage, reactive power, and nonlinear losses. Empirical occurrence
weights, sub-zonal congestion, AC feasibility, and measured network-size
scaling are therefore needed before extending the numerical frequencies to an
operational system.

The VOLL model assigns each participant a constant marginal VOLL within
an event. Published interruption values vary by customer, method, notice, and
duration~\cite{schroder_voll_2015,charles_gibbons_value_2024}. The tested flat,
decreasing, and increasing duration weights leave the allocation unchanged
because they never reverse the customer-value ordering. This null result does
not cover closer values, steeper
curves, or non-convex costs such as restart penalties. Piecewise-linear
concave bids could preserve linear clearing with more variables. Non-convex
costs may require uplift or a different settlement design because linear
prices need not support the efficient allocation
\cite{oneill_nonconvexities_2005}.

The repeated-hour study adds duration weights but leaves each hourly clearing
independent. A dynamic mechanism would need to represent horizon coupling,
uncertain future requirements, load recovery, storage state, prior delivery,
and strategic bidding across periods. It would also need to limit repeated
concentration on the same low-value participants. Existing dynamic mechanism
results use incentive concepts designed for these intertemporal links
\cite{bergemann_dynamic_2010,athey_efficient_2013}.

Delivery adds a further distinction between compliance and capability. A
metered service ceiling can verify compliance with the final obligation, while
additionality and future capability require separate audits and performance
records. Double-counting controls, capability updates, and non-performance
charges also change participant payoffs. The single-event VCG result does not
cover these delivery risks.

The payment-rule evidence leaves an unresolved design trade-off. VCG supplies
the dominant-strategy benchmark and may require external funding. Uniform
pricing balances transfers and permits measured manipulation. Locational
credit pricing is revenue adequate in the stated LP and remains manipulable,
including the unchanged 3-bus exposure. Its use would require a canonical
dual-price rule, an assignment for the congestion residual, and a monitoring
and enforcement design.

\FloatBarrier

\section{Conclusion}\label{sec:conclusion}

This paper presents the network-constrained Curtailment Credit Market (CCM),
which reallocates curtailment obligations among participating loads after an
upstream process establishes the inputs for an event. Before CCM clearing,
that process sets the total curtailment requirement, identifies eligible
loads, and fixes each load's initial obligation and approved load reference.

The paper's structural contribution is feasible-set
equivalence under the direct-current power-transfer-distribution-factor
(DC/PTDF) model. Every planner-feasible final curtailment allocation can be
implemented through non-negative bilateral credit flows because network
feasibility depends on final curtailment levels and is independent of the
contractual trade path. The CCM and planner feasible sets therefore coincide,
$\Cccm=\Cplan$, under the conditions in Section~\ref{sec:properties}.

The welfare, incentive, and participation statements require additional
conditions. Truthful VOLL bids make the CCM choose the planner's
welfare-maximizing allocation. We use the classical VCG truthfulness result as
an incentive benchmark. Under the fixed, report-independent tie-breaking rule
in Proposition~\ref{prop:vcg}, VCG Clarke-pivot transfers make truthful bidding
a dominant strategy. Under truthful bids, Clarke-pivot transfers, feasible
initial obligations, and the stated leave-one-out convention, every
participant receives at least its no-trade surplus under
Proposition~\ref{prop:ir}. Uniform-price settlement can
violate this participation condition when transmission binds. Under the valid
Big-M bounds in Assumption~\ref{ass:bigm}, the lower-level clearing problem has
an allocation-equivalent mixed-integer linear program (MILP). This
reformulation \hl{supports an offline bid-deviation analysis that supplies
candidate bids externally and re-solves the fixed-bid MILP}. Operational CCM
clearing remains a linear program.

In the three fixed cases, truthful CCM clearing matches
planner welfare to numerical tolerance and produces welfare
$1.41$--$1.83$ times the pro-rata level. These ratios describe
high-gain realizations within the tested design. Under the survey-based
baseline customer-value construction, the fixed-case ratios fall between the
92nd and 100th percentiles of comparable cases, whose median ratios are
$1.16\times$--$1.26\times$.
The fixed ratios illustrate potential gains in the studied instances and are
bounded by a design whose scenario weights do not represent real-world
occurrence probabilities.

Across the 24 combinations of three networks, four initial-obligation rules,
and two headline customer-value constructions (F1 and F3), median avoided interruption cost ranges from
0.350 to 0.831. In every combination, trading removes at least 10\% of
interruption cost in at least 86\% of the equally weighted scenario
realizations. Separate populations
of 100 and 200 participants produce pooled medians of 0.566 and 0.557 on the
fixed test-network topologies. These results support persistence with greater
participant granularity within the randomized design. Real-world occurrence
probabilities and network-size scaling were not evaluated.

The value sensitivities identify two distinct effects. Adding the survey's
high-value small commercial and industrial class changes customer shares,
sizes, loads, and values, and raises median gains in all 12 network and
obligation-rule combinations. Holding that population fixed and capping only
VOLL at the \hl{sensitivity cap of} \$150,000/MWh lowers the median in all 12
combinations. Independent customer-level estimation errors lose 5.3\% to
35.0\% of the achievable gain at the 95th percentile as error increases. A
constructed rank swap leaves the market below its no-trade welfare in 3.4\%
of solved draws. The evidence therefore makes accurate relative VOLL estimates
near the allocation margin an implementation concern.

In the repeated-hour NYGrid cases, customer capability changes avoided
interruption cost by tens of percentage points, whereas extending an episode
from one to six hours changes it by less than one percentage point. The tested
flat, decreasing, and increasing duration weights never change the allocation.
Under the full-load capability bound, one customer absorbs the entire
curtailment requirement in every hour of the six-hour primary episode. IF\_11
is the modeled transmission corridor into Long Island. It first binds when the
requirement reaches 18.5\% of system load in that hour, above the 5.627\%
maximum in the historical call sample. At 15\%, the earlier small reduction
under the full-load bound reflects exhaustion of the lowest-VOLL customer's
load while IF\_11 remains slack. This threshold covers incremental CCM
flow and excludes interface loading from generation dispatch. The experiment
therefore supports repeated application of single-period clearing under the
tested inputs.

The payment tests leave a three-way design tradeoff. VCG provides
dominant-strategy truthful bidding. It has a material deficit in 37.6\% of
18,000 randomized cases, where each case pairs one scenario realization with
one initial-obligation rule. Uniform pricing balances buyer payments and seller
receipts. Its largest measured unilateral gains are 197.1\%, 14.6\%, and
173.0\% on the 3-bus, IEEE 24-bus, and NYGrid cases. Locational
credit pricing is revenue adequate in all 18,000 tested cases. Its measured
gains remain 197.1\% on the 3-bus network and 4.4\% on IEEE 24-bus. On NYGrid,
the 31-point primary grid gives a 162.4\% maximum and the 100-point sensitivity
raises the finite-grid lower bound to 173.0\%. The rule draws on locational
marginal pricing in wholesale electricity markets, where network congestion
creates different marginal energy prices at different buses~\cite{schweppe_spot_1988}.
Locational pricing is therefore a conditional candidate for further evaluation
when network-specific prices and the absence of an external settlement subsidy
are required.

The timing results support fast operational clearing on the fixed test-network
topologies. At 200 participants, one operational clearing takes 41~ms of solver
time and 80~ms of model-construction time. VCG settlement takes 8.41~s of total
solver time for $N+1$ linear programs. \hl{One fixed-bid MILP evaluation used
by the offline bid-deviation analysis takes 18.6~s to solve after 10.2~s of
construction; the full bid-grid search repeats that evaluation at every
candidate point.} These experiments quantify
participant-count scaling on the tested topologies. Network-size scaling across
buses, lines, and monitored interfaces remains for future study.

\paragraph{Future work} Although settlement design remains unresolved, we identify
the decisions
needed to advance locational credit pricing. These comprise a canonical rule
for selecting among valid dual prices, an assignment for the congestion
residual, and safeguards against strategic manipulation. A repeated market
must also connect current delivery and settlement outcomes to future
eligibility, capability, and bidding.

Delivery requires separate rules for compliance, additionality, and future
capability. An operator-approved load reference and an absolute post-clearing
service ceiling can verify event compliance. Prequalification, periodic
reference updates, capability derating, collateral, and non-performance
charges are still needed to prevent overstated or undeliverable flexibility.
These rules can change participant payoffs and therefore require a renewed
incentive analysis.

The network and time representations require two further extensions. An
operational network study must combine CCM's incremental flows with the
upstream dispatch state, alternating-current feasibility, and sub-zonal
constraints. The bilateral accounting decomposition remains valid whenever
the planner and CCM impose the same feasibility projection over final
curtailment allocations. Alternating-current network states, prices, and the
MILP reformulation require new models. A horizon-coupled mechanism must
also represent uncertain requirements, load recovery, storage state, prior
delivery, and strategic bidding across periods. These extensions define the
evidence needed to move from a single-event allocation mechanism to an
operational market design.


\section*{Acknowledgements}
This work was supported by gifts to Environmental Defense Fund and Cornell
Atkinson Center for Sustainability from the David and Patricia Atkinson
Foundation, and in part by the Indonesia Endowment Fund for Education (LPDP).
The authors thank Reynold Li for helpful discussions.

\section*{CRediT authorship contribution statement}
\textbf{Richard Mahuze:} Conceptualization, Methodology, Software, Formal
analysis, Investigation, Data curation, Visualization, Writing -- original
draft. \textbf{Charlotte Gressel:} Methodology, Investigation, Writing --
review \& editing. \textbf{Ali Amadeh:} Methodology, Investigation, Writing
-- review \& editing. \textbf{K. Max Zhang:} Conceptualization, Supervision,
Project administration, Funding acquisition, Writing -- review \& editing.

\section*{Declaration of competing interest}
The authors declare that they have no known competing financial interests or
personal relationships that could have appeared to influence the work
reported in this paper.

\section*{Data availability}
The simulation code, scenario inputs, and result files that support the
findings of this study, including the randomized-study campaign outputs and
the timing benchmarks, are archived by the authors and are available from
the corresponding author upon reasonable request.

\clearpage
\appendix
\numberwithin{equation}{section}
\makeatletter
\@addtoreset{table}{section}
\@addtoreset{figure}{section}
\makeatother
\renewcommand{\theequation}{\Alph{section}.\arabic{equation}}
\renewcommand{\thetable}{\Alph{section}.\arabic{table}}
\renewcommand{\thefigure}{\Alph{section}.\arabic{figure}}

\section{Proofs of Main Results}\label{app:proofs}

\begin{proof}[Proof of Lemma~\ref{lem:wellposed}]
The no-trade allocation $(x, c, s) = (0,\,\dbar,\,L - \dbar)$
satisfies every constraint in~\eqref{eq:ccm}: the equality
constraints hold by substitution, system balance and agent bounds
hold by the instance assumptions $\sum_i \dbar_i = D$ and
$0 \le \dbar_i \le L_i$, the PTDF constraints hold because
$\Delta P_n = 0$ at every bus and the directional margins give
$-\Fbar_\ell^-\le0\le\Fbar_\ell^+$, and nonnegativity holds
trivially.  Because $0 \le s_i \le L_i$ and $b_i \ge 0$,
the objective satisfies
$\sum_i b_i s_i \le \sum_i b_i L_i < \infty$, so the objective is
bounded above.  Since the feasible region is a nonempty closed
polyhedron, a linear objective bounded above on it attains its
maximum, so an optimal solution exists~\cite{bertsimas_introduction_1997}.
\end{proof}

\begin{proof}[Proof of Lemma~\ref{lem:reduced}]
The equality constraints \eqref{eq:ccm:curtbal}
and~\eqref{eq:ccm:served} define $c$ and $s$ as affine functions
of~$x$.  Substituting into the objective and the remaining
inequality constraints yields~\eqref{eq:reduced}.  For the
balance constraint:
$\sum_i c_i(x) = \sum_i \dbar_i
  - \sum_i\sum_{j\ne i} x_{ji}
  + \sum_i\sum_{j\ne i} x_{ij} = D$,
where reindexing $(i,j)\mapsto(j,i)$ shows that the incoming and outgoing
double sums cover the same ordered pairs.
\end{proof}

\begin{proof}[Proof of Proposition~\ref{prop:vcg}]
The CCM clearing LP maximizes $\sum_i b_i\,s_i$ over the
feasible set~$\Cccm$, which depends only on the exogenous
parameters $(\dbar, L, \Phi, \Fbar)$ and not on the bid
vector~$b$.  Each agent's valuation $v_i\,s_i$ depends only on
its own allocation~$s_i$.  These two properties satisfy the
hypotheses of Nisan's Theorem~9.17~\cite{nisan_algorithmic_2007}, where the
allocation rule maximizes declared welfare over a bid-independent
feasible set, and valuations are private.  A fixed,
report-independent tie-breaking rule among LP optima (e.g.,
lexicographic vertex selection) completes the preconditions,
so VCG Clarke-pivot payments make truthful reporting a dominant
strategy for every agent.  The tie-breaking formalism and its
connection to the MPEC literature are detailed in \ref{app:tiebreak}.
\end{proof}

\subsection*{Proof of Proposition~\ref{prop:ir}}

\emph{Feasibility.} For any agent subset $S$ with requirement
$R_S = \sum_{j \in S} \dbar_j$, the no-trade point ($x = 0$,
$c_j = \dbar_j$, $s_j = L_j - \dbar_j$) is feasible: per-agent balance and
the total balance hold by construction, $0 \le \dbar_j \le L_j$ gives the
bounds, and every injection deviation is zero, so each line constraint
reads $-\Fbar^{-}_\ell \le 0 \le \Fbar^{+}_\ell$, which holds because
$\Fbar \ge 0$. Taking $S = \NN$ covers the full problem; taking
$S = \NN \setminus \{i\}$ covers the leave-one-out problem, because its
requirement $D - \dbar_i$ equals the remaining agents' total obligation.

\emph{Attainment.} The objective is bounded above by $\sum_j v_j L_j$
since $0 \le s_j \le L_j$. The feasible set is unbounded in the credit
variables (adding $t$ to both $x_{ij}$ and $x_{ji}$ changes no
curtailment), so attainment is argued on the projection onto $(c, s)$: the
projection of a polyhedron is a polyhedron, hence closed; it is bounded by
the box constraints; a linear function on a nonempty compact set attains
its maximum, and any maximizing $(c, s)$ lifts back to a feasible
$(x, c, s)$.

\emph{Selection independence.} Substituting
Eq.~\eqref{eq:vcg-payment} into $S_i = v_i s_i^{*} - p_i$ gives
$S_i = \sum_j v_j s_j^{*} - \sum_{j \ne i} v_j s_j^{(-i)}
     = W_{\mathrm{CCM}} - W^{(-i)}$,
a difference of optimal \emph{values}, which no tie-breaking rule can
move.

\emph{Participation.} Let $(x', c', s')$ be any optimal solution of the
leave-one-out problem. Define the composite point: agent $i$ at
$c_i = \dbar_i$ with all its credit variables zero; every other agent
exactly as in $(x', c', s')$. Per-agent balance holds for $i$ (no trades)
and for each $j \ne i$ (the terms involving $i$ are zero and the rest
coincide with the leave-one-out constraint). The total is
$\dbar_i + (D - \dbar_i) = D$. The bounds are inherited. For every bus,
agent $i$ contributes $c_i - \dbar_i = 0$ to the injection deviation, so
every flow expression equals its leave-one-out counterpart and every line
constraint holds; no property of $\Phi$ is used. Hence the composite
point is feasible for the full problem, and
\[
  W_{\mathrm{CCM}} \;\ge\; v_i (L_i - \dbar_i) + \sum_{j \ne i} v_j s'_j
  \;=\; v_i (L_i - \dbar_i) + W^{(-i)} .
\]
Rearranging and applying selection independence gives
$S_i = W_{\mathrm{CCM}} - W^{(-i)} \ge v_i (L_i - \dbar_i)$. For
$\card{\NN} = 1$ the total balance forces $c_1 = \dbar_1$ and the
inequality holds with equality. \qed

\begin{proof}[Proof of Proposition~\ref{prop:milp}]
By Lemma~\ref{lem:reduced}, the CCM clearing problem reduces to
a feasible LP with finite optimal value in
$x = (x_{ij})_{i \ne j}$ and with only
inequality constraints~\eqref{eq:reduced:lb}--\eqref{eq:reduced:nonneg}.
KKT conditions are therefore necessary and sufficient for
optimality~\cite{bertsimas_introduction_1997}.  The KKT system comprises
primal feasibility, dual feasibility, a stationarity equation
for each ordered pair $(i,j)$, and one complementarity pair per
inequality constraint (the full system is given in
\ref{app:kkt}).  Replacing each complementarity pair
$0 \le a \perp b \ge 0$ with the Fortuny-Amat and McCarl
linearization ($a \le M_a\,z$, $b \le M_b\,(1-z)$,
$z \in \{0,1\}$)~\cite{fortuny-amat_representation_1981} produces an exact
linearization of the bounded KKT representation under
Assumption~\ref{ass:bigm}.  One binary
variable per complementarity pair yields the
count~\eqref{eq:bincount}.
\end{proof}

\section{Additional Numerical Results}\label{app:additional-results}

\subsection{Fixed-Case Allocation and Participant Outcomes}
\label{app:fixed-outcomes}

Tables~\ref{tab:peragent} and~\ref{tab:case24-peragent} report the
participant-level outcomes summarized in Section~\ref{sec:welfare-results}.

\begin{table}[ht]
\centering
\caption{Per-agent curtailment and surplus on the 3-bus network.
  $S_{\mathrm{CCM}}$ equals each agent's gross value of served
  load plus net VCG transfers received, so the column sum
  exceeds $W_{\mathrm{CCM}}$ whenever VCG runs a deficit.}
\label{tab:peragent}
\begin{tabular}{lcccccc}
\toprule
Agent & $v_i$ & $\dbar_i$ & $c_i^{\mathrm{CCM}}$
  & $S_{\mathrm{admin}}$ & $S_{\mathrm{CCM}}$
  & Gain \\
  & (\$/MWh) & (MW) & (MW)
  & (\$/event) & (\$/event) & (\$/event) \\
\midrule
DC-1  & 50{,}000 & 50 &  0 & 7.50M & 9.75M & 2.25M \\
DC-2  & 20{,}000 & 50 &  0 & 2.00M & 2.75M & 0.75M \\
C\&I  &  5{,}000 & 30 & 50 & 0.35M & 0.46M & 0.11M \\
Res   &  1{,}500 & 20 & 80 & 0.09M & 0.36M & 0.27M \\
VPP-1 &  1{,}000 & 40 & 50 & 0.02M & 0.06M & 0.04M \\
VPP-2 &    500   & 40 & 50 & 0.01M & 0.01M & 0.01M \\
\bottomrule
\end{tabular}
\end{table}

\subsection{Additional Network and Participant Detail}

Table~\ref{tab:case24-peragent} reports the IEEE 24-bus participant outcomes
referenced in Section~\ref{sec:properties}. Both data centers receive full
curtailment relief, lower-VOLL agents accept additional obligations, and every
agent gains relative to the administrative allocation under VCG benchmark
transfers. Proposition~\ref{prop:ir} guarantees the weak inequality on any
network; the strict gains are specific to this instance.

\begin{table}[H]
\centering
\caption{Per-agent curtailment and surplus on the stylized IEEE 24-bus
network. $S_{\mathrm{CCM}}$ equals each agent's gross value of served load plus
net VCG transfers received, so the column sum exceeds $W_{\mathrm{CCM}}$
whenever VCG runs a deficit.}
\label{tab:case24-peragent}
\resizebox{\textwidth}{!}{\begin{tabular}{l r r r r r r}
\toprule
Agent & $\bar{d}_i$ (MW) & $c_\text{admin}$ (MW) & $c_\text{CCM}$ (MW) & $S_\text{admin}$ (\$/event) & $S_\text{CCM}$ (\$/event) & Gain (\$/event) \\
\midrule
DC-1 & 105 & 105.0 & 0.0 & 4,750,000 & 9,459,705 & 4,709,705 \\
DC-2 & 88 & 88.0 & 0.0 & 1,240,000 & 2,510,705 & 1,270,705 \\
C\&I-1 & 88 & 88.0 & 78.1 & 460,000 & 477,973 & 17,973 \\
C\&I-2 & 70 & 70.0 & 170.0 & 300,000 & 461,116 & 161,116 \\
Res-1 & 53 & 53.0 & 53.9 & 70,500 & 73,705 & 3,205 \\
Res-2 & 35 & 35.0 & 80.0 & 45,000 & 347,504 & 302,504 \\
VPP-1 & 35 & 35.0 & 60.0 & 20,000 & 223,514 & 203,514 \\
VPP-2 & 18 & 18.0 & 50.0 & 16,000 & 273,352 & 257,352 \\
\midrule
Total & & & & 6,901,500 & 13,827,573 & 6,926,073 \\
\bottomrule
\end{tabular}
}
\end{table}

Table~\ref{tab:ptdf-matrix} gives the $11 \times 5$ interface PTDF matrix
referenced in Section~\ref{sec:edit-nygrid-testbed}. Buses~41, 47, and~55
influence the upstate interfaces IF\_1--IF\_8, while bus~78 influences
IF\_9 and bus~80 influences IF\_9 and IF\_11. IF\_10 has zero sensitivity
to all five agent buses and therefore does not constrain credit trades.

\begin{table}[H]
\centering
\caption{Interface PTDF matrix for the five CCM agent buses. Entry
$\Phi_{\ell,i}$ gives the fraction of a 1~MW injection at bus~$i$ that flows
through interface~$\ell$ in its positive (southbound) direction.}
\label{tab:ptdf-matrix}
\begin{tabular}{lccccc}
\toprule
 & Bus 41 & Bus 47 & Bus 55 & Bus 78 & Bus 80 \\
 & (Zone F) & (Zone E) & (Zone A) & (Zone I) & (Zone K) \\
\midrule
IF\_1   &  0.02 &  0.13 &  0.92 &  0.0 &  0.0 \\
IF\_2   &  0.02 &  0.13 &  0.92 &  0.0 &  0.0 \\
IF\_3   & $-$0.16 &  0.81 &  0.76 &  0.0 &  0.0 \\
IF\_4   & $-$0.02 & $-$0.13 &  0.08 &  0.0 &  0.0 \\
IF\_5   & $-$0.62 &  0.36 &  0.33 &  0.0 &  0.0 \\
IF\_6   &  0.45 &  0.45 &  0.43 &  0.0 &  0.0 \\
IF\_7   &  0.32 &  0.30 &  0.27 &  0.0 &  0.0 \\
IF\_8   &  1.00 &  1.00 &  1.00 &  0.0 &  0.0 \\
IF\_9   &  0.0 &  0.0 &  0.0 & $-$1.00 & $-$1.00 \\
IF\_10  &  0.0 &  0.0 &  0.0 &  0.0 &  0.0 \\
IF\_11  &  0.0 &  0.0 &  0.0 &  0.0 & $-$1.00 \\
\bottomrule
\end{tabular}
\end{table}

\FloatBarrier
\subsection{Additional Randomized Results}
\label{app:randomized}

Figure~\ref{fig:wp1-paired} reports the paired value-family comparisons
summarized in Section~\ref{sec:robustness}.

\begin{figure}[htbp]
\centering
\includegraphics[width=\textwidth]{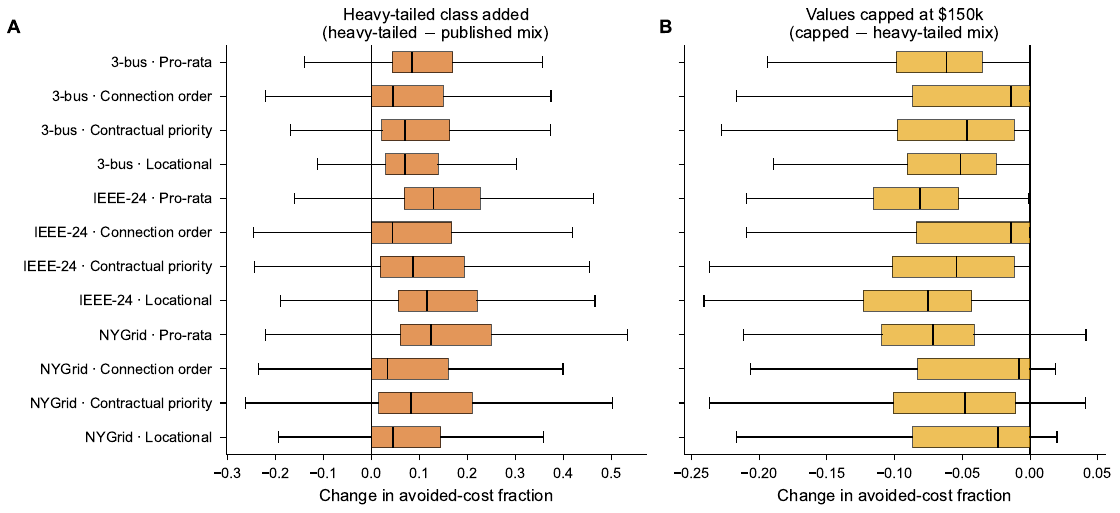}
\caption{Adding the high-value small commercial and industrial class raises
  the gains, and capping its value reduces them. Paired differences in
  avoided interruption cost. (A) F3, which adds the high-value class, minus
  the survey-based F1 construction. (B) F4, which caps declared values at
  150,000 \$/MWh, minus F3. Each comparison uses the same scenario realization
  and random seed, which removes independent sampling
  variation. Panel A measures the combined effect of adding the
  high-value class. Panel B isolates the value cap because only the
  declared values change. Boxes show the
  interquartile range and median. Whiskers show 1.5 times the
  interquartile range. Outliers are omitted. $n = 2{,}700$ paired draws per
  row.}
\label{fig:wp1-paired}
\end{figure}

\FloatBarrier

Across all 12 network-and-rule combinations, F3 adds the high-value class and
avoids more interruption cost than F1. Capping declared values at
150,000~\$/MWh in F4 reduces the avoided cost in every combination. The results
show that customer values materially affect the gains from reallocating
curtailment.

\FloatBarrier
\subsection{Multi-Hour Case Details}
\label{app:multihour}

This section reports the hourly inputs and detailed results behind the
multi-hour analysis. It shows how customer capability, episode duration,
duration-dependent interruption costs, and available network capacity affect
the allocation and total interruption cost.

\begin{figure}[htbp]
\centering
\includegraphics[width=0.82\textwidth]{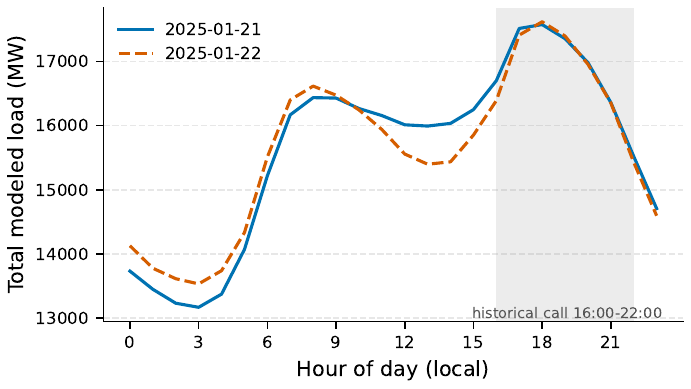}
\caption{The two 24-hour modeled NYGrid load profiles used in the case study,
  January 21 and 22, 2025. Each NYISO zonal series from report P-58B is divided
  by its daily mean and scaled to the corresponding NYGrid proxy-bus baseline.
  This normalization preserves the observed daily shape while retaining the
  reduced network's load scale. The vertical axis is the sum of the five scaled
  proxy-bus loads in MW. The shaded band marks the historical
  16:00--22:00 mandatory demand-response call on both days.}
\label{fig:wp3-profiles}
\end{figure}

\FloatBarrier

Figure~\ref{fig:wp3-profiles} shows the two profiles, which give each hourly
clearing a different set of customer loads.
The shaded 16:00--22:00 window defines the six-hour episode tested on both
days. We use these measured daily shapes as fixed inputs to compare the two
days. Estimating how often scarcity will occur requires a separate empirical
model.

\begin{table}[ht]
\centering
\footnotesize
\caption{Reduction in total interruption cost relative to proportional
  assignment (percent), by capability cap and episode duration, under
  the flat cost shape at the central 3.795 percent requirement. Each
  cell reports January 21 / January 22, 2025.}
\label{tab:wp3-core}
\begin{tabular}{lrrrrrr}
\toprule
Capability & 1 h & 2 h & 3 h & 4 h & 5 h & 6 h \\
\midrule
$\kappa_C=1$ (no trade) & 0.0 / 0.0 & 0.0 / 0.0 & 0.0 / 0.0 & 0.0 / 0.0 & 0.0 / 0.0 & 0.0 / 0.0 \\
$\kappa_C=2$ & 80.8 / 82.0 & 80.9 / 81.8 & 81.0 / 81.5 & 81.1 / 81.3 & 81.2 / 81.3 & 81.3 / 81.3 \\
$\kappa_C=4$ & 95.3 / 95.3 & 95.2 / 95.3 & 95.2 / 95.2 & 95.2 / 95.2 & 95.2 / 95.2 & 95.2 / 95.2 \\
full-load bound & 97.9 / 97.9 & 97.9 / 97.9 & 97.9 / 97.9 & 97.9 / 97.9 & 97.9 / 97.9 & 97.9 / 97.9 \\
\bottomrule
\end{tabular}

\end{table}

\FloatBarrier

Table~\ref{tab:wp3-core} shows that customer capability determines most of the
cost reduction in both daily profiles. The reduction is zero at $\kappa_C=1$
because customers cannot trade beyond their initial obligations. It rises to
about 81 percent at $\kappa_C=2$, 95 percent at $\kappa_C=4$, and 98 percent under
the full-load bound. Extending the episode from one to six hours changes the
reduction by less than one percentage point at each capability level.

\begin{figure}[htbp]
\centering
\includegraphics[width=0.82\textwidth]{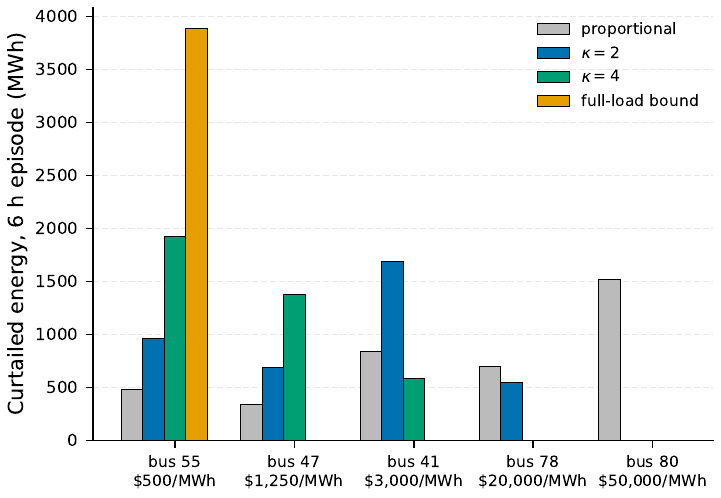}
\caption{NYGrid curtailed energy per customer over the six-hour episode on
  January 21, 2025, by capability cap. Tightening the cap spreads the
  requirement across customers. Under the original full-load bound the
  single lowest-value customer absorbs the entire 3,889 MWh alone.}
\label{fig:wp3-exposure}
\end{figure}

\begin{table}[ht]
\centering
\footnotesize
\caption{NYGrid obligations and cleared curtailment by bus over the six-hour
  episode on January 21, 2025, by capability cap. Columns report starting
  obligation $\bar d$, cleared curtailment $c_{\mathrm{CCM}}$, the traded
  change, spare capability, and the maximum consecutive hours curtailed.}
\label{tab:wp3-capability}
\begin{tabular}{llrrrrr}
\toprule
Capability & Bus & $\bar{d}$ (MWh) & $c_\text{CCM}$ (MWh) & Change & Spare (MWh) & Max consec. h \\
\midrule
$\kappa_C=1$ & 41 & 844 & 844 & +0 & 0 & 6 \\
$\kappa_C=1$ & 47 & 345 & 345 & +0 & 0 & 6 \\
$\kappa_C=1$ & 55 & 480 & 480 & +0 & 0 & 6 \\
$\kappa_C=1$ & 78 & 698 & 698 & +0 & 0 & 6 \\
$\kappa_C=1$ & 80 & 1,522 & 1,522 & +0 & 0 & 6 \\
$\kappa_C=2$ & 41 & 844 & 1,688 & +844 & 0 & 6 \\
$\kappa_C=2$ & 47 & 345 & 690 & +345 & 0 & 6 \\
$\kappa_C=2$ & 55 & 480 & 961 & +480 & 0 & 6 \\
$\kappa_C=2$ & 78 & 698 & 551 & $-$147 & 846 & 6 \\
$\kappa_C=2$ & 80 & 1,522 & 0 & $-$1,522 & 3,043 & 0 \\
$\kappa_C=4$ & 41 & 844 & 588 & $-$256 & 2,787 & 6 \\
$\kappa_C=4$ & 47 & 345 & 1,379 & +1,034 & 0 & 6 \\
$\kappa_C=4$ & 55 & 480 & 1,922 & +1,441 & 0 & 6 \\
$\kappa_C=4$ & 78 & 698 & 0 & $-$698 & 2,794 & 0 \\
$\kappa_C=4$ & 80 & 1,522 & 0 & $-$1,522 & 6,087 & 0 \\
full-load bound & 41 & 844 & 0 & $-$844 & 22,233 & 0 \\
full-load bound & 47 & 345 & 0 & $-$345 & 9,085 & 0 \\
full-load bound & 55 & 480 & 3,889 & +3,409 & 8,770 & 6 \\
full-load bound & 78 & 698 & 0 & $-$698 & 18,403 & 0 \\
full-load bound & 80 & 1,522 & 0 & $-$1,522 & 40,096 & 0 \\
\bottomrule
\end{tabular}

\end{table}

\FloatBarrier

Figure~\ref{fig:wp3-exposure} and Table~\ref{tab:wp3-capability} show who
curtails as capability expands. At $\kappa_C=2$, three customers accept more than
their initial obligations and the customer at bus~80 avoids all curtailment.
At $\kappa_C=4$, the customers at buses~47 and~55 absorb most of the requirement.
Under the full-load bound, the customer at bus~55 absorbs all 3,889~MWh of
curtailed energy. The largest cost reduction therefore relies on concentrating
the full six-hour requirement on one customer.

\begin{table}[ht]
\centering
\footnotesize
\caption{Reduction in total interruption cost (percent) under the three
  declared duration-cost shapes, by capability cap and episode duration,
  on January 21, 2025. The shapes shift the level of modeled cost and leave
  the allocation unchanged.}
\label{tab:wp3-shapes}
\begin{tabular}{llrrrrrr}
\toprule
Cost shape & Capability & 1 h & 2 h & 3 h & 4 h & 5 h & 6 h \\
\midrule
flat & $\kappa_C=2$ & 80.83 & 80.94 & 81.03 & 81.06 & 81.18 & 81.31 \\
flat & full-load bound & 97.92 & 97.92 & 97.92 & 97.92 & 97.92 & 97.91 \\
decreasing (ICE-shaped) & $\kappa_C=2$ & 80.83 & 80.91 & 80.97 & 81.01 & 81.09 & 81.19 \\
decreasing (ICE-shaped) & full-load bound & 97.92 & 97.92 & 97.92 & 97.92 & 97.92 & 97.92 \\
increasing (constructed mirror) & $\kappa_C=2$ & 80.83 & 80.96 & 81.05 & 81.09 & 81.21 & 81.35 \\
increasing (constructed mirror) & full-load bound & 97.92 & 97.92 & 97.93 & 97.92 & 97.92 & 97.91 \\
\bottomrule
\end{tabular}

\end{table}

\FloatBarrier

Table~\ref{tab:wp3-shapes} shows that the three duration-cost shapes leave the
customer allocation unchanged. At $\kappa_C=2$, the six-hour cost reduction
ranges from 81.19 to 81.35 percent across the shapes. Under the full-load
bound, it remains about 97.9 percent. Within the tested weights, customer
capability has a much larger effect than the duration-cost shape.

\begin{figure}[htbp]
\centering
\includegraphics[width=0.82\textwidth]{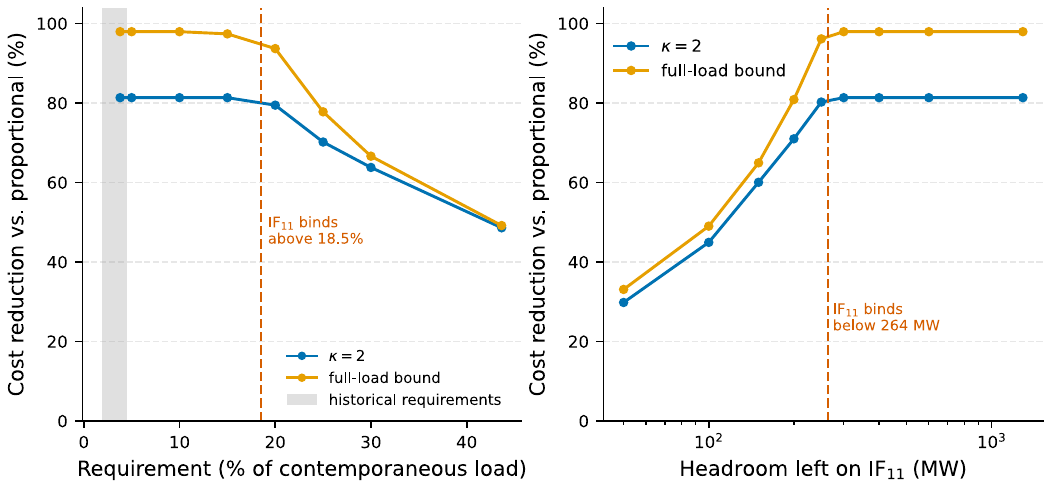}
\caption{Two routes to a binding interface. The left panel sweeps the required
  curtailment from 5 percent of hourly system load to the fixed-case
  depth. The interface IF\_11 first binds at 18.5 percent. The right panel
  sweeps the corridor capacity left for the market at the historical
  requirement. The analytically derived binding point is 264 MW of remaining
  capacity. Tested capacities below that threshold change the allocation.}
\label{fig:wp3-network}
\end{figure}

\FloatBarrier

Figure~\ref{fig:wp3-network} shows when the network constraint changes the
allocation. The central 3.795 percent requirement lies below the 18.5 percent
level at which IF\_11 first binds. At that requirement, IF\_11 changes the
allocation only when its remaining capacity falls below 264~MW. Network
capacity therefore affects the result under deeper curtailment or tighter
remaining headroom. These thresholds describe the incremental flow caused by
CCM reallocation and exclude the loading created by upstream generation
dispatch.

\begin{table}[htbp]
\centering
\footnotesize
\caption{Load-processing sensitivity against the central observed-load case.
  The reported-response-adjusted measured-load estimate adds the response
  credited to customer-baseline-reporting Special Case Resources back into
  the measured loads from 16:00 through 21:00. This bounded sensitivity
  evaluates one load-processing choice. Reconstructing untouched demand would
  require additional data.}
\label{tab:wp3-robustness}
\begin{tabular}{lrrr}
\toprule
Capability & Duration & Reduction (\%) & vs.\ central (pp) \\
\midrule
$\kappa_C=2$ & 1 & 82.05 & +1.22 \\
full-load bound & 1 & 97.89 & -0.04 \\
$\kappa_C=2$ & 2 & 82.10 & +1.16 \\
full-load bound & 2 & 97.89 & -0.03 \\
$\kappa_C=2$ & 3 & 82.18 & +1.15 \\
full-load bound & 3 & 97.89 & -0.03 \\
$\kappa_C=2$ & 4 & 82.23 & +1.17 \\
full-load bound & 4 & 97.89 & -0.03 \\
$\kappa_C=2$ & 5 & 82.35 & +1.17 \\
full-load bound & 5 & 97.89 & -0.03 \\
$\kappa_C=2$ & 6 & 82.50 & +1.18 \\
full-load bound & 6 & 97.88 & -0.03 \\
\bottomrule
\end{tabular}

\end{table}

\FloatBarrier

Table~\ref{tab:wp3-robustness} shows that this load adjustment leaves the main
comparison unchanged. At $\kappa_C=2$, the estimated cost reduction increases by
1.15 to 1.22 percentage points across episode lengths. Under the full-load
bound, it changes by no more than 0.04 percentage points. Capability remains
the main driver of the cost reduction under either load-processing choice.

\FloatBarrier
\subsection{Detailed Incentive and Settlement Results}
\label{app:settlement}

Table~\ref{tab:incentive} reports the participant-level deviation results, and
Table~\ref{tab:wp2-subsample} reports the network-level randomized settlement
summary used in Sections~\ref{sec:incentive-results} and
\ref{sec:payment-rules}.

\begin{table}[ht]
\centering
\footnotesize
\caption{Agent-level incentive metrics under VCG, uniform-price (UP), and
  locational (LOC) settlement on the 3-bus, IEEE 24-bus, and NYGrid networks.
  Each
  $g_{\mathrm{rule}}$ is the percentage surplus gain from the best unilateral
  deviation on the network's primary bid grid relative to truthful bidding.
  $b^{*}_{\mathrm{LOC}}/v$ is the corresponding bid multiple when LOC admits a
  profitable deviation; a dash means that no profitable LOC deviation was
  observed. The separate 100-point NYGrid sensitivity raises the Rotterdam
  uniform-price and locational gains from 162.4\% to 173.0\%.}
\label{tab:incentive}
\resizebox{\textwidth}{!}{\begin{tabular}{llrrrrrr}
\toprule
Network & Agent & Bus & $v$ (\$/MWh) & $g_{\mathrm{VCG}}$ (\%) & $g_{\mathrm{UP}}$ (\%) & $g_{\mathrm{LOC}}$ (\%) & $b^{*}_{\mathrm{LOC}}/v$ \\
\midrule
3-bus & DC-1 & 1 & 50,000 & 0.0 & 0.0 & 0.0 & -- \\
 & DC-2 & 1 & 20,000 & 0.0 & 0.0 & 1.8 & 0.26 \\
 & C\&I & 2 & 5,000 & 0.0 & 29.6 & 29.2 & 2.02 \\
 & Res & 2 & 1,500 & 0.0 & 48.5 & 0.0 & -- \\
 & VPP-1 & 3 & 1,000 & 0.0 & 197.1 & 197.1 & 4.94 \\
 & VPP-2 & 3 & 500 & 0.0 & 0.0 & 0.0 & -- \\
\midrule
IEEE 24-bus & DC-1 & 22 & 50,000 & 0.0 & 0.0 & 0.0 & -- \\
 & DC-2 & 21 & 20,000 & 0.0 & 0.0 & 0.0 & -- \\
 & C\&I-1 & 3 & 5,000 & 0.0 & 14.6 & 3.2 & 0.71 \\
 & C\&I-2 & 8 & 3,000 & 0.0 & 0.0 & 0.0 & -- \\
 & Res-1 & 9 & 1,500 & 0.0 & 1.9 & 4.4 & 3.27 \\
 & Res-2 & 10 & 1,000 & 0.0 & 0.0 & 0.0 & -- \\
 & VPP-1 & 16 & 800 & 0.0 & 0.0 & 0.0 & -- \\
 & VPP-2 & 23 & 500 & 0.0 & 0.0 & 0.0 & -- \\
\midrule
NYGrid & Rotterdam & 41 & 3,000 & 0.0 & 162.4 & 162.4 & 6.21 \\
 & Moses W & 47 & 1,250 & 0.0 & 0.0 & 0.0 & -- \\
 & Niagara E & 55 & 500 & 0.0 & 0.0 & 0.0 & -- \\
 & CE UG & 78 & 20,000 & 0.0 & 0.0 & 0.0 & -- \\
 & Northport & 80 & 50,000 & 0.0 & 0.0 & 34.4 & 0.10 \\
\bottomrule
\end{tabular}
}
\end{table}

\begin{table}[ht]
\centering
\footnotesize
\caption{Settlement over the paired subsample of the Monte Carlo design of
  Section~\ref{sec:robustness}, with 500 scenario realizations per testbed, three value
  families, four obligation rules. Deficit and surplus are counted at a
  $10^{-6}$ relative tolerance. A degenerate primal basis is necessary but
  not sufficient for the credit prices to be non-unique, so the last row
  provides an upper bound on that incidence
  (\ref{app:tiebreak}).}
\label{tab:wp2-subsample}
\begin{tabular}{lrrr}
\toprule
& 3-bus & IEEE 24-bus & NYGrid \\
\midrule
Rule-level cases & 6,000 & 6,000 & 6,000 \\
Cases with a binding constraint (\%) & 6.9 & 10.5 & 26.2 \\
Cases with a VCG deficit (\%) & 36.7 & 33.8 & 42.2 \\
\quad median deficit (\% of $W_{\mathrm{CCM}}$) & 4.3 & 2.9 & 2.6 \\
\quad 95th percentile (\% of $W_{\mathrm{CCM}}$) & 20.9 & 18.4 & 19.0 \\
\quad of those, with no binding constraint (\%) & 85.0 & 81.6 & 67.1 \\
Cases with a VCG surplus (\%) & 1.4 & 4.2 & 8.8 \\
Median congestion rent, congested cases (\% of $W_{\mathrm{CCM}}$) & 7.7 & 4.9 & 6.1 \\
Revenue-adequate cases (\%) & 100.00 & 100.00 & 100.00 \\
Cases where every agent beats no trade (\%) & 100.00 & 100.00 & 100.00 \\
Cases with a degenerate primal basis (\%) & 15.4 & 13.5 & 18.4 \\
\bottomrule
\end{tabular}

\end{table}

\FloatBarrier

\section{Karush--Kuhn--Tucker System and MILP Construction}\label{app:kkt}
 
This appendix provides the full KKT system of the reduced
LP~\eqref{eq:reduced}, the Fortuny-Amat and McCarl
linearization that produces the MILP in
Proposition~\ref{prop:milp}, the binary variable count
derivation, and the Big-M sufficient conditions referenced by
Assumption~\ref{ass:bigm}.
 
\subsection{KKT Stationarity Derivation}\label{app:stationarity}
 
The reduced LP (Lemma~\ref{lem:reduced}) maximizes
$\sum_{i \in \NN} b_i\bigl(L_i - c_i(x)\bigr)$ over
$x = (x_{ij})_{i \ne j}$ subject to the inequality
constraints~\eqref{eq:reduced:lb}--\eqref{eq:reduced:nonneg}.
We associate non-negative dual multipliers with each constraint:
$\alpha_i \ge 0$ with the lower bound $c_i(x) \ge 0$,
$\beta_i \ge 0$ with the upper bound $c_i(x) \le L_i$,
$\mu^+_\ell \ge 0$ with the upper PTDF limit
$\sum_n \Phi_{\ell n}\,\Delta P_n \le \Fbar_\ell^+$,
$\mu^-_\ell \ge 0$ with the lower PTDF limit
$-\sum_n \Phi_{\ell n}\,\Delta P_n \le \Fbar_\ell^-$,
and $\gamma_{ij} \ge 0$ with the nonnegativity constraint
$x_{ij} \ge 0$.

The stationarity condition requires
$\partial \mathfrak{L} / \partial x_{ij} = 0$ for every ordered
pair $(i,j)$ with $i \ne j$, where $\mathfrak{L}$ is the
Lagrangian.  A unit increase in $x_{ij}$ raises $c_i$ by one
and lowers $c_j$ by one (from the definition
$c_k(x) = \dbar_k - \sum_{m \ne k} x_{mk}
+ \sum_{m \ne k} x_{km}$).  This change affects the objective,
the agent bounds on both~$i$ and~$j$, and every PTDF constraint
through the net injection change
$\Delta P_n = \sum_{k:\,n(k)=n}(c_k - \dbar_k)$.
 
Differentiating the Lagrangian with respect to $x_{ij}$ and
collecting terms yields the stationarity equation:
\begin{equation}\label{eq:app:stationarity}
0 = (b_i - b_j) - (\alpha_i - \alpha_j) + (\beta_i - \beta_j) + \sum_{\ell \in \LL}(\mu^+_\ell - \mu^-_\ell)\Phi_{\ell,n(i)} - \sum_{\ell \in \LL}(\mu^+_\ell - \mu^-_\ell)\Phi_{\ell,n(j)} - \gamma_{ij}.
\end{equation}
The five terms have the following origin.  The term $(b_i - b_j)$
captures the marginal objective effect, where shifting one MW of
curtailment from agent~$j$ to agent~$i$ reduces~$j$'s curtailment
(raising welfare by~$b_j$) and raises~$i$'s curtailment (lowering
welfare by~$b_i$), and the LP maximizes so the gradient sign
is $(b_i - b_j)$ after accounting for the max-to-min conversion
in the KKT conditions.
The term $-(\alpha_i - \alpha_j)$ accounts for the agent lower
bounds: increasing $c_i$ relaxes the lower bound at agent~$i$, while
decreasing $c_j$ tightens it at agent~$j$.
The term $+(\beta_i - \beta_j)$ accounts for the upper bounds:
increasing $c_i$ tightens the upper bound at agent~$i$, while decreasing
$c_j$ relaxes it at agent~$j$.
The PTDF term reflects the marginal congestion cost, where the net
injection change at bus~$n(i)$ increases by one MW and at
bus~$n(j)$ decreases by one MW, and the PTDF coefficients
$\Phi_{\ell,n(i)} - \Phi_{\ell,n(j)}$ translate these bus-level
changes into line-flow changes.
The multiplier $\gamma_{ij}$ enforces the nonnegativity of
$x_{ij}$.

\subsection{Full Complementarity System}\label{app:complementarity}
 
The KKT system of the reduced LP comprises primal feasibility
\eqref{eq:reduced:lb}--\eqref{eq:reduced:nonneg}, dual
feasibility $\alpha, \beta, \mu^+, \mu^-, \gamma \ge 0$,
the stationarity equation~\eqref{eq:app:stationarity} for each
ordered pair $(i,j)$ with $i \ne j$, and the following
complementarity conditions:
\begin{subequations}\label{eq:app:comp}
\begin{align}
  0 \le \alpha_i \;\perp\; c_i(x)
    &\ge 0
    &&\forall\, i \in \NN,
    \label{eq:app:comp:lb}\\[3pt]
  0 \le \beta_i \;\perp\; \bigl(L_i - c_i(x)\bigr)
    &\ge 0
    &&\forall\, i \in \NN,
    \label{eq:app:comp:ub}\\[3pt]
  0 \le \mu^+_\ell \;\perp\;
    \Bigl(\Fbar_\ell^+
      - \textstyle\sum_n \Phi_{\ell n}\,\Delta P_n\Bigr)
    &\ge 0
    &&\forall\, \ell \in \LL,
    \label{eq:app:comp:ptdfp}\\[3pt]
  0 \le \mu^-_\ell \;\perp\;
    \Bigl(\Fbar_\ell^-
      + \textstyle\sum_n \Phi_{\ell n}\,\Delta P_n\Bigr)
    &\ge 0
    &&\forall\, \ell \in \LL,
    \label{eq:app:comp:ptdfm}\\[3pt]
  0 \le \gamma_{ij} \;\perp\; x_{ij}
    &\ge 0
    &&\forall\, i \ne j.
    \label{eq:app:comp:nonneg}
\end{align}
\end{subequations}
Each condition enforces that at least one of the two
non-negative quantities equals zero.  Since the reduced LP is
feasible and has finite optimal value, these KKT conditions are
both necessary and sufficient for
optimality~\cite{bertsimas_introduction_1997}.
 
\subsection{Fortuny-Amat and McCarl Linearization}%
\label{app:linearization}
 
Each complementarity pair $0 \le a \perp b \ge 0$ encodes the
nonlinear constraint $a \cdot b = 0$ together with $a, b \ge 0$.
The Fortuny-Amat and McCarl construction~\cite{fortuny-amat_representation_1981}
replaces each pair with a binary variable $z \in \{0,1\}$ and
the linear constraints
\begin{equation}\label{eq:app:bigm}
  a \;\le\; M_a\, z,
  \qquad
  b \;\le\; M_b\,(1 - z),
  \qquad
  z \in \{0,1\},
\end{equation}
where $M_a$ and $M_b$ are finite upper bounds on $a$ and $b$ at
every KKT point.  When $z = 0$, the first constraint forces
$a = 0$ and the second allows $b \le M_b$; when $z = 1$, the
roles reverse.  This construction is exact for the bounded KKT
representation whenever $M_a$ and $M_b$ satisfy
Assumption~\ref{ass:bigm}.
 
Applying~\eqref{eq:app:bigm} to every complementarity pair
in~\eqref{eq:app:comp} replaces the nonlinear KKT system with a
system of linear equalities and inequalities plus binary variables.
Combined with the linear stationarity
equations~\eqref{eq:app:stationarity} and the primal/dual
feasibility constraints, this yields a single-level MILP that is
allocation-equivalent to the original CCM clearing LP. When problem-specific Big-M bounds are unavailable,
parameter-free alternatives such as SOS1
constraints~\cite{kleinert_why_2023} or solver-native indicator
constraints can replace the Fortuny-Amat and McCarl linearization;
because the bounds derived in
\ref{app:bigm} exploit the CCM problem structure, whether SOS1 reformulations offer computational
advantages at larger network scales remains a direction for future work.

\subsection{Binary Variable Count}\label{app:bincount}
 
One binary variable per complementarity pair produces the count
stated in Proposition~\ref{prop:milp}:
\begin{equation}\label{eq:app:bincount}
  \underbrace{\card{\NN}}_{\eqref{eq:app:comp:lb}}
  \;+\;
  \underbrace{\card{\NN}}_{\eqref{eq:app:comp:ub}}
  \;+\;
  \underbrace{\card{\LL}}_{\eqref{eq:app:comp:ptdfp}}
  \;+\;
  \underbrace{\card{\LL}}_{\eqref{eq:app:comp:ptdfm}}
  \;+\;
  \underbrace{\card{\NN}(\card{\NN}-1)}_{\eqref{eq:app:comp:nonneg}}.
\end{equation}
The first two groups correspond to agent lower and upper bounds
($\card{\NN}$ each).  The next two correspond to the upper and lower
PTDF limits ($\card{\LL}$ each, totaling $2\card{\LL}$).
The last group corresponds to credit-flow nonnegativity
constraints, one for each ordered pair $(i,j)$ with $i \ne j$,
which yields $\card{\NN}(\card{\NN} - 1)$ pairs.
 
On the 3-bus network ($\card{\NN} = 6$, $\card{\LL} = 3$), this formula
gives $6 + 6 + 3 + 3 + 30 = 48$ binaries. On the IEEE 24-bus network
($\card{\NN} = 8$, $\card{\LL} = 38$), the count is
$8 + 8 + 38 + 38 + 56 = 148$.  The credit-nonnegativity binaries
dominate the count and grow quadratically in $\card{\NN}$.

\subsection{Big-M Sufficient Conditions and Heuristic Bounds}%
\label{app:bigm}
 
Assumption~\ref{ass:bigm} holds whenever the dual optimal set of
the reduced LP is bounded.  LP strong duality guarantees that the
dual is feasible and attains the same finite optimal value as the
primal (Lemma~\ref{lem:wellposed}), so at least one dual optimum
exists.  However, the dual optimal set need not be bounded. When
the primal LP is degenerate (multiple dual optima exist), the set
of all dual optima can form an unbounded face of the dual
polyhedron.  In such cases, Assumption~\ref{ass:bigm} can be
enforced by selecting Big-M values from a specific (finite) dual
optimum, or by adding a regularization term (e.g.,
$\epsilon \sum \alpha_i^2$) to obtain a unique bounded dual
solution.
 
For non-degenerate instances, the following heuristic bounds guide
implementation.
 
\textit{Primal bounds.}
Agent curtailment satisfies $0 \le c_i \le L_i$ by constraint, and served load satisfies $0 \le s_i \le L_i$. The raw credit-flow variables, however, are not bounded by flow conservation alone, because adding a directed cycle or equal counterflows $x_{ij}$ and $x_{ji}$ leaves every $c_i$, $s_i$, and PTDF flow unchanged, so the uncompactified $x$-polyhedron is generally unbounded. For the KKT/MILP reformulation, we use the compactified accounting representation $0 \le x_{ij} \le D$, which is without loss for the allocation problem because every feasible curtailment vector admits a cycle-free transport representation with each bilateral flow bounded by $D$. We therefore use $M^x_{ij} = D$ as the primal Big-M bound for credit-flow variables; this compactification selects a bounded accounting representation without removing any feasible final curtailment allocation.
PTDF flow slack satisfies
$0 \le \Fbar_\ell^{+} - \sum_n \Phi_{\ell n}\,\Delta P_n
\le \Fbar_\ell^{+} + \Fbar_\ell^{-}$ and
$0 \le \Fbar_\ell^{-} + \sum_n \Phi_{\ell n}\,\Delta P_n
\le \Fbar_\ell^{+} + \Fbar_\ell^{-}$.  These bounds supply the $M_b$ constants for
every complementarity pair in which a primal quantity appears.
 
\textit{Agent capacity duals.}
The stationarity equation~\eqref{eq:app:stationarity} constrains
$\alpha_i - \alpha_j$ to equal the bid spread plus congestion
terms minus $\gamma_{ij}$.  In typical instances, this gives
$\alpha_i, \beta_i = O(\max_j b_j +
\max_\ell\{\mu_\ell^+,\mu_\ell^-\})$.
 
\textit{Congestion duals.}
From~\eqref{eq:app:stationarity}, when $x_{ij} > 0$ (so that
$\gamma_{ij} = 0$ by complementarity) and the PTDF difference
$\Phi_{\ell,n(i)} - \Phi_{\ell,n(j)} \ne 0$, one can bound the
net congestion shadow price
$\sum_\ell (\mu^+_\ell - \mu^-_\ell)
(\Phi_{\ell,n(i)} - \Phi_{\ell,n(j)})$ by the bid spread
$(b_i - b_j)$ plus the agent-bound duals.  This yields
$\mu^+_\ell, \mu^-_\ell = O(\max_j b_j / \min_{i \ne j}
|\Phi_{\ell,n(i)} - \Phi_{\ell,n(j)}|)$ when the minimum PTDF
difference is bounded away from zero.
 
\textit{Co-located agents (division-by-zero case).}
Co-located pairs do occur in the 3-bus network, where DC-1 and DC-2 share Bus 1, C\&I
and Res share Bus 2, and VPP-1 and VPP-2 share Bus 3. In the Big-M construction,
congestion-dual bounds should therefore not be inferred from co-located pairs,
for which $\Phi_{\ell,n(i)}-\Phi_{\ell,n(j)}=0$ for every line $\ell$. Instead,
bounds should be computed from non-co-located active pairs, from a solved dual
LP, or from the dual-regularization approach described above.

\textit{Fixed-case sensitivity.}
We multiplied every Big-M constant by 1, 10, and 100 on each of the three fixed
networks. The strategic objective remained stable to four significant figures
across these tests. This sensitivity supports the reported fixed-case solutions
and does not validate the bounds for other instances.


\section{Tie-Breaking Formalism for VCG}\label{app:tiebreak}
 
Proposition~\ref{prop:vcg} applies Nisan's Theorem~9.17. The theorem requires
the allocation rule to maximize welfare and use the same tie-breaking rule for
every bid report. This appendix explains why the CCM needs this requirement
and how it relates to the mathematical program with equilibrium constraints
(MPEC) literature.
 
\textit{Why tie-breaking matters.}
The CCM clearing LP can have several optimal allocations when the optimum lies
on a degenerate face of the feasible set. These allocations all produce the
same objective value $\sum_i b_i s_i^*$ and therefore the same welfare in
Corollary~\ref{cor:welfare}. VCG Clarke-pivot payments also depend on the
specific allocation $s_i^*$ that the clearing rule selects. A participant
could otherwise change its bid to steer the rule toward a welfare-equivalent
allocation with a more favorable VCG payment. Fixing the selection rule before
participants submit their bids prevents this form of manipulation.
 
\textit{Lexicographic vertex selection.}
Lexicographic vertex selection provides one rule that meets this requirement.
We first fix an ordering of the decision variables. The clearing rule then
selects the lexicographically smallest vertex among the optimal vertices of
the LP. The feasible vertices depend only on $\dbar, L, \Phi, \Fbar$, and the
same ordering applies to every bid vector~$b$. Standard LP post-processing can
compute this deterministic selection.
 
\textit{Connection to the MPEC literature.}
The MILP reformulation in Proposition~\ref{prop:milp} is
allocation-equivalent to the reduced LP. For each lower-level optimal
allocation admitted by the stated selection rule, the bounded accounting
representation retains at least one Karush--Kuhn--Tucker (KKT) certificate.
The representation may omit other raw credit-flow or dual descriptions of the
same LP optimum. The MILP solver returns one retained point, and its internal
rules determine which point it returns when several remain.

The MPEC literature~\cite{hobbs_linear_2001} often uses an optimistic bilevel
selection rule. Under this convention, the lower level selects the optimum
that is most favorable to the upper-level objective. The identical
lower-level objective value across LP optima leaves the allocation tie
unresolved by that objective alone. We resolve this tie for VCG with the fixed
lexicographic rule described above. The paper's welfare results hold under any
tie-breaking rule because they depend only on the objective value. The chosen
rule determines the allocation and the resulting VCG payments.

\textit{Allocation and price ties require separate rules.}
The lexicographic rule selects one optimal primal allocation. Dual prices
require an additional rule whenever several dual solutions support that
allocation. In those cases, the credit prices $\pi_n$ in
Equation~\eqref{eq:loc-price} are non-unique, which also makes the locational
settlement non-unique. Our numerical analysis reports the dual solution
returned by the solver without modification and does not average alternative
dual solutions. An averaged vector need not correspond to the optimal basis
that the solver selected.

We use two diagnostics to identify cases where non-uniqueness can arise. In
15.8\% of the tested design, a binding flow constraint has a zero multiplier
and the primal basis is degenerate. Degeneracy is necessary but insufficient
for multiple dual solutions, so 15.8\% is an upper bound on the incidence of
non-unique prices. In 91.7\% of the tested design, an agent reaches a
curtailment bound with zero reduced cost. This second condition identifies an
alternative optimal allocation and provides no measure of price
non-uniqueness. An implemented tariff must specify a fixed dual-selection
rule. The solver's returned basis records a computational choice, whereas a
market rule must state its selection procedure in advance. We identify this
requirement as an open design item in Section~\ref{sec:limitations}.

\bibliographystyle{elsarticle-num}
\bibliography{core_references,policy_references,benchmark_references}

\end{document}